\documentclass{article}
\usepackage[a4paper]{geometry}
\usepackage[utf8]{inputenc}
\usepackage[T1]{fontenc}                     
\usepackage[english]{babel}		   
\usepackage{amsmath,amssymb, amsthm} 
\usepackage[font=scriptsize]{caption}
\usepackage{epsfig}
\usepackage{psfrag}
\usepackage{scalerel}
\usepackage{mathrsfs}
\usepackage{url}
\usepackage{hyperref}
\usepackage{bm}%
\usepackage{cite}
\usepackage{color}
\usepackage{lmodern}
\usepackage{booktabs}
\usepackage{scalerel}
\usepackage{enumitem}
\usepackage{siunitx}
\usepackage{dsfont}
\usepackage{xfrac}
\usepackage{authblk}
\usepackage{comment}

\newcommand{\R}{\mathbb{R}}                  
\newcommand{\N}{\mathbb{N}}
\newcommand{\Z}{\mathbb{Z}}

\newcommand{\norm}[1]{\lVert#1\rVert}      

\newcommand{\Id}{I}

\def\A{\mathcal{A}}
\def\K{\mathcal{K}}

\def\wconv{\rightharpoonup}
\def\usquare{u^{\scaleto{\square}{3pt}}}
\def\sPi{\scaleto{\Pi}{4pt}}

\newcommand{\olda}{b}
\newcommand{\eort}{\mathsf{e}_\perp}
\newcommand{\er}{\mathsf{e}_r}
\newcommand{\calpha}{c_\alpha}

\def\Cquad{C_0}

\def\unitsphere{\mathbb{S}^2}

\def\minloop{u_I^*}
\def\minloopgen{u_I^*}

\def\npm{\mathsf{n}^\pm}

\def\ord_h{\mathfrak{o}_h}

\theoremstyle{plain}
\newtheorem{theorem}{Theorem}[section]
\newtheorem{conjecture}{Conjecture}[section]
\newtheorem{lemma}{Lemma}[section]
\newtheorem{proposition}{Proposition}[section]
\newtheorem{corollary}{Corollary}[section]

\theoremstyle{definition}
\newtheorem{definition}{Definition}[section]

\theoremstyle{remark}
\newtheorem{remark}{Remark}[section]

\numberwithin{equation}{section}

\title{\bf Symmetric constellations of satellites moving around a central body of large mass}

\author[1,2]{M. Fenucci\thanks{email: \texttt{fenucci@mail.dm.unipi.it}}}
\author[1]{G. F. Gronchi}

\affil[1]{Dipartimento di Matematica, Università di Pisa, Largo
  B. Pontecorvo 5, 56127 Pisa, Italy}
\affil[2]{Department of Astronomy, Faculty of Mathematics, University
  of Belgrade, Studentski trg 16, 11000 Belgrade, Serbia}

\begin{document}
\maketitle
\begin{abstract}
   We consider a $(1+N)$-body problem in which one particle has mass
   $m_0 \gg 1$ and the remaining $N$ have unitary mass.  We can
   assume that the body with larger mass (central body) is at rest
   at the origin, coinciding with the center of mass of the $N$ bodies
   with smaller masses (satellites).  The interaction force between
   two particles is defined through a potential of the form
   \[
      U \sim \frac{1}{r^\alpha}, 
   \]
   where $\alpha \in [1,2)$ and $r$ is the distance between the
     particles.  Imposing symmetry and topological constraints, we
     search for periodic orbits of this system by variational
     methods. Moreover, we use $\Gamma$-convergence theory to study
     the asymptotic behaviour of these orbits, as the mass of the
     central body increases.  It turns out that the Lagrangian action
     functional $\Gamma$-converges to the action functional of a
     Kepler problem, defined on a suitable set of loops. In some cases, minimizers of
     the $\Gamma$-limit problem can be easily found, and they are
     useful to understand the motion of the satellites for large
     values of $m_0$. We discuss some examples, where the symmetry is
     defined by an action of the groups $\Z_4$ , $\Z_2 \times \Z_2$
     and the rotation groups of Platonic polyhedra on the set of
     loops.


\end{abstract}
\tableofcontents
\section{Introduction}
In the last two decades several new periodic solutions of the
Newtonian $N$-body problem have been found by the direct methods of
Calculus of Variations, see for instance \cite{hiphop, simo2000,
  simo:newfamilies, chen03, FT2004, ch05, barrabes2006, terrvent07,
  fgn10}.  In fact an intensive search of such solutions followed
\cite{CM2000}, where the authors proved the existence of a new
periodic solution of the $3$-body problem, with three equal masses
following the same eight-shaped path, by minimizing the Lagrangian action
over a fundamental domain of a set of symmetric loops.  Indeed the use
of variational methods to search for periodic solutions of the
$3$-body problem was already proposed by Poincar\'e
in \cite{poincare1896, poincare1897}, where he noticed that collision
solutions have a finite value of the action, see also \cite{saariBook,
  wintner1941}.

The obstructions to the variational approach are essentially
two.
The first is that the Lagrangian action functional $\A$ is
not coercive on the whole Sobolev space of $T$-periodic loops
$H^1_T(\R, \R^{3N})$, which is the natural domain for
$\A$.  We can gain the coercivity
by restricting the domain of the
action, imposing symmetry or/and topological constraints on the admissible loops,
see for instance \cite{bcz91, dgg89, gordon77}.
The second problem is the possible presence of collisions in the minimizers.
Since we are interested in classical periodic solutions, we have to exclude them.
This can be done by level estimates and local perturbations
\cite{marchal01, ch02, chen03}.


In this paper we consider a $(1+N)$-body problem, composed by a
particle of mass $m_0 \gg 1$ (central body) and $N$ particles of equal
mass $m=1$ (satellites). The central body is at rest at the center of
mass of the whole system.  The $(1+N)$-body problem was already 
considered by Maxwell \cite{maxwell1859} to study the dynamical
structure of Saturn's rings.
Here we assume that the
interaction force between two particles is defined by a
potential of the form
\[
U \sim \frac{1}{r^\alpha}, 
\]
where $\alpha \in [1,2)$ and $r$ is the distance between the
  particles.  Imposing symmetry and topological constraints on the
  possible configurations of the satellites we find
  periodic orbits as minimizers of the Lagrangian action functional for each value of $m_0$ in a diverging sequence. Moreover,
  using $\Gamma$-convergence theory
  \cite{dalmaso1993, braides2002gamma}, we study the asymptotic
  behavior of the related sequence of minimizers.
  After
  a suitable rescaling, it turns out that the
  $\Gamma$-limit functional is the functional of a Kepler problem, defined on a
  set which may not contain planar loops, depending on the symmetry
  and topological constraints that we impose.
%
  We shall show some examples, with different symmetry constraints. In particular, we
  shall consider symmetries defined by the group $\Z_4$ (leading to
  the Hip-Hop solution \cite{hiphop}), by $\Z_2 \times \Z_2$, and by
  the rotation groups of Platonic polyhedra (used for instance in
  \cite{ferrario07, fgn10, fgfg14, FG18, FJ19}).
  $\Gamma$-convergence was already applied to the $N$-body problem in
  \cite{fgfg14}, where the authors considered the exponent $\alpha$ of
  the potential as a parameter, and studied the behavior of the
  minimizers as $\alpha\to +\infty$.
  
The paper is structured as follows.  In Section \ref{s:gammaConv} we
recall the definition of $\Gamma$-convergence and the results needed
for our purpose.  In Section \ref{s:Np1} we introduce the $(1+N)$-body
problem with symmetries and prove the $\Gamma$-convergence of the
Lagrangian action to the functional of a Kepler problem.
In Sections \ref{s:hiphopGamma}, \ref{s:kleinGamma}, and
\ref{s:platoGamma} we consider respectively the symmetry defined by
$\Z_4$, $\Z_2 \times \Z_2$, and the symmetry of a Platonic polyhedron.
In all the considered cases, we prove the existence of sequences of
collision-free minimizers, depending on $m_0$, and study the
minimizers of the corresponding $\Gamma$-limit problem.

\section{Definition and properties of $\Gamma$-convergence}
\label{s:gammaConv}
Many mechanical systems appearing in different branches of applied
mathematics depend on a parameter. As this parameter varies, sometimes it is
possible to imagine a certain limit behavior.
Studying such systems with variational techniques, we often deal with
a family of minimum problems
\[
   \min\{ \A_\varepsilon(u): u \in X\},
\]
with $\varepsilon>0$, where $X$ is a set
endowed with a notion of convergence. 
$\Gamma$-convergence theory can be used to describe the asymptotic
behavior of this family by means of a limit problem: this theory was
introduced by De Giorgi in the mid 70s in a series of papers, see
for example \cite{deGiorgi75, deGiorgi77, deGiorgi75_I}.

In the literature we can find many equivalent definitions of
$\Gamma$-convergence, as reported for example in
\cite{dalmaso1993}. Here we state the definition given in
\cite{braides2002gamma}, and list the main properties that we are
going to use in this paper.
\begin{definition}
   A sequence $\A_j : X \to \overline{\R}, \, j \in \N$ of functionals
   \textit{$\Gamma$-converges} in $X$ to a functional $\A_\infty: X
   \to \overline{\R}$ if for all $u \in X$ we have
   \begin{itemize}
      \item[(i)]($\liminf$ inequality) for every sequence $\{u_j\}_{j\in\N}$ converging to $u$ in $X$
         \begin{equation}
            \A_\infty(u) \leq \liminf_{j\to\infty}\A_j(u_j);
            \label{eq:limInfIn}
         \end{equation}
      \item[(ii)]({$\limsup$} inequality) there exists a sequence $\{u_j\}_{j\in\N}$ converging
         to $u$ in $X$ such that
         \begin{equation}
            \A_\infty(u) \geq \limsup_{j\to \infty}\A_j(u_j).
            \label{eq:limSupIn}
         \end{equation}
         This is usually called \textit{recovery sequence}.
   \end{itemize}
   We say that the functional $\A_\infty$ is the $\Gamma$-limit of $\{\A_j\}_{j \in \N}$, and write
   \[
   \Gamma\text{-}\lim_{j\to\infty} \A_j = \A_\infty\].
   \label{def:gammaLim}
\end{definition}

\begin{definition}
   Given a family of functionals $\A_\varepsilon:X \to \overline{\R}$
   depending on a real parameter $\varepsilon >0$, that can attain all
   the values in a right interval of $0$, we say that $\{
   \A_\varepsilon \}_{\varepsilon>0}$ $\Gamma$-converges to $\A_0$ if
   for all sequences $\{ \varepsilon_j \}_{j\in\N}\subseteq\R$
   converging to $0$ we have $\Gamma$-$\lim_{j\to\infty}
   \A_{\varepsilon_j} = \A_0$. If this is the case, we write
   $\Gamma$-$\lim_{\varepsilon\to0 }\A_\varepsilon = \A_0$.
\end{definition}

\begin{definition}
   We say that a sequence $\A_j : X \to
   \overline{\R}, \, j \in \N$ is \textit{equi-coercive} if
   there exists a compact set $K \subseteq X$ such that
   \[
      \inf_X \A_j = \inf_K \A_j,
   \]
   for every $j \in \N$.
   \label{def:equicoercive}
\end{definition}
Given an equi-coercive sequence of functionals, $\Gamma$-convergence
can be used to study the asymptotic behavior of sequences of
minimizers of such functionals.

\begin{theorem}[Convergence of minima and minimizers]
  Assume that 
  \[
  \Gamma\text{-}\lim_{j\to\infty} \A_j = \A_\infty,
  \]
  and the sequence $\{ \A_j \}_{j\in\N}$ is equi-coercive, and let $K$ be the compact set of
  Definition \ref{def:equicoercive}. Then
  \begin{enumerate}
  \item[(i)] $\A_\infty$ has a minimum in $X$;
  \item[(ii)] the sequence of infimum values converges to the minimum of $\A_\infty$, i.e.
    \begin{equation}
      \lim_{j\to \infty} \inf_X \A_j = \min_X \A_\infty;
      \label{eq:minConvTeo}
    \end{equation}
  \item[(iii)] if $\{ u_j \}_{j \in \N} \subset K$ is a sequence such that
    \[
       \lim_{j\to\infty}  \A_j(u_j) = \lim_{j\to\infty} \inf_{X}\A_j,
    \]
    then for every limit point $u_\infty\in X$ of $\{ u_{j} \}_{j \in \N}$ we have
    \[
    \A_\infty(u_\infty) = \min_X \A_\infty.
    \]
  \end{enumerate}
  \label{th:convMin}
\end{theorem}
\begin{proof}
   The proof can be deduced from Theorem 1.21 of \cite{braides2002gamma}.
\end{proof}

For monotone sequences of functionals, we can state more properties of the $\Gamma$-limit. To this purpose we introduce the following definition.
\begin{definition}
   Let $\A:X \to \overline{\R}$ be a functional. The \textit{lower semicontinuous envelope} of $\A$ is defined as
   \begin{equation}
   \begin{split}
       \overline{\A}(u) =  \sup  \bigl\{ & \mathcal{B}(u) \text{ s.t. } \mathcal{B}:X\to\overline{\R} \text{ is lower semicontinuous and 
       } \\ & \mathcal{B}(v) \leq \A(v) \text{ for all } v\in X \bigr\}.
       \label{eq:SCIenv}
   \end{split}
   \end{equation}
   \label{def:SCIenv}
\end{definition}

\begin{theorem}
   Let $\{ \A_j \}_{j \in \N}$ be a sequence of functionals such that $\A_{j+1} \leq \A_j$ for all $j \in \N$. Then the $\Gamma$-limit exists and corresponds to the lower semicontinuous envelope of $\inf_j \A_j$, i.e.
   \begin{equation}
       \Gamma\text{-}\lim_{j\to\infty}\A_j = \overline{\inf_j \A_j}.
       \label{eq:gammaLimMonotone}
   \end{equation}
   \label{th:gammaLimMonotone}
\end{theorem}
\begin{proof}
  The proof can be found in Remark 1.41 of \cite{braides2002gamma}.
\end{proof}

To explicitly compute the lower semicontinuous envelope we introduce the definition of relaxed functional, and prove that they are the same.
\begin{definition}
   Let $\A:X \to \overline{\R}$ be a functional. The \textit{relaxed functional} of $\A$ is defined as
   \begin{equation}
       \widehat{\A}(u) = \inf \bigl\{ \liminf_{j\to\infty}\A(u_j) : \{u_j\}_{j\in\N} \mbox{ is a sequence converging to } u \mbox{ in } X\bigr\}.
       \label{eq:relaxed}
   \end{equation}
   \label{def:relaxed}
\end{definition}
\begin{remark}
    The functionals $\overline{\A}$ and $\widehat{\A}$ are both lower semicontinuous (see for instance \cite{braides2002gamma}).
\end{remark}

\begin{theorem}
  We have
    \[ \overline{\A}(u) = \widehat{\A}(u) \]
    for all $u \in X$.
    \label{th:envelope_eq_relaxed}
\end{theorem}
\begin{proof}
  Let $u \in X$ and $\{ u_j \}_{j\in \N} \subset X$ be a sequence such that $u_j \to u$. Then
  \begin{equation}
    \liminf_{j\to\infty} \A(u_j) \geq \liminf_{j\to\infty} \overline{\A}(u_j) \geq \overline{\A}(u),
    \label{eq:rel1}
  \end{equation}
  where
  we used Definition~\ref{def:SCIenv}
  and the lower semicontinuity of $\overline{\A}$. Since \eqref{eq:rel1} holds for every sequence converging to $u$, we obtain $\overline{\A}(u) \leq \widehat{\A}(u)$.
  
  On the other hand, $\widehat{\A}$ is lower semicontinuous and
  $\widehat{\A}(v) \leq \A(v)$ for all $v \in X$. Indeed, using the
  constant sequence $v_j = v$ we get $\widehat{\A}(v) \leq \liminf_j
  \A(v_j) = \A(v)$, where the first inequality follows from
  Definition~\ref{def:relaxed}.
  Hence we have that $\overline{\A}(u) \geq \widehat{\A}(u)$ for all $u\in X$.
\end{proof} 

\section{The $(1+N)$-body problem with symmetries}
\label{s:Np1}
Let us consider a system of $N$ satellites with masses
$m_1=\dots=m_N=1$, and a central body with mass $m_0 \gg 1$, and
denote their positions with $u_i \in \R^3, i=0,\dots,N$. We assume
that
\begin{enumerate}
   \item[(1)] the center of mass of the whole system corresponds to
     the origin of the reference frame:
      \[
         \sum_{i=0}^N m_i u_i \equiv 0;
      \]
   \item[(2)] the central body is in equilibrium at the origin:
      \[
         u_0 \equiv 0.
      \]
\end{enumerate}
We define the configuration space $\mathcal{X}$ as
\[
   \mathcal{X} = \Bigl\{ u=(u_0,\dots,u_N) \in \R^{3(N+1)} : u_0 = 0, \, \sum_{i=1}^N u_i =
0 \Bigr\}.
\]
The particles move under the interaction forces generated by
potentials of the form $1/r^\alpha$, where $\alpha \in [1, 2)$ and $r$
  is the distance between two particles.  Note that for $\alpha=1$ we
  obtain the usual Newtonian gravitational interaction.  We write the
  potential separating the contribution of the central body from the
  interaction among the satellites:
\begin{equation}
   U_\alpha(u) = 
   \sum_{i=1}^{N}\frac{m_0}{|u_i|^\alpha}+ 
\sum_{1\leq i<j\leq N}\frac{1}{|u_i-u_j|^\alpha}.
   \label{eq:potentialNPU}
\end{equation}
Since $m_0$ is at rest, the kinetic energy contains only the terms due to
the motion of the satellites, that is
\begin{equation}
   K(\dot{u}) = \frac{1}{2}\sum_{i=1}^N|\dot{u}_i|^2,
   \label{eq:kineticPNU}
\end{equation}
and the Lagrangian is given by the sum
\[
   L_\alpha(u,\dot{u}) = K(\dot{u})+U_\alpha(u).
\]
For a fixed period $T>0$, consider the set of $T$-periodic loops
\[
H^1_T(\R, \mathcal{X}) = \{u \in H^1(\R, \mathcal{X}): u(0) = u(T)\}
\]
and define the Lagrangian action functional as
\[
   \A^\alpha(u) = \int_{0}^{T}L_\alpha(u,\dot{u}) \, dt,
\]
for $u \in H^1_T(\R, \mathcal{X})$. In the following we restrict
$\A^\alpha$ to sets of loops which are invariant under an action of a
group of rotations.  Let us denote with $\mathcal{G}$ a
subgroup of the $3$-D orthogonal group $O(3)$, containing as many
elements as the number of satellites,
i.e. $|\mathcal{G}|=N$. Then, labeling the satellites with the
elements of $\mathcal{G}$, we introduce the space of symmetric loops
\[
   \Lambda_\mathcal{G} = \big\{ u \in H^1_T(\R, \mathcal{X}) : u_R(t)=R u_I(t), \, R \in
      \mathcal{G}, \, t \in \R \big\},
\]
where $u_I:[0,T] \to \R^3$ is the motion of an arbitrarily selected
satellite, that we call the \textit{generating particle}\footnote{Note that since the
masses of the satellites are all equal and the loops in
$\Lambda_{\mathcal{G}}$ are invariant under rotations of the group $\mathcal{G}$, the
assumption $u_0 \equiv 0$ is admissible from a dynamical point of view, and non-collision
minimizers with this constraint are classical solutions of the $N$-body problem. We
remark that this would not be true if the sum of the forces exerted by $m_1,\dots,m_N$ on
$m_0$ did not vanish identically.}.  In the
following we shall discuss some examples, considering the $\Z_4$ group
(leading to the Hip-Hop solution \cite{hiphop, barrabes2006}
with a central body), the Klein group $\Z_2 \times \Z_2$ and the
symmetry groups of Platonic polyhedra \cite{fgn10, fgfg14, FG18}.
If we restrict the action $\A^\alpha$ to $\Lambda_{\mathcal{G}}$, then it depends only on the motion 
of the generating particle:
\begin{equation}
   \A^\alpha(u_I) = N \int_{0}^{T}\bigg( \frac{|\dot{u}_I|^2}{2} +\frac{m_0}{|u_I|^\alpha}
   + \frac{1}{2}\sum_{R \in \mathcal{G} \setminus \{ I\}}
   \frac{1}{|(R-\Id)u_I|^\alpha}  \bigg)dt.
   \label{eq:symmActFunAlpha}
\end{equation}
Note that a collision occurs if and only if there exist $R \in \mathcal{G}\setminus \{ I
\}$ and
$t_c \in [0,T]$ such that
\[
   u_I(t_c) = R u_I(t_c).
\]
We denote with
\[
   \Gamma = \{ x \in \R^3: Rx=x \text{ for some } R \in \mathcal{G} \setminus \{ I \} \},
\]
the set of collisions.
In the following we shall apply this scheme:
\begin{itemize}
   \item[(i)] impose additional constraints on the set of admissible loops
      to obtain coercivity of $\A^\alpha$;
   \item[(ii)] prove that, with these constraints, there exists a collision-free minimizer for each value
      of $m_0\geq 0$;
   \item[(iii)] find the $\Gamma$-limit and study the properties of its minimizers.
\end{itemize}

\subsection{The $(1+N)$-body problem and $\Gamma$-convergence}
Here we focus on the determination of the $\Gamma$-limit.  If we
consider the limit $m_0 \to \infty$, the integrand function in \eqref{eq:symmActFunAlpha} tends to $+\infty$, and it is not
clear what the $\Gamma$-limit is.
The usual technique to deal with this case is to perform a suitable rescaling of the motion.
We set
\[
   u_I(t) = m_0^{\beta} v_I(t), \quad t \in \R,
   \]
  where $\beta>0$ is the rescaling parameter, to be determined, 
and get
\[
   \A^\alpha(u_I) = N\int_{0}^{T} \bigg( 
   m_0^{2\beta}\frac{\lvert \dot{v}_I\rvert^2}{2} +
   \frac{m_0^{1-\alpha\beta}}{\lvert v_I \rvert^\alpha} +
   \frac{1}{2 m_0^{\alpha\beta}} 
   \sum_{ R \in \mathcal{G} \setminus \{I\}}\frac{1}{\lvert (R-\Id)v_I \rvert^{\alpha}} 
    \bigg) dt.
\]
We choose $\beta$ in a way to balance the exponent of $m_0$ in the
first and second terms inside the parentheses above, i.e we set
$2\beta=1-\alpha\beta$, so that
\[
   \beta = \frac{1}{2+\alpha}.
\]
Using this value, the action becomes
\[
\begin{split}
   \A^\alpha(u_I) & =N \int_{0}^{T} \bigg(
   m_0^{\frac{2}{2+\alpha}}\frac{\lvert  \dot{v}_I\rvert^2}{2}
   +
   \frac{m_0^{\frac{2}{2+\alpha}}}{\lvert v_I \rvert^\alpha} 
   +
   \frac{1}{2m_0^{\frac{\alpha}{2+\alpha}}} \sum_{R \in \mathcal{G} \setminus
   \{I\}}\frac{1}{\lvert (R-\Id)v_I \rvert^{\alpha}}
   \bigg) dt \\
   & = N m_0^{\frac{2}{2+\alpha}} \int_{0}^{T}\bigg(
   \frac{\lvert \dot{v}_I\rvert^2}{2}
   + 
   \frac{1}{\lvert v_I \rvert^\alpha}
   +
   \frac{1}{2m_0} \sum_{R \in \mathcal{G} \setminus \{I\}} \frac{1}{\lvert (R-\Id)v_I \rvert^{\alpha}} 
   \bigg) dt.
\end{split}
\]
Setting 
\[
   \varepsilon = \frac{1}{m_0}, 
\]
and neglecting the constants in front of the integral, we can consider the functional 
\begin{equation}
   \A^\alpha_\varepsilon(v_I) 
    = \int_{0}^{T}\bigg( 
    \frac{\lvert \dot{v}_I\rvert^2}{2} 
    +
    \frac{1}{\lvert v_I \rvert^\alpha} 
    +
    \frac{\varepsilon}{2} \sum_{R \in \mathcal{G} \setminus \{I\}} \frac{1}{\lvert
       (R-\Id)v_I \rvert^{\alpha}}
    \bigg) dt.
   \label{eq:actionFunctionEpsilon}
\end{equation}
This is the action of a perturbed Kepler problem, where the
perturbation becomes smaller and smaller as the mass of the central
body increases and, in the limit $\varepsilon \to 0$, it disappears.

Let us denote with
\[
   \mathcal{K} \subseteq H^1_T(\R, \R^3 \setminus \Gamma), 
\]
a set where $\A_\varepsilon^\alpha$ is defined and coercive for all $\varepsilon > 0$. 
We assume that $\K$ is open and connected in the $H^1$ topology, and the loops belonging to $\mathcal{K}$ are all collision-free.
Denoting with $\overline{\mathcal{K}}^{H^1}$ the $H^1$-closure of $\mathcal{K}$, collision loops 
in $\overline{\mathcal{K}}^{H^1}$ necessarily belong to the boundary
\[
\partial \K = \overline{\K}^{H^1} \setminus \K.
\]
Moreover, we assume the following property on the loops in $\mathcal{K}$, which
will be satisfied in all the examples that we are going to consider: there exists
a constant $c_{\mathcal{K}}>0$ such that, for every $v_I \in \mathcal{K}$ and for every
$\tau \in [0,T]$, we have
\begin{equation}
   \lvert v_I(\tau) \rvert \leq c_{\mathcal{K}}
       \max_{t,s \in [0,T]} \lvert
  v_I(t)-v_I(s) \rvert.
   \label{eq:coneCondition}
\end{equation}
Note also that the coercivity of
$\A^\alpha$ follows from condition \eqref{eq:coneCondition} because, along a sequence of
loops diverging in the $H^1$ norm, this condition implies that the trajectories become more 
and more elongated, and the kinetic energy goes to infinity.

Then we define\footnote{for simplicity, here we write $v$ instead of $v_I$.}
\begin{equation}
   \A_\varepsilon^\alpha(v) = 
   \begin{cases}
      \displaystyle \int_{0}^{T} \bigg(
      \frac{\lvert \dot{v} \rvert^2}{2} 
      +
      \frac{1}{\lvert v \rvert^\alpha} 
      +
      \frac{\varepsilon}{2} \sum_{R \in \mathcal{G} \setminus \{ I
      \}} \frac{1}{\lvert (R-\Id)v \rvert^\alpha} 
      \bigg)  dt, & v \in
      \overline{\mathcal{K}}^{H^1}, \\[2ex]
      +\infty, & v \in \overline{\mathcal{K}}^{L^2} \setminus
      \overline{\mathcal{K}}^{H^1},
   \end{cases}
   \label{eq:actionSeq}
\end{equation}
and
\begin{equation}
   \A_0^\alpha(v) = 
   \begin{cases}
      \displaystyle \int_{0}^{T} \bigg( \frac{\lvert \dot{v} \rvert^2}{2}+ \frac{1}{\lvert v
      \rvert^\alpha} \bigg) dt, & v \in
      \overline{\mathcal{K}}^{H^1}, \\[2ex]
      +\infty, & v \in  \overline{\mathcal{K}}^{L^2} \setminus
      \overline{\mathcal{K}}^{H^1},
   \end{cases}
   \label{eq:gammaLim}
\end{equation}
where $\overline{\mathcal{K}}^{L^2}$ denotes the $L^2$-closure of
$\mathcal{K}$. We also introduce 
\begin{equation}
    \mathcal{S} = \bigg\{ v \in \partial\K : 
    \int_0^T \sum_{R \in \mathcal{R}\setminus\{I\}} \frac{1}{|(R-I)v|^\alpha} dt = +\infty
    \bigg\}
\end{equation}
and define
\begin{equation}
    \widetilde{\A}_0^\alpha(v) = \inf_\varepsilon \A_\varepsilon^\alpha(v) = 
    \begin{cases}
      \displaystyle \int_{0}^{T} \bigg( \frac{\lvert \dot{v} \rvert^2}{2}+ \frac{1}{\lvert v
      \rvert^\alpha} \bigg) dt, & v \in
      \overline{\mathcal{K}}^{H^1} \setminus \mathcal{S}, \\[2ex]
      +\infty, & v \in  \big(\overline{\mathcal{K}}^{L^2} \setminus
      \overline{\mathcal{K}}^{H^1}\big) \cup \mathcal{S}.       
    \end{cases}
\end{equation}
In what follows we endow the space $\overline{\mathcal{K}}^{L^2}$ with the topology induced
by the $L^2$ norm. 
\begin{theorem}
  For every $\alpha \geq 1$, we have
  \begin{itemize}
   \item[\rm{i)}]$
      \Gamma\text{-}\lim_{\varepsilon \to 0} \A_\varepsilon^\alpha = \overline{\widetilde{\A}_0^\alpha};
      $
   \item[\rm{ii)}] $\A_0^\alpha(v) \leq \overline{\widetilde{\A}_0^\alpha}(v)$ for all $v \in \overline{\mathcal{K}}^{L^2}$, and if $v \in \overline{\mathcal{K}}^{H^1}$ does not pass through the origin then the equality holds;
   \item[\rm{iii)}] if a minimizer $v^*$ of $\A_0^\alpha$ does not pass through the origin, then $v^*$ is also a minimizer
   of $\overline{\widetilde{\A}_0^\alpha}$;
    \item[\rm{iv)}] the sequence $\{\A_\varepsilon^\alpha\}_{\varepsilon>0}$ is equi-coercive.
  \end{itemize}
   \label{th:gammaConv}
\end{theorem}
\begin{proof}
Point i) follows from the fact that $\{ \A_\varepsilon^\alpha \}_{\varepsilon>0}$ is a decreasing sequence of functionals and by Theorem~\ref{th:gammaLimMonotone}.

To prove point ii) we first show that $\A_0^\alpha$ is lower semicontinuous in the topology induced by the $L^2$ norm.
Let $\{ v_j\}_{j \in \N}\subseteq \overline{\K}^{L^2}$ such that $v_j \to v$
in $L^2$. If $\liminf_{j} \A_{0}^\alpha(v_j) = + \infty$ there is nothing to prove. 
Therefore we assume that
\begin{equation}
  \liminf_{j\to\infty} \A_{0}^\alpha(v_j) < +\infty.
  \label{finite_liminf}
\end{equation}
Then, up to subsequences, there exists $M>0$ such that
\begin{equation}
  \int_{0}^{T}\frac{\lvert \dot{v}_j \rvert^2}{2} dt \le \A_{0}^\alpha(v_j) \leq
    M,
  \label{eq:kinIneq}
\end{equation}
hence $\{\norm{v_j}_{H^1}\}_{j\in\N}$ is bounded and, again up to
subsequences, $v_j \wconv v$ in $H^1$. 
From H{\"o}lder's inequality and \eqref{eq:kinIneq} it follows that for all $t,s \in [0,T]$ and for all $j\in\N$ we
have
\[
\lvert v_j(t) - v_j(s) \rvert \le \int_{s}^{t} \lvert
\dot{v}_j(\tau) \rvert d\tau \le \sqrt{2TM}.
\]
Moreover, the functions are all bounded by the same constant, since
for every $\tau \in [0,T]$, by assumption \eqref{eq:coneCondition},
we have
\[
  \lvert v_j(\tau) \rvert \le c_{\mathcal{K}} 
  \max_{t,s \in [0,T]} \lvert
  v_j(t)-v_j(s) \rvert \le c_{\mathcal{K}} 
  \sqrt{2TM}.
\]
Then, by the Ascoli-Arzel\`a theorem, $v_j \to v$
uniformly in $[0,T]$, up to subsequences. We conclude that there exists a
subsequence $\{v_{j_k}\}_{k \in \N}\subseteq\overline{\K}^{H^1}$ such that
\begin{itemize}
       \item[(1)] $\lim_{k} \A^\alpha_{0}(v_{j_k}) =
          \liminf_{j} \A^\alpha_{0}(v_j)$;
       \item[(2)] $v_{j_k} \wconv v$ in $H^1$;
       \item[(3)] $v_{j_k} \to v$ uniformly in $[0,T]$.
\end{itemize}
It follows that
\[
    \begin{split}
       \liminf_{j \to \infty} \A^\alpha_{0}(v_j) & = \lim_{k \to
       \infty}\A^\alpha_{0}(v_{j_k}) \\
       & = \liminf_{k \to \infty} \int_{0}^{T}\bigg(\frac{\lvert \dot{v}_{j_k}
    \rvert^2}{2} +
       \frac{1}{\lvert v_{j_k} \rvert^\alpha}\bigg)dt \\
       & \geq \int_{0}^{T}\frac{\lvert \dot{v} \rvert^2}{2} dt + \int_{0}^{T}\bigg( \liminf_{k\to
       \infty} \frac{1}{\lvert v_{j_k} \rvert^\alpha} \bigg) dt \\
       & \geq \int_{0}^{T}\bigg( \frac{\lvert \dot{v} \rvert^2}2 +
       \frac{1}{\lvert v \rvert^\alpha} \bigg) dt \\
       & = \A_0^\alpha(v),
   \end{split}
\]
where we used the lower semicontinuity of the $L^2$ norm with
respect to the weak convergence and Fatou's lemma. 
Since $\A_0^\alpha$ is lower semicontinuous and $\A_0^\alpha(v) \leq \widetilde{\A}_0^\alpha(v)$ for all $v \in \overline{\mathcal{K}}^{L^2}$, the inequality $\A_0^\alpha(v) \leq \overline{\widetilde{\A}_0^\alpha}(v)$ follows from Definition~\ref{def:SCIenv}. Let $v \in \overline{\mathcal{K}}^{H^1}$ be a loop that does not pass through the origin. Then there exists a sequence $\{ v_j \}_{j\in\N}\subset \K$ such that $v_j \to v$ in $H^1$, and $v_j \to v$ uniformly. Since $v$ does not vanish, then $1/|v_j|^\alpha \to 1/|v|^\alpha$ uniformly. It follows that
\[
    \overline{\widetilde{\A}_0^\alpha}(v) \leq \liminf_{j\to \infty}\widetilde{\A}_0^\alpha(v_j) = \liminf_{j\to \infty} \A_0^\alpha(v_j) = \lim_{j\to\infty}\A_0^\alpha(v_j) = \A_0^\alpha(v),
\]
where the inequality follows from \eqref{eq:rel1} and from Definition \ref{def:relaxed} of relaxed functional, and the first equality from the fact that $\widetilde{\A}_0^\alpha$ and $\A_0^\alpha$ are the same on $\K$. The other two equalities follow from both the $H^1$ and the uniform convergence, used for the kinetic and the potential term respectively.
Point iii) follows immediately from ii). 

Let us prove iv), i.e. that the sequence $\{
\A^\alpha_\varepsilon\}_{\varepsilon > 0}$ is equi-coercive.  The
functional $\A_0^\alpha$ is coercive in
$\overline{\mathcal{K}}^{L^2}$ and lower semicontinuous, hence a minimizer exists. We observe that the sequence $\{
\A_\varepsilon^\alpha(v)\}_{\varepsilon> 0}\subseteq\R$ is decreasing
with $\varepsilon$ for each $v \in \overline{\mathcal{K}}^{L^2}$ and
\begin{equation}
   \A_0^\alpha(v) \leq \A_\varepsilon^\alpha(v), \quad 0 < \varepsilon \leq\varepsilon_0,
   \label{eq:diseqFunz}
\end{equation}
where $\varepsilon_0>0$.
Given $s\in \R$, we introduce the sub-levels 
\begin{gather}
   \K_s^{\varepsilon, \alpha} = \big{\{} v \in \overline{\mathcal{K}}^{L^2} : \A_\varepsilon^\alpha(v) \leq s
   \big{\}}, \\
   \K_s^{0,\alpha} = \big{\{} v \in \overline{\mathcal{K}}^{L^2} : \A_0^\alpha(v) \leq s
   \big{\}}.
\end{gather}
From \eqref{eq:diseqFunz} we have $\K_s^{\varepsilon,\alpha} \subseteq
\K_s^{0,\alpha}$ for all $\varepsilon>0$ and for $s \in \R$ large enough they are all non-empty.
Moreover, the sub-levels $\K_s^{0,\alpha}$ are compact w.r.t. the $L^2$ convergence since $\A_0^\alpha$ is coercive.
Therefore, the set $\K_s^{0,\alpha}$, for a fixed $s \in \R$ large enough, satisfies Definition \ref{def:equicoercive} 
of equi-coercivity for the sequence $\{\A_\varepsilon^\alpha\}_{\varepsilon>0}$.

\end{proof}


Theorems \ref{th:convMin} and \ref{th:gammaConv} implies that, to understand the
asymptotic behaviour of the minimizers of $\A^\alpha_\varepsilon$, we
can simply study the minimizers of the functional
$\A_0^\alpha$.

\section{$\Z_4$ symmetry: Hip-Hop constellations}
\label{s:hiphopGamma}
In this section we consider $N=4$ and discuss the existence of periodic
orbits called \textit{Hip-Hop solutions}, appearing in \cite{hiphop}
in the case without central body.  These solutions oscillates between
the square central configuration and the tetrahedral one.
%

Here we consider only the Keplerian case $\alpha=1$.
The rotation group of the Hip-Hop solution is isomorphic to 
\begin{equation}
   \Z_4 = \big\{ \Id, \mathsf{R}, \mathsf{R}^2, \mathsf{R}^3\big\}, \quad 
   \mathsf{R} = 
   \begin{pmatrix}
      0& -1& 0\\
      1 & 0 & 0 \\
      0 & 0 &-1
   \end{pmatrix}.
   \label{eq:hiphopgroup}
\end{equation}
Moreover, the collision set $\Gamma$ corresponds to the vertical axis
\[
   \Gamma = \R e_3,
\]
where $e_3 \in \R^3$ is the unit vector corresponding to the third coordinate axis.
To obtain the coercivity of the action functional, we restrict its domain to 
the loops $u \in \Lambda_{\Z_4}$ such that 
\begin{equation}
   u_I\bigg(t+\frac{T}{2}\bigg) = -u_I(t), \quad t \in \R.
   \label{eq:italianSymmetry}
\end{equation}
Relation \eqref{eq:italianSymmetry} is often called the \textit{Italian symmetry}, because it was already used in \cite{dgg89, zelati90}. Therefore, the set of admissible loops is
\begin{equation}
   \K = \big\{u_I \in H^1_T(\R, \R^3\setminus\Gamma) : u_I(t+T/2)=-u_I(t), \, t \in \R\big\}.
   \label{eq:Khiphop}
\end{equation}
\begin{proposition}
    The action $\A$ is coercive on $\K$ defined by \eqref{eq:Khiphop}, for all $m_0 \geq 0$.
    \label{prop:HiphopCoercivity}
\end{proposition}
\begin{proof}
Let $\{ u_I^{(k)} \}_{k \in \N} \subset \K$ be a sequence such that $\norm{u_I^{(k)}}_{H^1}\to \infty$, 
hence necessarily 
\begin{equation}
    \int_0^T |\dot{u}^{(k)}_I|^2 dt \to \infty.
    \label{eq:infVel}
\end{equation}
Indeed, if
\[
    \int_0^T |u_I^{(k)}(t)|^2 dt \to \infty,
\]
then there exists a sequence $\{t_k\}_{k\in\N}\subset [0,T]$ such that $|u_I^{(k)}(t_k)| \to \infty$. By the symmetry \eqref{eq:italianSymmetry}, it follows that
\[
    \big| u_I^{(k)}(t_k) - u_I^{(k)}(t_k+T/2)  \big|  = 2 \big| u_I^{(k)}(t_k) \big|\to\infty,
\]
and this is sufficient to have \eqref{eq:infVel}, hence the coercivity.
\end{proof}
Therefore, for each value of $m_0 \geq 0$, there exists a minimizer in the $H^1$ closure of $\K$, possibly with collisions.
The next step is the exclusion of collisions, in order to obtain a sequence of classical
solutions, depending on the parameter $m_0$, of the Newtonian $(1+4)$-body problem.
Note that in $\K$ there exists a $T$-periodic solution of the $(1+4)$-body problem with the satellites
placed at the vertexes of a square, which uniformly rotates around the central
body with period $T$. Let us denote with $\usquare_I(t)$ this
solution. With straightforward computations we get that
\begin{equation}
   \usquare_I(t) = 
   \begin{pmatrix}
     a \cos(\omega t) \\
     a \sin(\omega t) \\
      0
   \end{pmatrix}, \quad 
   a =
   \bigg(\frac{T}{2\pi}\bigg)^{2/3}\bigg(\frac{1}{\sqrt{2}}+\frac{1}{4}+m_0\bigg)^{1/3},
   \quad \omega = \frac{2\pi}{T}.
   \label{eq:solSquare1}
\end{equation}
\begin{lemma}
   The solution $\usquare_I \in \K$ given by \eqref{eq:solSquare1} is not a local
   minimizer of the action.
   \label{lemma:squareSolMin}
\end{lemma}
\begin{proof}
   To prove this result, it is sufficient to compute the second variation
   $\delta^2\A(\usquare_I)$ of the action
   and see that there exists a periodic variation $w:[0,T]\to \R^3$ for which 
   \begin{equation}
      \delta^2\A(\usquare_I)(w) < 0.
      \label{eq:negativeSecVar}
   \end{equation}
   To this end we consider vertical variations, i.e. we take 
   \[
      w(t) =
      \begin{pmatrix}
         0 \\0 \\w_z(t)
      \end{pmatrix},
      \]
      with $w_z:[0,T]\to\R$.
   Using the symmetries, the potential
   \[
      U(u_I) = \frac{m_0}{|u_I|} + \frac{1}{2}\sum_{R \in \Z_4 \setminus \{ I\}}
      \frac{1}{|(R-I)u_I|},
   \]
   can be written as
   \[
      U(u_I) = 
       \frac{m_0}{\sqrt{x^2+y^2+z^2}}
      +
      \frac{1}{2}\bigg[ \frac{\sqrt{2}}{\sqrt{x^2+y^2+2z^2}} +
      \frac{1}{2 \sqrt{x^2+y^2}} \bigg],
   \]
   where we have set $u_I = (x,y,z)$.
   The second variation $\delta^2\A$ is given by
   \[
      \delta^2\A(\usquare_I)(w) = \int_{0}^{T}\bigg( |\dot{w}(t)|^2 + w(t)\cdot \frac{\partial^2
      U\big(\usquare_I(t)\big)}{\partial u^2} w(t)\bigg) dt.
   \]
   Since we consider vertical variations, we only need to consider the following second derivatives
   \[
      \begin{split}
         \frac{\partial^2 U}{\partial z \partial x} & =
         3\sqrt{2}\frac{xz}{(x^2+y^2+2z^2)^{5/2}} + 3m_0
         \frac{xz}{(x^2+y^2+z^2)^{5/2}}, \\
         \frac{\partial^2 U}{\partial z \partial y} & =
         3\sqrt{2}\frac{yz}{(x^2+y^2+2z^2)^{5/2}} + 3m_0
         \frac{yz}{(x^2+y^2+z^2)^{5/2}}, \\
         \frac{\partial^2 U}{\partial z^2} & = -\frac{\sqrt{2}}{(x^2+y^2+2z^2)^{3/2}} -
         \frac{m_0}{(x^2+y^2+z^2)^{3/2}}\\&+3\sqrt{2}\frac{2z^2}{(x^2+y^2+2z^2)^{5/2}} + 3m_0
         \frac{z^2}{(x^2+y^2+z^2)^{5/2}}.
      \end{split}
   \]
   When we evaluate them at $\usquare_I(t)$ the only non-zero derivative is 
   \[
      \frac{\partial^2 U\big(\usquare_I(t)\big)}{\partial z^2} = -\frac{\sqrt{2}+m_0}{a^3}.
   \]
   Therefore, substituting in the second variation and using the
   expressions of $a$ and $\omega$ given in \eqref{eq:solSquare1} we
   obtain
   \[
      \delta^2\A(\usquare_I)(w) = \int_{0}^{T}\bigg( \dot{w}_z(t)^2 - w_z(t)^2
      \frac{\sqrt{2}+m_0}{\frac{1}{\sqrt{2}}+\frac{1}{4}+m_0} \omega^2 \bigg) dt.
   \]
   Using as vertical variation the function
   \[
      w_z(t) = \cos(\omega t)
   \]
   we get
   \[
      \delta^2\A(\usquare_I)(w) = \frac{\omega^2 T}{2}\bigg(1 -
      \frac{\sqrt{2}+m_0}{\frac{1}{\sqrt{2}}+\frac{1}{4}+m_0}\bigg) < 0,
   \]
   hence $\usquare_I$ is not a local minimizer.
\end{proof}
From this lemma and from the discussion in the next sections we can conclude that
minimizers of $\A$ are not planar, in a way similar to \cite{hiphop}.

\subsection{Total collisions}
To exclude total collisions we use level estimates. First we find an upper bound for the action 
of a solution with total collisions, and then we search for a collision-free loop which has a value of the
action smaller than that bound.
\begin{proposition}
    Let $u_I^* \in \overline{\K}$ be a solution with total collisions. Then the action satisfies
\begin{equation}
   \begin{split}
      \A(u_I^*) \geq 3 \cdot 2^{1/3}(2\pi)^{2/3}(3+4m_0)^{2/3}T^{1/3}.
   \end{split}
   \label{eq:totCollEstHipHop}
\end{equation}
    \label{prop:hiphopTotColl}
\end{proposition}
\begin{proof}
To estimate the action of a loop with total collisions we use Proposition \ref{prop:lower_bound_gen}.
The total mass is $\mathcal{M} = 4+m_0$,
and, if $u \in \Lambda_{\Z_4}$ and relation \eqref{eq:italianSymmetry} holds, 
the distance between two satellites satisfies
\[
   |u_h-u_k| \leq  2|u_I|, \quad h,k=1,\dots,4, \quad h \neq k,
   \]
where $u_j$ stands for $u_{\mathsf{R}^j}$, with $j=1,\ldots,4$.
Therefore, the potential $\mathcal{U}_\alpha$ of Proposition 
\ref{prop:lower_bound_gen}, with $\alpha=1$, satisfies 
\[
   \begin{split}
      \mathcal{U}_1(u) & = \frac{1}{4+m_0}\bigg( \sum_{\substack{h,k=1 \\ h \neq k}}^4
      \frac{1}{|u_h-u_k|} + 2 \sum_{i=1}^4 \frac{m_0}{|u_i|} \bigg) \\
      & = \frac{1}{4+m_0}\bigg( \sum_{\substack{h,k=1 \\ h \neq k}}^4
      \frac{1}{|u_h-u_k|} + \frac{8 m_0}{|u_I|} \bigg) \\
      & \geq \frac{1}{4+m_0} \bigg( \frac{6}{|u_I|} + \frac{8 m_0}{|u_I|} \bigg).
   \end{split}
\]
Moreover, defining
\[
   \rho(u) = \bigg( \sum_{i=1}^4 \frac{|u_i|^2}{\mathcal{M}} \bigg)^{1/2}, 
\]
because of the symmetry we have that
\[ 
\rho(u) = \frac{2}{\sqrt{4+m_0}} |u_I|. 
\]
The minimum of $\mathcal{U}_1(u)$ restricted to $\rho(u)=1$ satisfies
\[
   U_{1,0} := \min_{\rho(u)=1} \mathcal{U}_1(u) \geq \frac{4}{(4+m_0)^{3/2}}(3+4m_0).
\]
Consider now a solution $u_I^* \in \overline{\K}$ with a total collision. Because of the symmetry
\eqref{eq:italianSymmetry}, there are at least two total collisions per period,
therefore, from Proposition \ref{prop:lower_bound_gen}, the action functional satisfies
\[
   \begin{split}
      \A(u_I^*)  & \geq 2 \cdot \frac{3}{2} (4+m_0) (\pi
      U_0)^{2/3}\bigg(\frac{T}{2}\bigg)^{1/3} \\
      & \geq 3 \cdot 2^{1/3}(2\pi)^{2/3}(3+4m_0)^{2/3}T^{1/3}.
   \end{split}
\]
\end{proof}
\begin{proposition}
    Let $u_I^* \in \overline{\K}$ be a solution with total collisions. Then we have
    \begin{equation}
    \A(u_I^*) > \A(\usquare_I),
    \label{eq:totCollExclusion} 
    \end{equation}
    where $\usquare_I$ is the rotating square solution.
    \label{prop:actionSquare}
\end{proposition}
\begin{proof}
The action of the rotating square solution $\usquare_I$ given by
\eqref{eq:solSquare1} is
\begin{equation}
   \begin{split}
   \A(\usquare_I) & = \bigg( \frac{3+6\sqrt{2}}{2} + 6m_0 \bigg) \frac{(2\pi)^{2/3}}{\bigg(
      \frac{1}{\sqrt{2}} +\frac{1}{4} +m_0\bigg)^{1/3} }T^{1/3} \\
      & = \frac{3}{2^{1/3}} \big(1+2\sqrt{2}+4m_0\big)^{2/3} (2\pi)^{2/3}T^{1/3}.
   \end{split} 
   \label{eq:actionSquareSol}
\end{equation}
Set 
\[
   f(m_0)  = 3\cdot 2^{1/3} \big(3+4m_0\big)^{2/3}, 
   \quad
   g(m_0)  = 
\frac{3}{2^{1/3}} \big(1+2\sqrt{2}+4m_0\big)^{2/3}.
\]
With this notation, the bound \eqref{eq:totCollEstHipHop} for the action of a solution with total collisions $u_I^*$ 
and the action of the rotating square solution $\usquare_I$ can be written as
\[
   \A(u^*_I) \geq f(m_0) (2\pi)^{2/3}T^{1/3}, \qquad \A(\usquare_I) = g(m_0)(2\pi)^{2/3}T^{1/3},
\]
respectively.
The equation
$
   f(m_0) = g(m_0)
$
has a unique real solution
\[
   m_0 = \frac{2\sqrt{2}-5}{4} < 0,
\]
and $f(0)>g(0)$. This means that 
\[
   f(m_0) > g(m_0), \quad \forall m_0 \geq 0,
\]
hence we obtain $\A(u_I^*) > \A(\usquare_I)$.
\end{proof}
Propositions \ref{prop:hiphopTotColl} and \ref{prop:actionSquare} yield the following result. 
\begin{corollary}
    Minimizers of $\A$ in $\overline{\K}$ are free of total collisions, for every value of the central mass $m_0 \geq 0$.
\end{corollary}
\subsection{Partial collisions}
The method used to exclude partial collisions is similar to the one used in
\cite{hiphop}, where the central body is missing. 
Using cylindrical coordinates for the generating particle
\begin{equation}
   u_I(t) = \frac{1}{2} 
   \begin{pmatrix}
    \rho(t) \cos\varphi(t) \\ \rho(t)\sin\varphi(t) \\ \zeta(t) 
   \end{pmatrix},
   \label{ew:cylCoord}
\end{equation}
the Lagrangian of the functional \eqref{eq:symmActFunAlpha} is $L = K + U,$
where
\begin{equation}
   K= \frac{\dot{\rho}+\rho^2\dot{\varphi}^2+\dot{\zeta}^2}{2},\qquad
   U = \frac{4\sqrt{2}}{(\rho^2+2\zeta^2)^{1/2}} + \frac{2}{\rho} + \frac{8
   m_0}{(\rho^2+\zeta^2)^{1/2}}.
   \label{eq:lagHipHopCyl}
\end{equation}

Let us consider a solution $u_I^*\in \overline{\K}$ which has a partial collision at time $t=0$.
Since partial collisions can occur only on the vertical axis, we have
\[
   \rho^*(0) = \rho^*(T/2) = 0, \qquad \zeta^*(0)=-\zeta^*(T/2) \neq 0.
\]
Moreover, since the total energy and the angular momentum $\Phi = \rho^2\dot{\varphi}$ are first
integrals, we can easily deduce that $\Phi = 0$ for a solution with partial collisions,
hence it is contained in a vertical plane.
If there are several partial collisions, then every piece of solution between two consecutive partial collisions is
contained in a vertical plane. Using rotations around the vertical axis, we can reduce the problem to the case where all the pieces
are contained in the same vertical plane.
%
Therefore, without loss of generality, we can assume that $\varphi=0$, hence
\[
   u^*_I(t) =\frac{1}{2} 
   \begin{pmatrix}
      \rho^*(t)\\ 0\\ \zeta^*(t)
   \end{pmatrix}.
\]
\begin{lemma}
   If the trajectory of a solution $u^*_I(t)$ lies in a vertical plane, then it does not minimize the
   action.
   \label{lemma:vertPlane}
\end{lemma}
\begin{proof}
   We show that the action decreases if we rotate the orbit around the
   $x$ axis by a small angle $\gamma$. Let us denote with $\bar{u}_I$
   the rotated orbit and with $\bar{\rho}, \bar{\varphi}, \bar{\zeta}$
   the corresponding cylindrical coordinates.  The kinetic part
   remains unchanged:
   \[
      K(\dot{\bar{u}}_I) = \frac{1}{2}\bigg(\dot{\bar{\rho}}^2  +
      \bar{\rho}^2\dot{\bar{\varphi}}^2 + \dot{\bar{\zeta}}^2\bigg) = \frac{1}{2} (\dot{\rho}^2 +
      \dot{\zeta}^2) = K(\dot{u}^*_I).
   \]
   On the other hand, the potential becomes
   \[
      \begin{split}
         U(\bar{u}_I) & =
      \frac{4\sqrt{2}}{\big(\bar{\rho}^2 + 2 \bar{\zeta}^2\big)^{1/2}} + \frac{2}{\bar{\rho}} +
      \frac{8m_0}{\big(\bar{\rho}^2+\bar{\zeta}^2\big)^{1/2}} \\
      & = \frac{4\sqrt{2}}{(\rho^2+2\zeta^2-\zeta^2\sin^2\gamma)^{1/2}} +
      \frac{2}{(\rho^2+\zeta^2\sin^2\gamma)^{1/2}}+ \frac{8m_0}{(\rho+\zeta^2)^{1/2}}.
      \end{split}
   \]
   The difference between the actions of the two loops is
   \[
      \A(\bar{u}_I) - \A(u^*_I) = 2\int_{0}^{T/2}\big(U(\bar{u}_I)-U(u^*_I)\big)dt,
   \]
   and the term $U(\bar{u}_I)-U(u^*_I)$ in the integral does not contain the 
   part of the attraction due to the central body, like in
   the case with $m_0=0$. Hence, to prove that
   \[
       \A(\bar{u}_I) - \A(u^*_I) < 0,
   \]
   we can simply use the same proof given in \cite[Lemma 4]{hiphop}.
\end{proof}
From Lemma \ref{lemma:vertPlane}, we can conclude that minimizers are free of
collisions for every value of the mass $m_0\geq0$, hence they are
classical periodic solutions of the $(1+4)$-body problem.

\subsection{Minimizers of the $\Gamma$-limit}
\label{minimizers_hiphop}
In this setting, circular Keplerian orbits are compatible with the
set $\K$ of admissible loops.  Indeed, fixed a plane $\Pi \subseteq \R^3$ passing
through the origin, there exists a unique (up to phase shifts and
inversion of time) circular Keplerian orbit $u_I^{\sPi}:\R\to\R^3$
with period $T$, lying on $\Pi$ and satisfying
\eqref{eq:italianSymmetry},
hence it is an element of $\K$. Therefore, there is an infinite number
of minimizers of the $\Gamma$-limit functional in $\K$, represented by
circular motions.  Indeed, from \cite{gordon77} it is known that
all the $T$-periodic Keplerian ellipses (including the circular and the degenerate rectilinear ones) are minimizers of the
action of the Kepler problem in the set of planar $T$-periodic loops winding around the origin only once, either clockwise or counter-clockwise. Moreover non-circular orbits are not
compatible with relation \eqref{eq:italianSymmetry}.

\begin{figure}[!ht]
    \centerline{\includegraphics[ height=6cm]{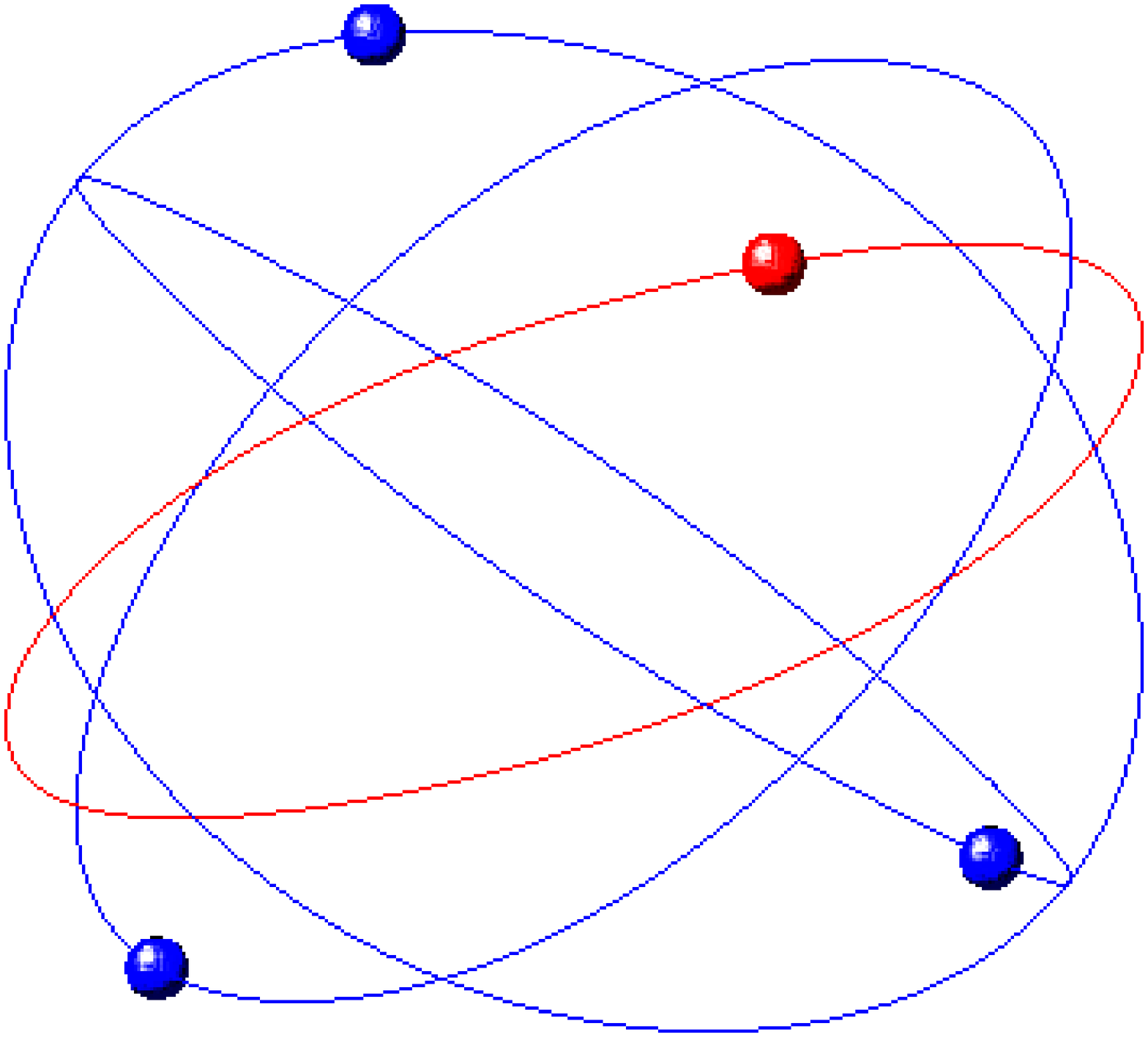}\hskip 1cm
    \includegraphics[ height=6cm]{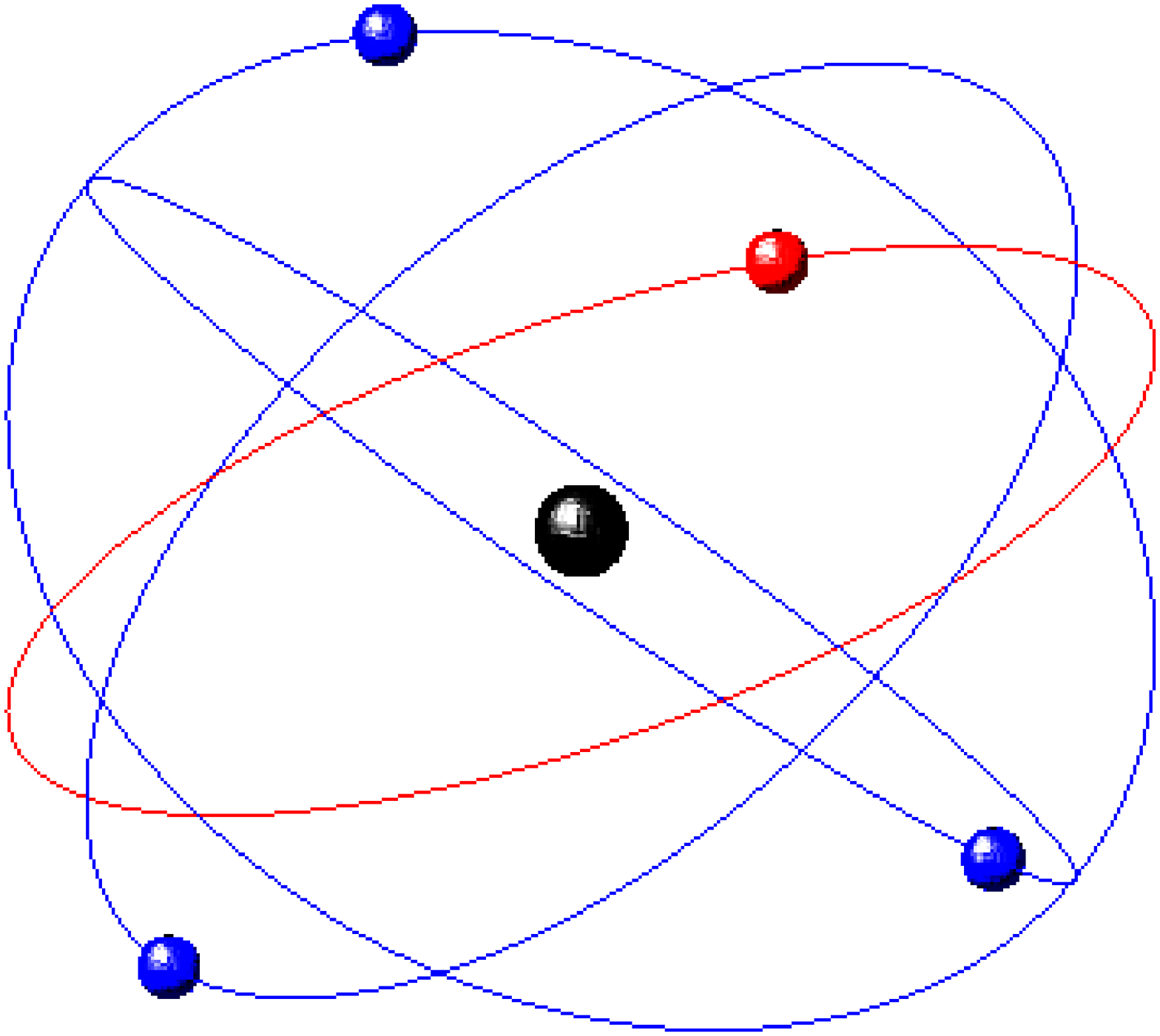}}
    \caption{The Hip-Hop solution. On the left we show the solution without central
    body, on the right the solution with a central body of mass $m_0=100$. The red curve
    represents the trajectory of the generating particle $u_I$.}
\label{fig:hiphopGamma}
\end{figure}
Note also that a solution with collisions cannot be a
minimizer. Indeed, because of relation \eqref{eq:italianSymmetry},
there are at least two collisions per period. In \cite{gordon77} these
are called multiple {\em legs} solutions and it is shown that their
action is strictly larger than the action of a circular orbit with
minimal period $T$.
This also means that all the
minimizers in a sequence $\{u^*_{m_0}\}_{m_0\geq0}$ are bounded away from the origin.
In Figure \ref{fig:hiphopGamma} we show two orbits, computed without a central body
(on the left) and with a central body of mass $m_0=100$ (on the right). Since the orbit with no
    central mass is almost circular, the difference in the trajectories of the satellites cannot be really appreciated in the two pictures.


\subsection{Constellations with $2N$ satellites}
In \cite{barrabes2006, terrvent07} the Hip-Hop
solution has been generalized to the case of $2N$ equal masses. 
Here we do the same in the case of the $(1+2N)$-body problem, with a massive central body
at the origin. The computations become longer, but techniques and arguments
are similar to the ones used above.
The symmetry group $\mathcal{G}$ in this case is 
\[
   \Z_{2N} = \{ I, \mathsf{R}, \dots, \mathsf{R}^{2N-1} \},
\]
where the group generator is
\[
   \mathsf{R} = 
   \begin{pmatrix}
      \cos\frac{\pi}{N} & -\sin\frac{\pi}{N} & 0 \\
      \sin\frac{\pi}{N} & \cos\frac{\pi}{N} & 0 \\
      0 & 0 & -1
   \end{pmatrix}.
\] 
As before, the collision set $\Gamma$ corresponds to the vertical axis
$ \R e_3$.
The loop set $\K$ is still defined imposing the symmetry \eqref{eq:italianSymmetry}:
\[
   \K = \big\{u_I \in H^1_T(\R, \R^3\setminus\Gamma) : u_I(t+T/2)=-u_I(t), \, t \in
   \R\big\},
\]
and the argument used to prove the coercivity of the action functional on $\K$ is the
same as before. 
\begin{figure}[!b]
    \centerline{\includegraphics[ height=6cm]{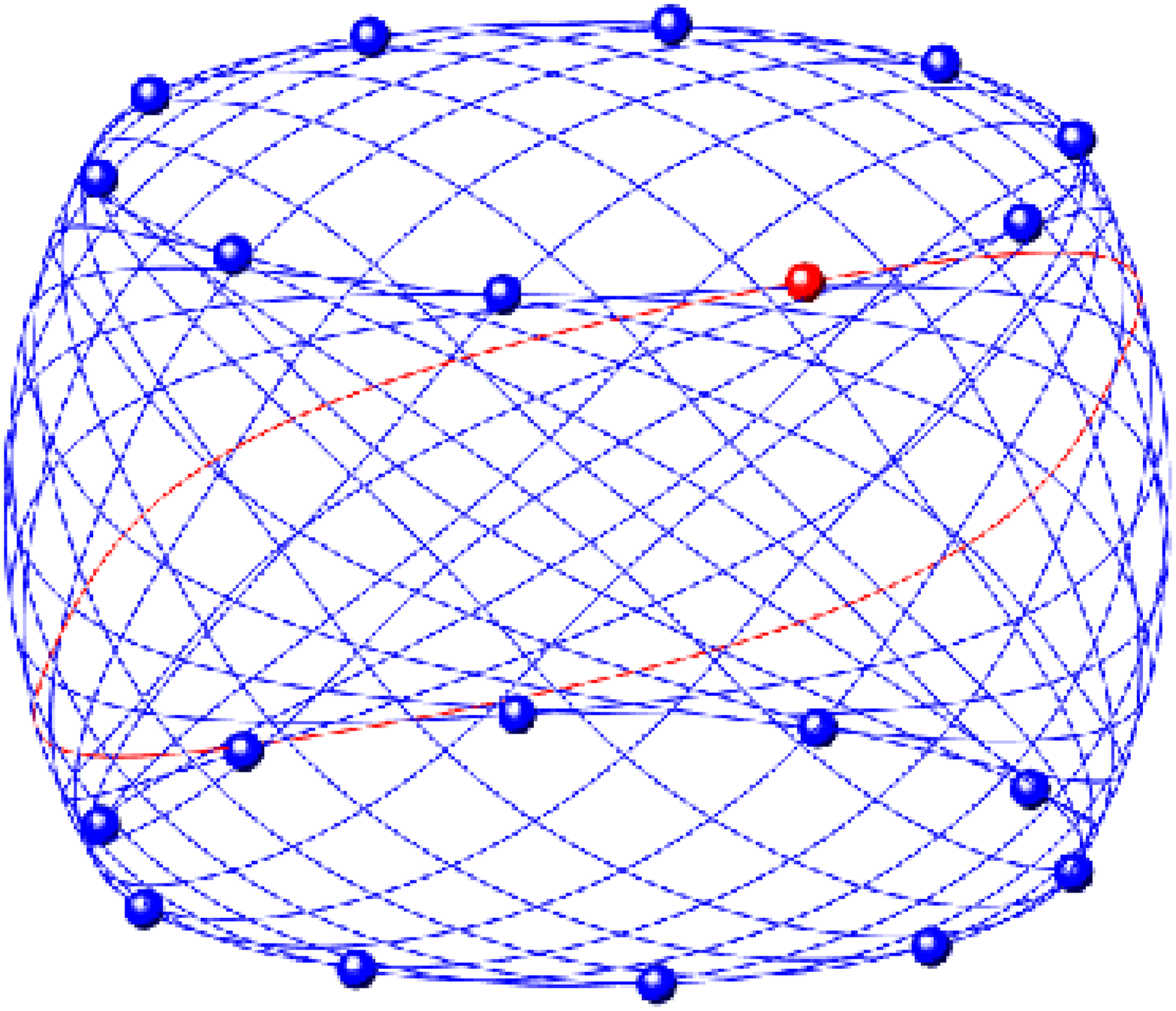}\hskip 1cm
    \includegraphics[ height=6cm]{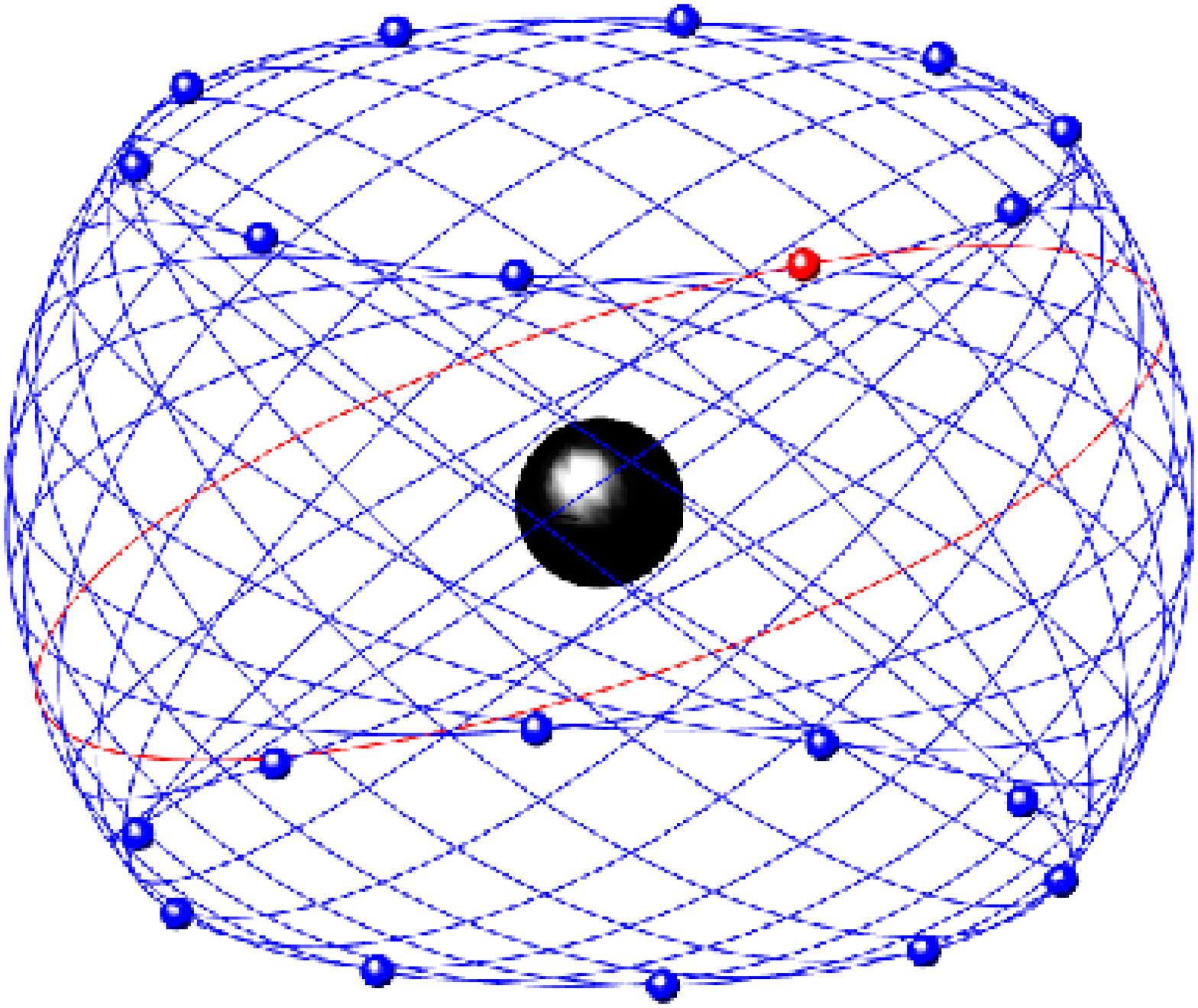}}
    \caption{The Hip-Hop solution for $20$ bodies. On the left it is shown the solution without central
    body, on the right the solution with a central body of mass $m_0=100$. The red curve
    represents the trajectory of the generating particle $u_I$.
    Note that here the difference
    between the trajectory of the satellites can be appreciated in the two pictures.}
   \label{fig:generalHipHop}
\end{figure}
The action of a solution $u_I^*$ with total collisions can be estimated with the results of Proposition
\ref{prop:lower_bound_gen}.
Then we can compare it with the action of the solution where
the satellites are placed at the vertexes of a planar regular $2N$-gon, 
which rotates uniformly around the origin, and check that the latter 
is lower.
Moreover, this problem is invariant under rotations around the
vertical axis, therefore solutions with partial collisions must lie on
a vertical plane.  Hence, we only have to find a small perturbation
$\bar u_I$ without collisions and with a lower value of the action.  This is
obtained in a way similar to \cite[Lemma 4]{hiphop} by applying a rotation of a
small angle $\gamma$ to the collision solution $u_I^*$. In this
case we have
\[
 {\cal A}(\bar{u}_I) - {\cal A}(u^*_I) = 2\int_{0}^{T/2}(A + B)dt,
\]
with
\small
\[
\begin{split}
  A&=  2\sqrt{2}N\sum_{h=1}^{N}\Biggl(\frac{1}{\sqrt{\rho^2\bigl(1-\cos\frac{(2h-1)\pi}{N}\bigr) + \zeta^2\bigl(1 - \sin^2\gamma\cos\frac{(2h-1)\pi}{N}\bigr)}} -  
  \frac{1}{\sqrt{\rho^2\bigl(1-\cos\frac{(2h-1)\pi}{N}\bigr)+ \zeta^2}} \Biggr),\cr
  B&=  2\sqrt{2}N\sum_{h=1}^{N-1}\Biggl(\frac{1}{\sqrt{(\rho^2+\zeta^2\sin^2\gamma)\bigl(1-\cos\frac{2h\pi}{N}\bigr)}} -  
  \frac{1}{\sqrt{\rho^2\bigl(1-\cos\frac{2h\pi}{N}\bigr)}} \Biggr).\cr
\end{split}
\]
\normalsize
The discussion for the $\Gamma$-limit is the same as in Section~\ref{minimizers_hiphop}, since the generating particle still
moves on a circular orbit.
%
An example of these orbits for $N=10$ is shown in Figure
\ref{fig:generalHipHop}, together with an approximation of the
minimizer of the $\Gamma$-limit.

\section{$\Z_2\times\Z_2$ symmetry}
\label{s:kleinGamma}
In this section we consider $N=4, \, \alpha=1$ and discuss the existence of periodic orbits with the symmetry of the
Klein group $\mathcal{G}=\Z_2 \times \Z_2$, appearing in \cite{fgn10} in the case without a central body. 
Using the rotations in $\R^3$, the Klein group can be written as 
\[
   \Z_2\times\Z_2 = \{ \Id,R_2,R_3,R_4 \},
\]
where
\[
   R_2 = 
   \begin{pmatrix}
      1 &  0 & 0 \\
      0 & -1 & 0 \\
      0 & 0 & -1
   \end{pmatrix}, \quad
   R_3 = 
   \begin{pmatrix}
      -1 &  0 & 0 \\
      0 & 1 & 0 \\
      0 & 0 & -1
   \end{pmatrix}, \quad
   R_4 = 
   \begin{pmatrix}
      -1 &  0 & 0 \\
      0 & -1 & 0 \\
      0 & 0 & 1
   \end{pmatrix}
\]
are the rotations of $\pi$ around the three coordinate axes.
Moreover, the collisions set $\Gamma$ corresponds to the union of the three coordinate
axes:
\[
   \Gamma = \bigcup_{i=1}^3 \{\alpha e_i, \, \alpha \in \R \},
\]
where $e_i \in \R^3$ is the unit vector corresponding to the $i$-th coordinate axis.

We consider loops $u \in \Lambda_{\Z_2 \times \Z_2}$ with the additional
symmetry
\begin{equation}
\begin{cases}
   u_I(t) = \widetilde{R}_3 u_I(-t), \\
   u_I(t) = \widetilde{R}_2 u_I(T/2 - t),
\end{cases}
\label{eq:refSym}
\end{equation}
where $\widetilde{R}_j$ is the reflection with respect to the plane $\{ x_j=0 \}$, passing through the origin and orthogonal to $e_j$.
Moreover, we restrict the action functional to the set 
\begin{equation}
   \K = \bigg{\{} u_I \in H^1_T(\R, \R^3 \setminus \Gamma) :
   u_I \text{ satisfies \eqref{eq:refSym} and } u_I(0) \in S_1, u_I(T/4) \in S_2 \bigg{\}},
   \label{eq:Kklein}
\end{equation}
where
\[
S_1 = \{ \alpha e_1 + \beta e_2, \alpha,\beta >0 \},\qquad  S_2 = \{ -\alpha e_1 + \beta
e_3, \alpha,\beta> 0\}
\]
are two quadrants of the planes $\{ x_3 = 0 \},
\, \{x_2 = 0 \}$, respectively. 
\begin{proposition}
The action $\A$ is coercive on $\K$ defined by \eqref{eq:Kklein}, for all $m_0 \geq 0$.
\end{proposition}
\begin{proof}
The proof is similar to Proposition~\ref{prop:HiphopCoercivity}. 
Note that by the definition of $S_1$ and $S_2$, we have
\begin{equation}
    \big| u_I(0) - u_I(T/4) \big| \geq \sqrt{|u_I(0)|^2 + |u_I(T/4)|^2},
    \label{eq:ineqT4}
\end{equation}
for every $u_I \in \K$. Let $\{ u_I^{(k)} \}_{k \in \N} \subset \K$ be a sequence such that $\norm{u_I^{(k)}}_{H^1}\to \infty$, hence necessarily 
\begin{equation}
    \int_0^T |\dot{u}^{(k)}_I(t)|^2 dt \to \infty.
    \label{eq:infVel2}
\end{equation}
Indeed, if
\[
    \int_0^T |u_I^{(k)}(t)|^2 dt \to \infty,
\]
then there exists a sequence $\{t_k\}_{k\in\N}\subset [0,T]$ such that $|u_I^{(k)}(t_k)| \to \infty$. If $u_I^{(k)}(0)$ is 
bounded, then $| u_I^{(k)}(0) - u_I^{(k)}(t_k)| \to \infty$, while if $u_I^{(k)}(0)$ is unbounded, then 
$| u_I^{(k)}(0) - u_I^{(k)}(T/4) | \to \infty$ by \eqref{eq:ineqT4}. Therefore in both cases \eqref{eq:infVel2} holds, and the coercivity follows from the fact that the kinetic energy goes to $+\infty$ along $\{ u_I^{(k)} \}_{k \in \N}$.
\end{proof}
Therefore, minimizers exist for every value of $m_0 \geq 0$. Next step is the exclusion of collisions. 
\subsection{Exclusion of collisions}
\paragraph{Total collisions.}
To exclude total collisions we still use level estimates.
\begin{proposition}
Let $u_I^* \in \overline{\K}$ be a solution with total collisions. Then the action satisfies
\begin{equation}
   \A(u_I^*)  \geq  6 \pi^{2/3}\bigg(3\sqrt{\frac{3}{2}} + 4m_0\bigg)^{2/3}T^{1/3}. 
   \label{eq:totCollEst4}
\end{equation}
\end{proposition}
\begin{proof}
With a notation similar to Section~\ref{s:hiphopGamma}, we have $\mathcal{M} = 4+m_0$ and
\[
   \begin{split}
   \mathcal{U}_1(u) &  = \frac{1}{4+m_0} \bigg( \sum_{\substack{h,k=1 \\ h \neq k}}^4 \frac{1}{|u_h-u_k|} +
   2\sum_{i=1}^4 \frac{m_0}{|u_i|} \bigg) \\
   & = \frac{2}{4+m_0} \bigg(  \sum_{j=1}^3\frac{1}{|u_I \times e_j|} +
   \frac{4 m_0}{|u_I|} \bigg).
   \end{split}
\] 
Moreover, $\rho(u) = 1$ if and only if 
\[
   |u_I| = \frac{\sqrt{4+m_0}}{2},
\]
hence
\[
   U_{1,0}:= \min_{\rho(u)=1} \mathcal{U}_1(u) =
   \frac{4}{(4+m_0)^{3/2}}\bigg(3\sqrt{\frac{3}{2}}+4m_0\bigg). 
\]
Let $u_I^* \in \overline{\K}$ be a solution with a total collision. Because of the
symmetry \eqref{eq:refSym}, there are at least two total collisions per period, therefore, from
Proposition \ref{prop:lower_bound_gen}, its action satisfies
\[
   \begin{split}
   \A(u_I^*) & \geq 2\cdot \frac{3}{2} (4+m_0)(\pi U_{1,0})^{2/3}\bigg( \frac{T}{2}\bigg)^{1/3} \\
   & = \frac{3}{2^{1/3}}\pi^{2/3}4^{2/3}\bigg(3\sqrt{\frac{3}{2}} +
   4m_0\bigg)^{2/3}T^{1/3} \\
   & = 6 \pi^{2/3}\bigg(3\sqrt{\frac{3}{2}} +
   4m_0\bigg)^{2/3}T^{1/3}. 
   \end{split}
\]
\end{proof}

We search for a collision-less loop whose action is less than the lower
bound in \eqref{eq:totCollEst4}. To do that, let $\rho > 0 $ and take
a loop $v_I \in \K$ such that the generating particle moves with
uniform velocity on a closed curve constructed as the union of four half
circles $C_1^{\pm}, C_2^{\pm}$ of radius $\rho$.  $C_1^{\pm}$ lies on
the plane $\{ x_3 = \pm \rho \}$, with its center on the axis $x_3$,
and $C_2^{\pm}$ lies on the plane $\{ x_2 = \pm \rho \}$, with its
center on the axis $x_2$, see Figure \ref{fig:kleinClass} for a
sketch.
\begin{figure}[!ht]
    \centerline{
       \includegraphics[width=6cm]{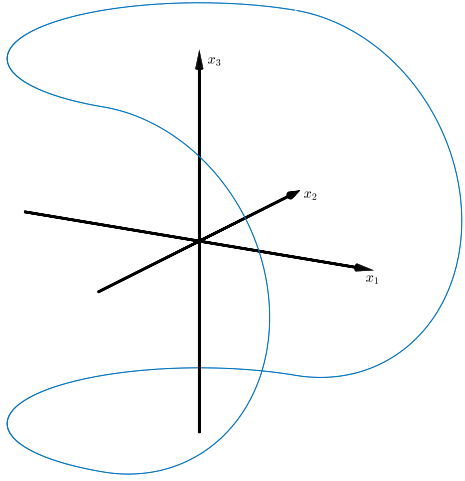}}
    \caption{The shape of the generating particle of the loop $v$ used to exclude total collisions.}
  \label{fig:kleinClass}
\end{figure}
\begin{proposition}
    Let $u_I^* \in \overline{\K}$ be a solution with total collisions. Then we have 
    \begin{equation}
    \A(u_I^*) > \A(v_I),
    \label{eq:totCollExclusionKlein} 
    \end{equation}
    for some $\rho>0$.
\end{proposition}
\begin{proof}
From the definition of $v_I$ we have
\[
|(R-I)v_I(t)|\geq 2\rho,
\]
for all $R \in \mathcal{R} \setminus \{ I \}$ and $t \in \R$, and $|v_I(t)|=\sqrt{2} \rho$, for all
$t \in \R$. Hence for the action we have the
estimate
\begin{equation}
   \A(v_I) \leq 32\frac{\pi^2 \rho^2}{T} + \frac{3+2\sqrt{2}m_0}{\rho}T.
   \label{eq:circEst}
\end{equation}
Choosing
\[
   \rho = \bigg(\frac{3+2\sqrt{2}m_0}{64\pi^2}T^2\bigg)^{1/3}
\]
we minimize the value of the right hand side of \eqref{eq:circEst}, therefore
\[
   \begin{split}
   \A(v_I) & \leq \bigg( \frac{32\pi^2}{(64\pi^2)^{2/3}} + (64\pi^2)^{1/3}
   \bigg)(3+2\sqrt{2}m_0)^{2/3}T^{1/3} \\
   & = 6\pi^{2/3}(3+2\sqrt{2}m_0)^{2/3}T^{1/3}.
\end{split}
\]
Comparing this estimate with \eqref{eq:totCollEst4} we see that
\[
   \A(v_I) < \A(u_I^*), 
\]
for every value of $m_0 \geq 0$. 
\end{proof}
\begin{corollary}
    Minimizers of $\A$ in $\overline{\K}$ are free of total collisions, for every value of the central mass $m_0 \geq 0$.
\end{corollary}
\paragraph{Partial collisions.}
Let $u^*_I \in \overline{\K}$ be a minimizer with partial collisions and $(t_1,t_2)$ be an interval of
regularity. Then $u^*_I$ is a solution of the equation
\begin{equation}
   \ddot{w} = \sum_{R \in \Z_2 \times \Z_2 \setminus\{ I \}} \frac{(R-I)w}{|(R-I)w|^3}-m_0\frac{w}{|w|^3}, \quad t \in (t_1,t_2).
   \label{eq:newton4b}
\end{equation}
Partial collisions can be excluded as in \cite{fgn10}. Indeed, they can only occur on a coordinate axes 
and, using the blow-up technique \cite{FT2004}, they can be seen locally as parabolic double collisions in a perturbed Kepler 
problem. The term due to the presence of the central body with mass $m_0$ turns out to be irrelevant
for the discussion (as for the case of the Hip-Hop solution of Section \ref{s:hiphopGamma}), 
since it is included in the perturbation, and does not play any relevant role in the
estimates. The proof is similar to the one recalled in
Section \ref{ss:platoPartialColl}, where the symmetry of Platonic polyhedra is considered.

%
%
%
Therefore, for every choice of the mass $m_0 \geq 0$, there exists
a collision-free minimizer, hence a classical solution of the $(1+4)$-body problem. 

\subsection{Minimizers of the $\Gamma$-limit}
Let $u_I^* \in \overline{\K}$ be a minimizer of the $\Gamma$-limit
functional \eqref{eq:gammaLim}. We note that there exists $\bar{u}_I \in \partial\K$ such
that $\bar{u}_I([0,T])$ is a Keplerian circle with center at the
origin, hence we can exclude total collisions in $u_I^*$ like in Section~\ref{s:hiphopGamma}.
%
It follows that, up to translations of time, $(0,T/4)$ is an interval of regularity of the solution,
and $u_I^*$ solves the Keplerian equations of motion.
Therefore $u_I^*\big( [0,T/4] \big) \subseteq \Pi$ where $\Pi \subset \R^3$ is a plane
passing through the origin.
Moreover
\begin{equation}
   u_I^*(0) \in \overline{S}_1, \quad u_I^*(T/4) \in \overline{S}_2.
   \label{eq:BCKlein}
\end{equation}
\begin{figure}[!ht]
    \centerline{
    \includegraphics[width=6cm, height=6cm]{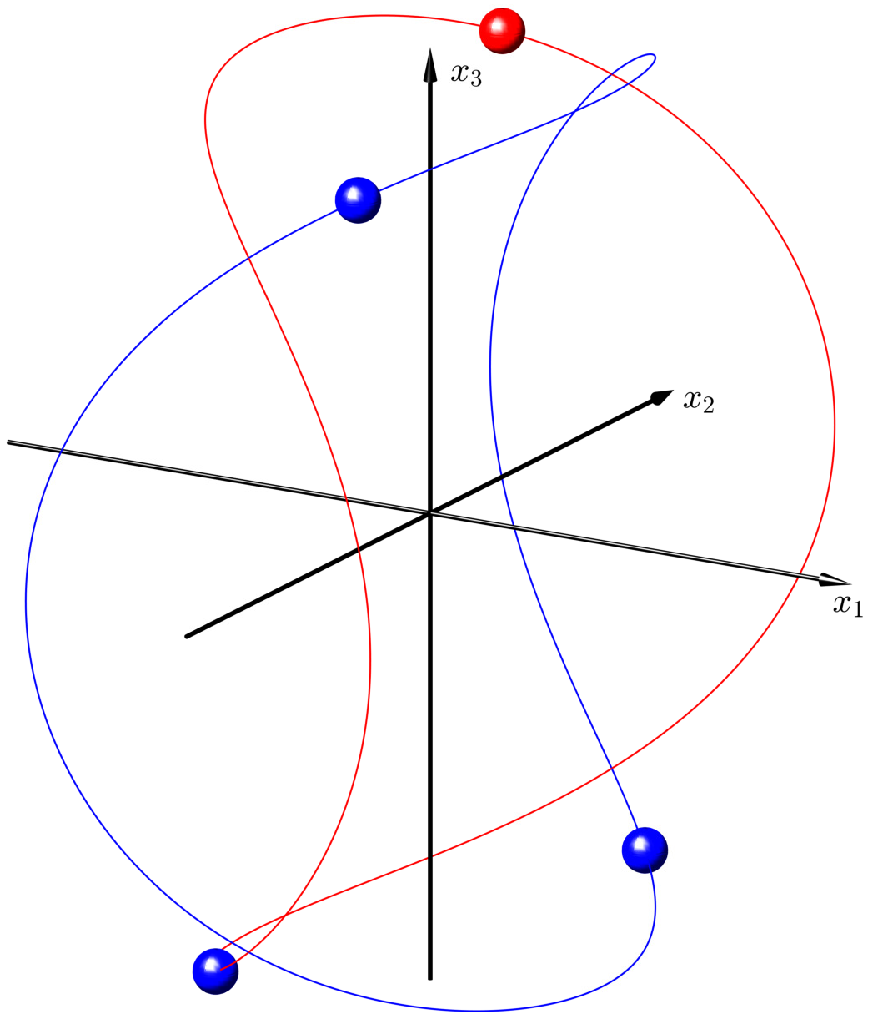}
    \hskip 0.5cm
    \includegraphics[width=6cm, height=6cm]{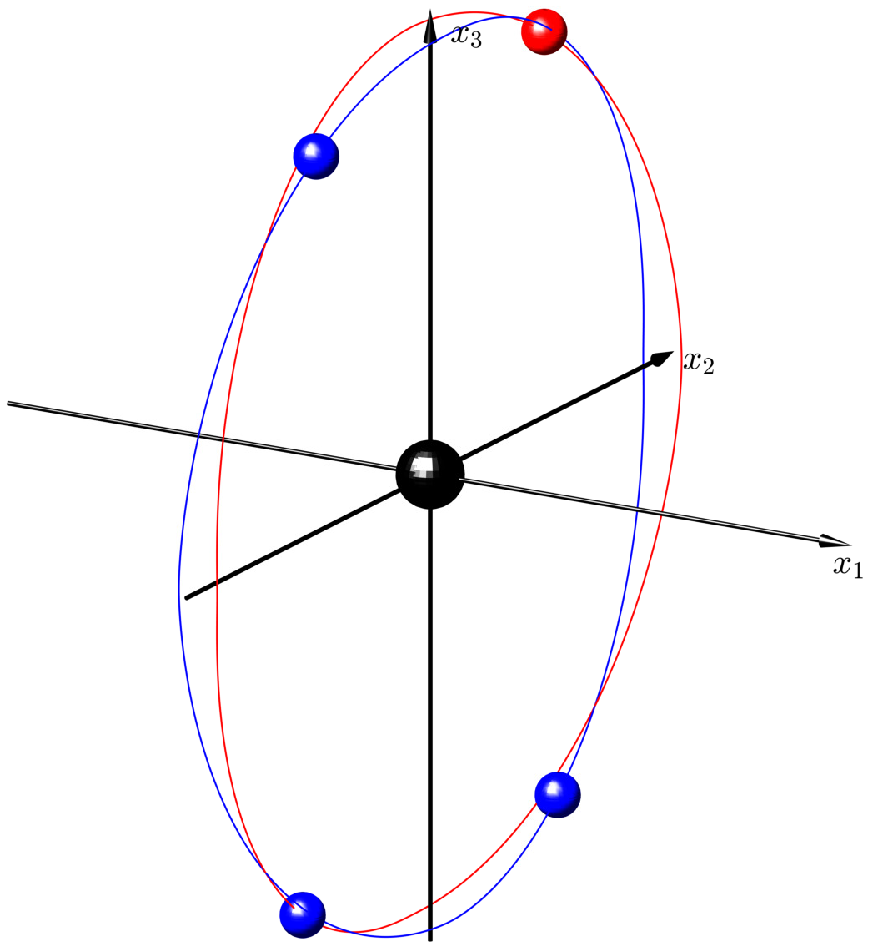}
    }
  \caption{Two minimizers of the action functional with the $\Z_2\times\Z_2$ symmetry. On
  the left $m_0=0$, on the right $m_0=1000$. The blue curve is the trajectory of the
  generating particle.}
  \label{fig:gammaMin4body}
\end{figure}
By conditions \eqref{eq:BCKlein} and \eqref{eq:refSym}, $\Pi$ coincides with a coordinate plane
$\{x_1=0\}, \{x_2=0\}$ or $\{x_3 = 0\}$. 
In fact, if this were not the case, then $u_I^*$ would lie on the same plane also for times $t>T/4$,
so that it would not belong to $\overline{\K}$. 
Moreover, $\Pi$ cannot be $\{x_2 = 0\}$ or $\{x_3 = 0\}$. Indeed, loops in $\overline{\K}$
are entirely contained in these two planes and, since they satisfy \eqref{eq:refSym}, are necessarily
multiple legs solutions, which have at least two collisions per period, therefore
they cannot be minimizers.  
The only possibility is that $\Pi = \{x_1 = 0\}$, and extending the minimizer to the whole time interval
$[0,T]$ using reflections, we obtain that $u_I^*$ lies entirely on $\Pi$.
Moreover, $u_I^*([0,T])$ cannot be an ellipse, since ellipses do not satisfy \eqref{eq:refSym}.
The only remaining possibility is that $u_I^*$ is circular, hence the minimizer $u_I^* \in \partial\K$ is 
\[
   \begin{split}
      u_{I}^*(t) = r\big(\cos(\omega t) e_2 +  \sin(\omega t)e_3\big), \quad \omega =
   \frac{2\pi}{T}, 
   \end{split}
\]
where  
\[
r = \bigg(\frac{T}{2\pi}\bigg)^{\frac{2}{3}}
\]
is given by the third Kepler's law.
In Figure \ref{fig:gammaMin4body} we draw the trajectories of a minimizer for different
values of the mass $m_0$. As $m_0$ increases, the
trajectory of the generating particle becomes closer and closer to the circular loop lying
in the plane $\{x_1=0\}$, and the four satellites pass closer and closer to two 
simultaneous double collisions.

It is worth noting that, in the limit $m_0 \to \infty$, the satellites do not bounce 
back and forth at double collisions, as it happens when we regularize the
Keplerian equations of motion (see, for instance, \cite{mcgehee1981}). 
Instead, they continue moving on the same circular trajectory and cross each other at collisions.
\section{Symmetry of Platonic polyhedra}
\label{s:platoGamma}
In this section we take into account the symmetry groups of 
Platonic polyhedra.  Periodic orbits sharing these symmetries have already been 
found in \cite{ferrario07,fgn10, FG18}, in the case without central
mass.

We denote with $\mathcal{T}$ the rotation group of the Tetrahedron,
with $\mathcal{O}$ the rotation group of the Octahedron and the Cube,
and with $\mathcal{I}$ the rotation group of the Icosahedron and the
Dodecahedron.  Then we take $\mathcal{G} = \mathcal{R} \in \{
\mathcal{T}, \mathcal{O}, \mathcal{I} \}$ and set the number of
satellites to be $N = |\mathcal{R}|$.  Therefore $N$ can be equal to
$12, 24$ or $60$, depending on the selected group.  Here we consider
 $\alpha \in [1,2)$ for the exponent of the potential energy,
  defining the force of attraction.
  Besides the constraint given by $u\in\Lambda_{\cal R}$, i.e.
  \begin{itemize}
\item[(a)] \hskip 4.5cm $u_R(t) = Ru_I(t), \qquad\forall R\in{\cal R}$,
  \end{itemize}
we assume that
\begin{itemize}
   \item[(b)] the trajectory of the generating particle $u_I$ belongs
     to a given non-trivial free-homotopy class of $\R^3 \setminus
     \Gamma$, where
      \[
         \Gamma = \cup_{R \in \mathcal{R} \setminus \{ I \}} r(R), 
      \]
      is the set of collisions and $r(R)$ is the rotation axis of $R$;
   \item[(c)] there exist $R \in \mathcal{R}$ and an integer $M$ such that
      \begin{equation}
         u_I(t+T/M) = Ru_I(t),
         \label{eq:extrasym}
      \end{equation}
      for all $t \in \R$.
\end{itemize}
We search for minimizers of $\A^\alpha$ in the cone
\begin{equation}
   \K = \{ u_I \in H_1^T(\R, \R^3\setminus\Gamma) : \text{ (b) and (c)
     hold} \}.
     \label{eq:Kplato}
\end{equation}
\begin{proposition}
    Assume that the free-homotopy class in (b) is not represented by a loop winding around one rotation axis only. Then the action $\A^\alpha$ is coercive on $\K$ defined by \eqref{eq:Kplato}, for all $m_0\geq 0 $ and for all $\alpha \geq 1$.
\label{prop:PlatoCoercivity}
\end{proposition}
\begin{proof}
    Again, the proof is similar to Proposition~\ref{prop:HiphopCoercivity}.  By the assumption made on the free-homotopy class, we have that there exists $c_{\K} > 0$ such that
    \begin{equation}
        \max_{t,s \in [0,T]} |u_I(t) - u_I(s)| \geq c_{\K}\min_{t \in [0,T]}|u_I(t)|,
    \end{equation}
    for all $u_I \in \K$.
    Let $\{ u_I^{(k)} \}_{k \in \N} \subset \K$ be a sequence such that $\norm{u_I^{(k)}}_{H^1}\to \infty$, 
    hence necessarily 
    \begin{equation}
        \int_0^T |\dot{u}^{(k)}_I|^2 dt \to \infty.
        \label{eq:infVel3}
    \end{equation}
    Indeed, if
    \[
        \int_0^T |u_I^{(k)}(t)|^2 dt \to \infty,
    \]
    then there exists a sequence $\{t_k\}_{k\in\N}\subset [0,T]$ such that $|u_I^{(k)}(t_k)| \to \infty$. Let $t_k^m$ be such that $|u_I^{(k)}(t_k^m) |= \min_{t\in[0,T]} |u_I^k(t)|$. Then
    \[
    \begin{split}
        |u_I^{(k)}(t_k)| & \leq |u_I^{(k)}(t_k^m)| + |u_I^{k}(t_k) - u_I^{k}(t_k^m)| \\
        & \leq (1/c_\K + 1) \max_{t,s \in [0,T]} |u_I^{(k)}(t) - u_I^{(k)}(s)| \to \infty, 
    \end{split}
    \]
    and this is sufficient to have \eqref{eq:infVel3}.
\end{proof}
Therefore, minimizers exist for every value
of the mass $m_0$ of the central body, but they may have collisions. In
the following, we shall state sufficient conditions to exclude partial and
total collisions, depending on the choice of the free-homotopy class.

\subsection{Representation of the free-homotopy classes}
\label{ss:encoding}
We use two different representations of the free-homotopy classes of
$\R^3\setminus\Gamma$.
\begin{enumerate}
   \item Let $\widetilde{\mathcal{R}}$ be the full symmetry group
     associated to $\mathcal{R}$, including reflections. These
     reflections induce a tessellation $\mathscr{T}_{\cal R}$ of the unit sphere
     $\unitsphere$, composed by $2N$ spherical triangles (see
     Figure~\ref{fig:encoding}, left), whose vertexes correspond to
     the set of poles $\mathcal{P} = \Gamma \cap \unitsphere$.  A
     free-homotopy class of $\R^3 \setminus \Gamma$ is described by a
     periodic sequence $\mathfrak{t}=\{\tau_k\}_{k\in\Z}$ of adjacent
     triangles, which is uniquely determined up to translations.
   \item We can also associate an Archimedean polyhedron
     $\mathcal{Q}_{\mathcal{R}}$ to $\mathcal{R}$, constructed in the 
     following way (see \cite{fgn10} for details). 
     Let $\tau$ be a triangle of the tessellation, and let $q$ be the middle point of the
     edge opposite to the right dihedral angle of $\tau$. Let $q_1, q_2$ be the points
     obtained by reflecting $q$ with respect to the planes passing through the two sides
     of $\tau$ not containing $q$ and the origin, and denote with $[q,q_{i}], i=1,2$ the segment connecting
     $q$ and $q_i, i=1,2$. Then the edges of $\mathcal{Q}_{\mathcal{R}}$ correspond
     to the union of the images of the two segments through the rotations in
     $\mathcal{R}$, i.e. $\{ R[q, q_1], R[q,q_2]\}_{R \in \mathcal{R}}$ (see
     Figures~\ref{fig:encoding} and \ref{fig:3poly} for a visual representation of the
     polyhedra).
     The faces of this polyhedron are in 1-1
     correspondence with the poles $p \in \mathcal{P}$, so that each
     rotation axis $r$ passes through the center of two opposite
     faces.  A free-homotopy class of $\R^3\setminus\Gamma$ is
     described by a periodic sequence $\nu=\{\nu_k\}_{k\in\Z}$ of
     vertexes of ${\cal Q}_{\cal R}$ such that each segment
     $[\nu_k,\nu_{k+1}]$ is an edge of ${\cal Q}_{\cal R}$.  Also the
     sequence $\nu$ is uniquely determined up to translations, and it
     can be used to construct a piecewise linear loop $v_I^\nu$ of
     $\K$ (see Figure~\ref{fig:encoding}, right).
\end{enumerate}

\begin{figure}[!h]
\centerline{
\epsfig{figure=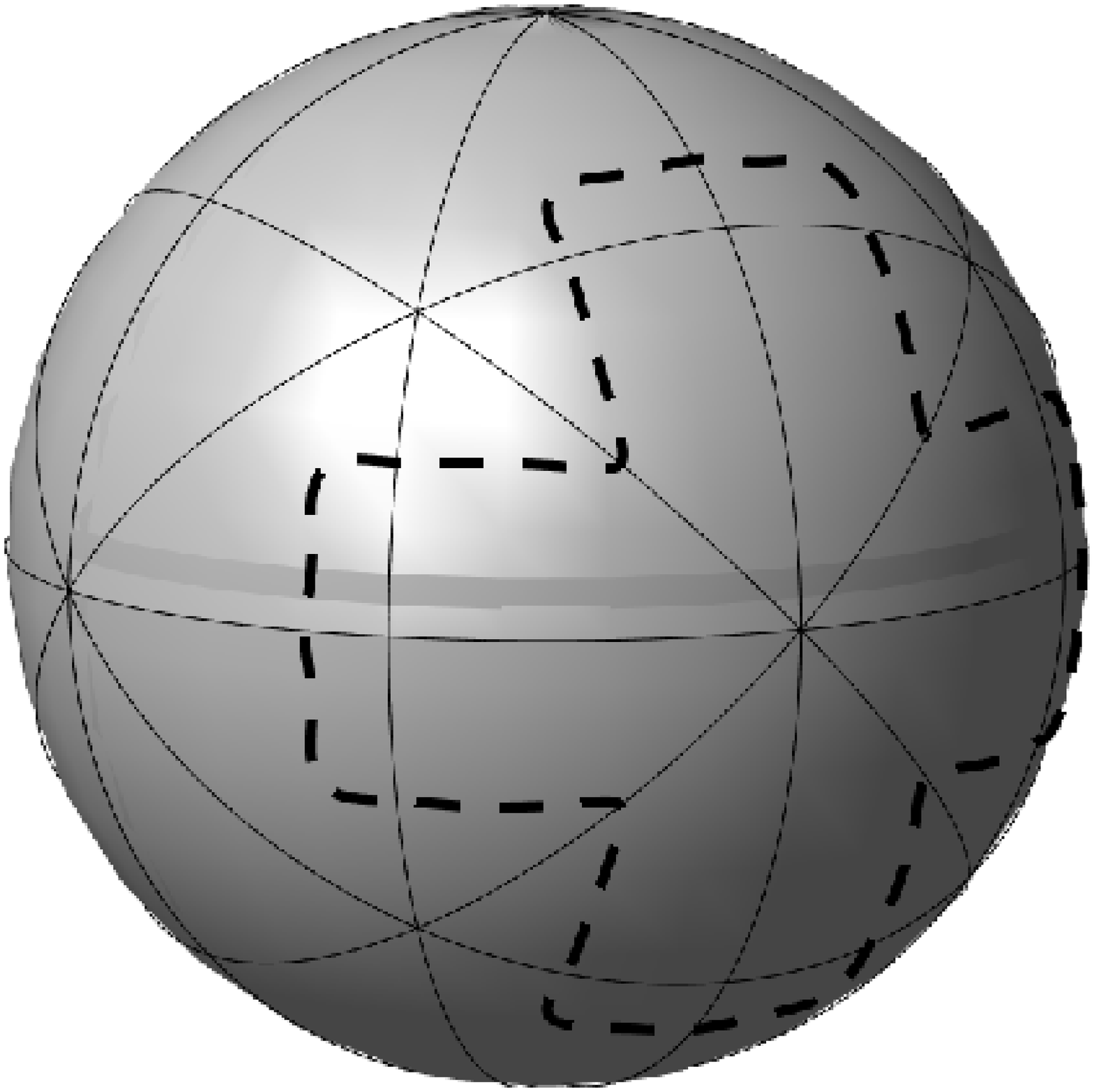,width=4.5cm}
\hskip 1.5cm
\epsfig{figure=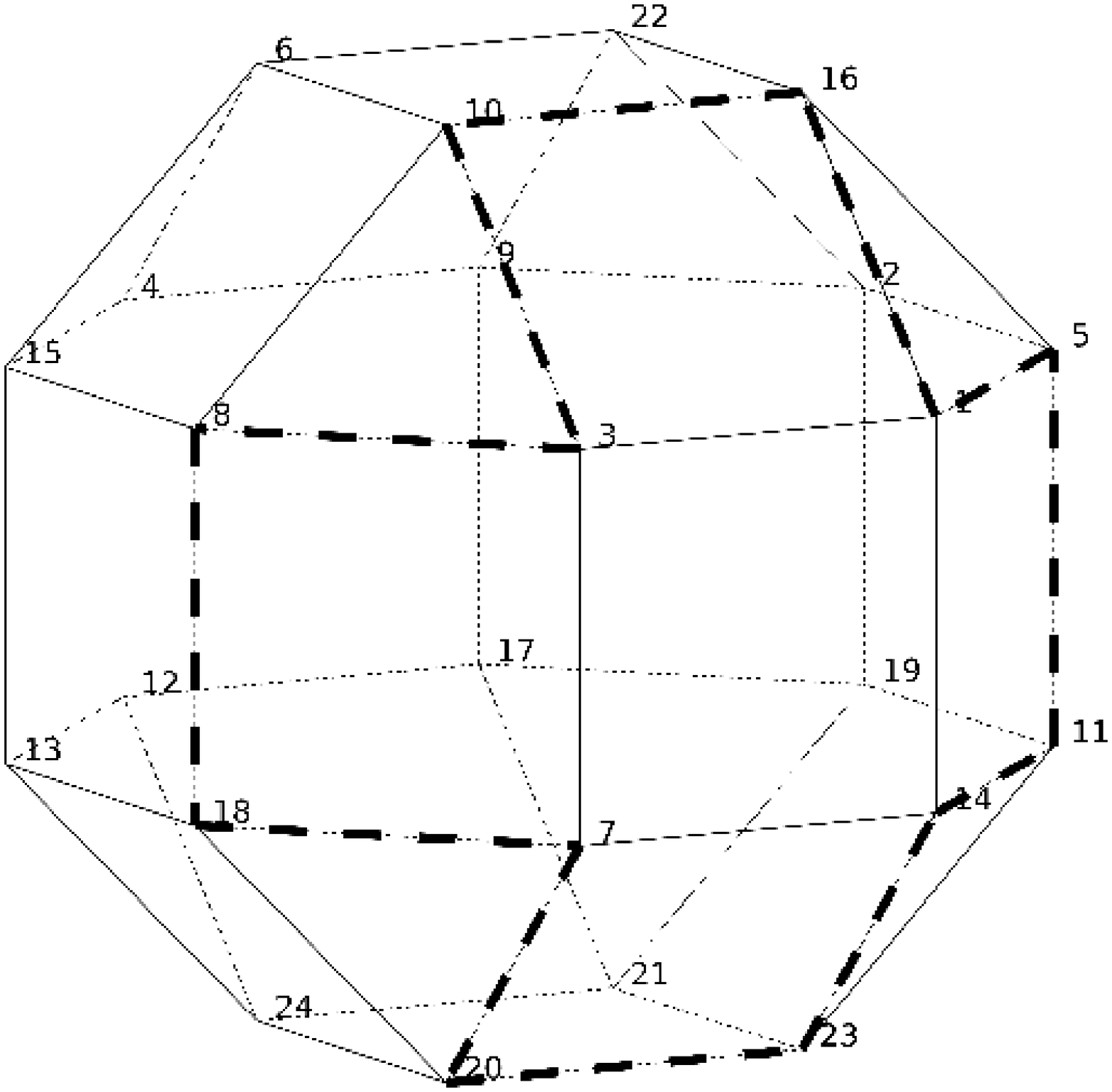,width=4.2cm}}
\caption{ The two ways of encoding a free-homotopy class of $\R^3
  \setminus \Gamma$, for $\mathcal{R}=\mathcal{O}$.  A possible
  trajectory of the generating particle (dashed path) is determined by
  a sequence $\mathfrak{t}$ of spherical triangles of the
  tessellation $\mathscr{T}_{\cal O}$ (on the left) or by a sequence $\nu$ of vertexes of the Archimedean polyhedron $\mathcal{Q}_{\mathcal{O}}$ (on the right).}
\label{fig:encoding}
\end{figure}
Therefore, to select a cone $\K$ we can use either a sequence
$\mathfrak{t}$ of triangles of $\mathscr{T}_{\cal R}$ or a sequence $\nu$ of vertexes of $\mathcal{Q}_{\mathcal{R}}$.  For
later use we introduce the following definitions.
\begin{definition}
   We say that a cone $\K$ is $\alpha$-\textit{simple} if the
   corresponding sequence $\mathfrak{t}$ does not contain a string
   $\tau_k\ldots \tau_{k+2[\frac{1}{2-\alpha}]\mathfrak{o}_p}$ such that
\[
   \bigcap_{j=0}^{2[\frac{1}{2-\alpha}]\mathfrak{o}_p}\overline{\tau_{k+j}} = p,
\]
where $p\in\cal P$, $\mathfrak{o}_p$ is the order of $p$ and $[\, \cdot \,]$ denotes the integer part of a real number.
\label{def:simpCone}
\end{definition}

\begin{definition}
   We say that a cone $\K$ is \textit{tied to two coboundary axes} if
  \begin{itemize}
  \item[i)] there exist two different poles $p_1, p_2$ such that the
    sequence $\mathfrak{t}$ is the union of strings
    $\sigma_i$ of the
    form $\tau_{k_i+1}\ldots \tau_{k_i+2 n_i\mathfrak{o}_j}$, with
    $n_i\in\N$ and $j\in\{1,2\}$, and
    \[
    \bigcap_{h=1}^{2 n_i\mathfrak{o}_j}\overline{\tau_{k_i+h}} = p_j,
    \]
    where $\mathfrak{o}_j$ is the order of $p_j$;
  \item[ii)] there exists $\tau_k\in\mathfrak{t}$ such that $p_1, p_2\in
    \overline{\tau_k}$.
  \end{itemize}
  \label{def:2ax}
\end{definition}
Therefore, a cone $\K$ is $\alpha$-simple if the trajectories of its loops do not wind around any rotation 
axis with an angle exceeding the maximum deflection angle of the potential $1/r^\alpha$. 
Moreover, a cone $\K$ is tied to two coboundary axes if the trajectories of its loops wind around two 
rotation axes only, that pass through two vertexes of the same spherical triangle, 
see Figure~\ref{fig:around2ax}.
\begin{figure}[!ht]
   \centering
   \includegraphics[scale=0.35]{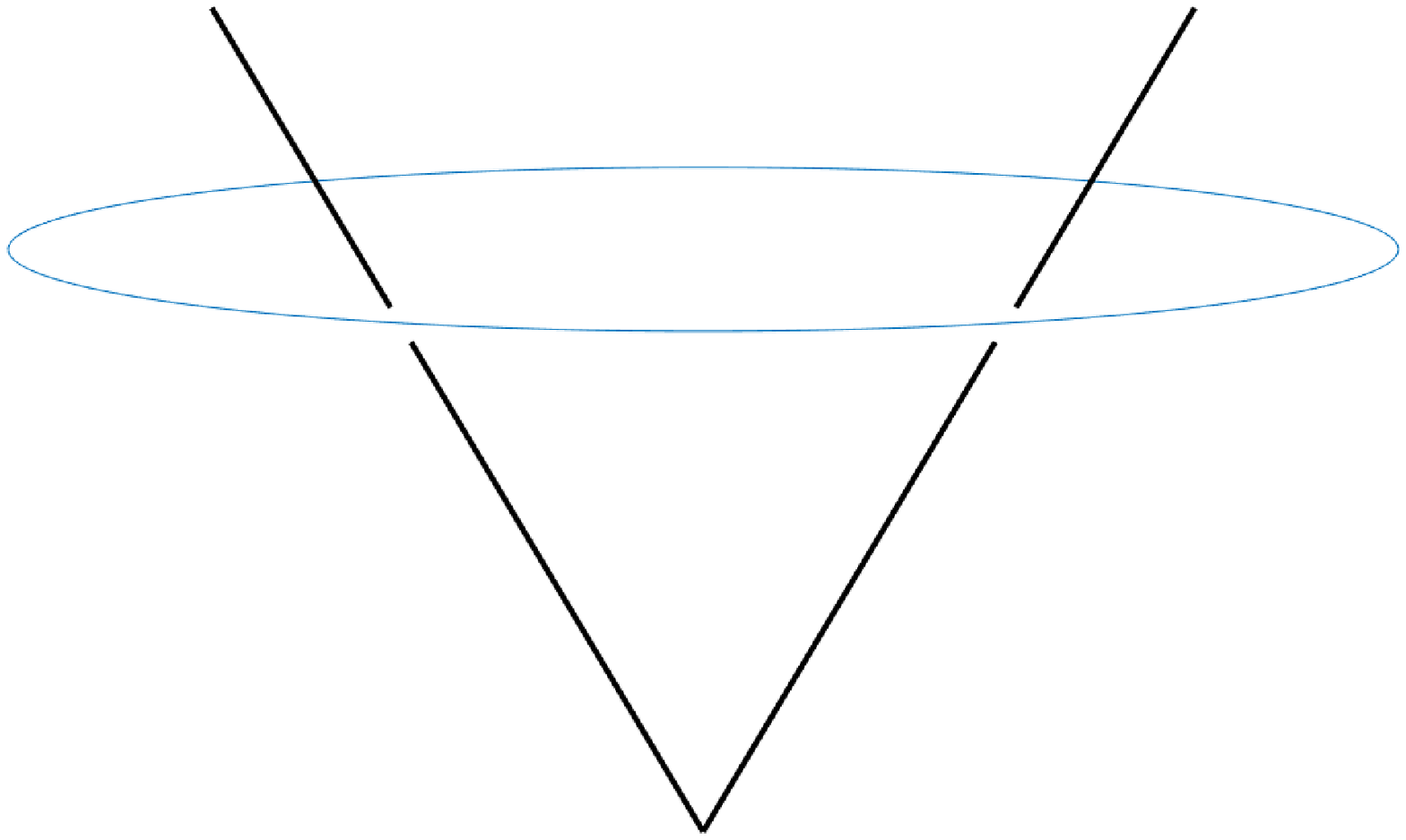}
   \hskip 1cm
   \includegraphics[scale=0.35]{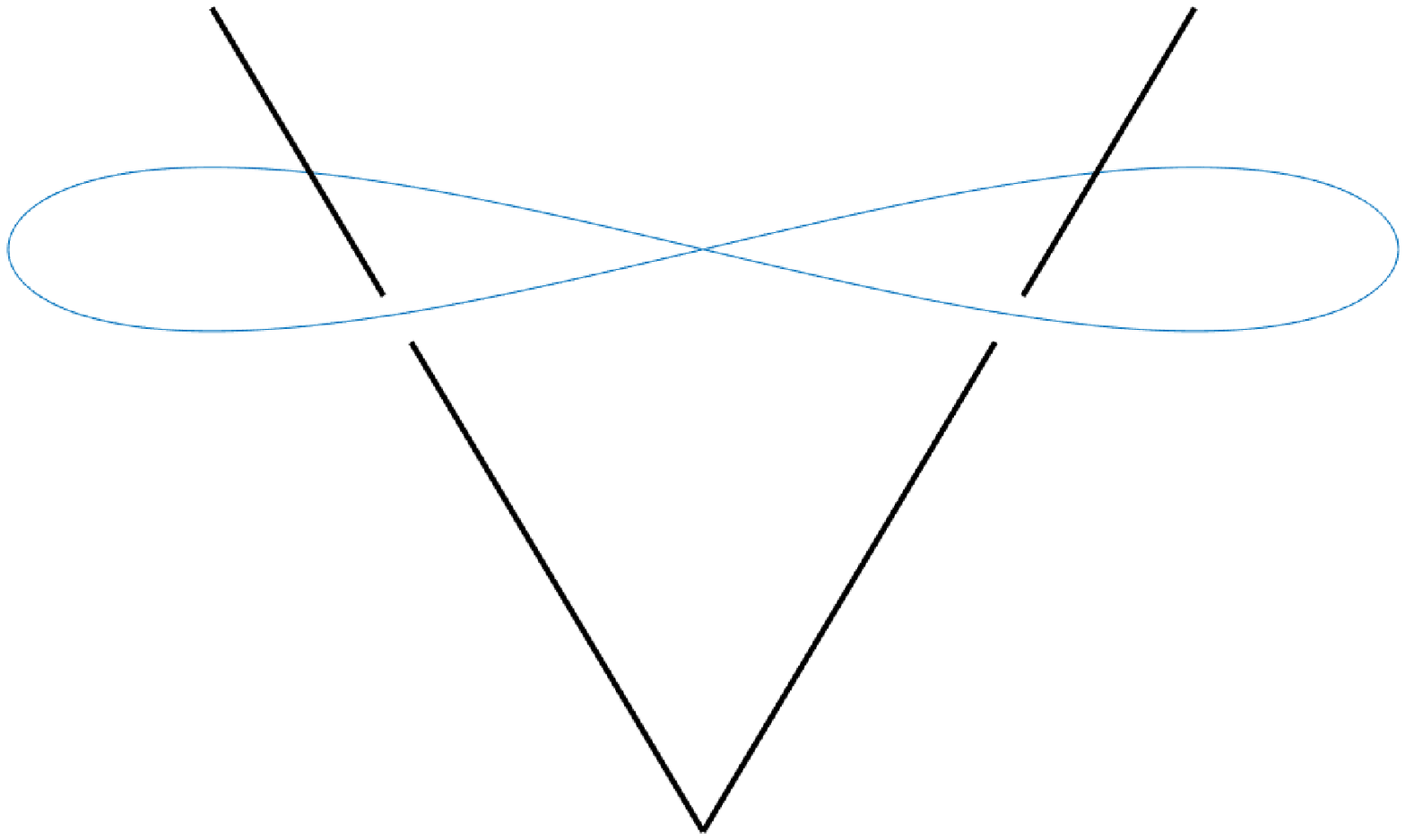}
   \caption{Sketch of two loops belonging to two different cones tied to two
     coboundary axes. The black half-lines correspond to the directions identified by the
     two poles $p_1, p_2$ belonging to the same triangle.}
   \label{fig:around2ax}
\end{figure}

\begin{definition}
We say that a cone $\K$ is \textit{central} if it contains a
loop lying in a plane passing through the origin $O$.
\end{definition}

To show that for a suitable choice of $\K$ the minimizers are
collision-free we consider total and partial collisions separately.


\subsection{Total collisions}
\label{ss:platoTotColl}
Also in this case we use level estimates to exclude total collisions.
In Appendix~\ref{app:LE} we show a general estimate useful for this
purpose, valid for $\alpha\in[1,2)$.

\subsubsection{A priori estimates}
\begin{proposition}
Let $u_I^* \in \overline{\K}$ be a solution with $M$ total collisions per period. Then the action satisfies
\begin{equation}
      \A^\alpha(u_I^*)\geq 
      \frac{2+\alpha}{2-\alpha}\frac{N}{2}\big(\tilde{U}_{\alpha,0}+2m_0\big)^{\frac{2}{2+\alpha}}
      \bigg(\frac{M\pi}{\alpha}\bigg)^{\frac{2\alpha}{2+\alpha}}T^{\frac{2-\alpha}{2+\alpha}},
    \label{eq:totCollEstAlpha}
\end{equation}
where 
\begin{equation}
      \tilde{U}_{\alpha,0} = \frac{1}{2^{\alpha+1}}\sum_{p \in \mathcal{P}}\frac{k_{\alpha,p}}{\displaystyle\max_{u_I \in \overline{\tau}}|u_I \times p|^\alpha}, \qquad    k_{\alpha,p} = \sum_{j=1}^{\mathfrak{o}_p-1}\frac{1}{\sin^\alpha(\frac{j\pi
   }{\mathfrak{o}_p})},
      \label{eq:U0tilde}
\end{equation}
and $\tau$ is any triangle of the tessellation of the sphere $\unitsphere$.  
\label{prop:platoLE}
\end{proposition}
\begin{proof}
  
The total mass is $\mathcal{M} = N+m_0$ and, by the symmetry
constraint, the potential $\mathcal{U}_\alpha$ defined in
Proposition~\ref{prop:lower_bound_gen} can be written as
\[
   \mathcal{U}_\alpha(u) = \frac{N}{N+m_0}\bigg( \sum_{R \in
      \mathcal{R}\setminus\{I\}}\frac{1}{|(R-I)u_I|^\alpha} +
      \frac{2 m_0}{|u_I|^\alpha}\bigg).
\]
We can also write 
\begin{equation}
   \mathcal{U}_\alpha(u) = \frac{N}{N+m_0}\bigg(\frac{1}{2^{\alpha+1}} \sum_{p \in
      \mathcal{P}}\frac{k_{\alpha,p}}{|u_I\times p|^\alpha} +
      \frac{2m_0}{|u_I|^\alpha}\bigg),
   \label{eq:potentialAlphaKappa}
\end{equation}
where the sum is done over the set of poles $\mathcal{P}$.
This equality follows from the fact that, for a rotation $R$ of an angle $2j\pi/\mathfrak{o}_p$ around an axis identified by pole $p \in \unitsphere$, we have
\[
\lvert (R-I)u_I \rvert^\alpha = 2^\alpha \lvert u_I \times p \rvert^\alpha \sin^\alpha \bigg( \frac{j \pi}{\mathfrak{o}_p} \bigg).
\]
Note that the order of the pole $p$ corresponds to the order of the rotation $R$.
Moreover, the additional 1/2 factor in \eqref{eq:potentialAlphaKappa} comes from the fact that there are two poles associated to each axis.
Using the definition of $\rho(u)$ given in Proposition~\ref{prop:lower_bound_gen}, by the symmetry we also have
\[
   \rho(u) =\sqrt{\frac{N}{N+m_0}}|u_I|,
\]
hence $\rho(u)=1$ if and only if
$|u_I|=\sqrt{\frac{N+m_0}{N}}$. Therefore, restricting to $\rho(u)=1$
and using the fact that $\mathcal{U}_\alpha$ is an $\alpha$-homogeneous
function, we have
\[
\begin{split}
   U_{\alpha,0} & := \min_{\rho(u)=1}\mathcal{U}_\alpha(u) \\
   & = \min_{|u_I|=\sqrt{\frac{N+m_0}{N}}}\mathcal{U}_\alpha(u) \\
   & = \bigg(\frac{N}{N+m_0}\bigg)^{\frac{\alpha}{2}} \min_{|u_I|=1}\mathcal{U}_\alpha(u) \\
   & = \bigg(\frac{N}{N+m_0}\bigg)^{\frac{2+\alpha}{2}}
   \bigg( \frac{1}{2^{\alpha+1}} \min_{|u_I|=1} \sum_{p \in
      \mathcal{P}}\frac{k_{\alpha,p}}{|u_I\times p|^\alpha} + 2m_0 \bigg) \\
   & = \bigg(\frac{N}{N+m_0}\bigg)^{\frac{2+\alpha}{2}}
   \bigg( \frac{1}{2^{\alpha+1}} \min_{u_I\in \overline{\tau}} \sum_{p \in
      \mathcal{P}}\frac{k_{\alpha,p}}{|u_I\times p|^\alpha} + 2m_0 \bigg),
\end{split}
\]
where $\tau$ is any triangle of the tessellation of $\mathbb{S}^2$. 
The last equality follows from the invariance of $\mathcal{U}_\alpha$ with respect to the full symmetry group $\widetilde{\mathcal{R}}$.
Indeed, given a triangle $\tau$ and $\widetilde{R} \in \widetilde{\mathcal{R}}$, we have 
\[
\begin{split}
    \min_{u_I \in \widetilde{R}\bar{\tau}}  \sum_{p \in \mathcal{P}} \frac{k_{\alpha,p}}{|u_I\times p|^\alpha} & =
    \min_{u_I \in \bar{\tau}}  \sum_{p \in \mathcal{P}} \frac{k_{\alpha,p}}{|\widetilde{R}u_I\times p|^\alpha}\\ 
    &= \min_{u_I \in \bar{\tau}}  \sum_{p \in \mathcal{P}} \frac{k_{\alpha,\widetilde{R} p}}{|\widetilde{R}u_I\times \widetilde{R} p|^\alpha} \\
    & = \min_{u_I \in \bar{\tau}}  \sum_{p \in \mathcal{P}} \frac{k_{\alpha,p}}{|u_I\times p|^\alpha}.
\end{split}
\]
We obtain that
\begin{equation}
   U_{\alpha,0} \geq
   \bigg(\frac{N}{N+m_0}\bigg)^{\frac{2+\alpha}{2}}\big(\tilde{U}_{\alpha,0} + 2m_0\big),
   \label{eq:Ualpha0geq}
\end{equation}
where
\[
   \tilde{U}_{\alpha,0} = \frac{1}{2^{\alpha+1}}\sum_{p \in \mathcal{P}}\frac{k_{\alpha,p}}{\displaystyle\max_{u_I \in \overline{\tau}}|u_I \times p|^\alpha}.
\]
Consider now a solution $u_I^* \in \overline{\K}$ with $M$ total
collisions per period.  From Corollary~\ref{cor:LE}
and from the above computations we have
\[
   \begin{split}
      \A^\alpha(u_I^*) & >
      \frac{2+\alpha}{2-\alpha}\frac{\big(N+m_0\big)}{2}U_{\alpha,0}^{\frac{2}{2+\alpha}}
      \bigg(\frac{M\pi}{\alpha}\bigg)^{\frac{2\alpha}{2+\alpha}}T^{\frac{2-\alpha}{2+\alpha}}
      \\
      & \geq 
      \frac{2+\alpha}{2-\alpha}\frac{N}{2}\big(\tilde{U}_{\alpha,0}+2m_0\big)^{\frac{2}{2+\alpha}}
      \bigg(\frac{M\pi}{\alpha}\bigg)^{\frac{2\alpha}{2+\alpha}}T^{\frac{2-\alpha}{2+\alpha}}.
   \end{split}
   \]
\end{proof}
\begin{remark}
Setting 
\[
\mathfrak{p}_{\max} = \max_{j \in \{1,2,3\}} p_j\cdot p,  \qquad
\mathfrak{p}_{\min} = \min_{j \in \{1,2,3\}} p_j\cdot p,
\]
where $p_j \in \mathcal{P}, \, j=1,2,3$ are the vertexes of the triangle $\tau$,
we have
\[
   \max_{u_I \in \overline{\tau}} |u_I\times p | = 
   \left\{
   \begin{array}{cl}
   \displaystyle \max_{j \in \{1,2,3\}}| p_j \times p| &\mbox{if } \mathfrak{p}_{\min}>0 \mbox{ or } \mathfrak{p}_{\max} <0, \cr
   1                                                   &\mbox{if } \mathfrak{p}_{\min}< 0 < \mathfrak{p}_{\max}. \cr
   \end{array}
\right.
\]
\end{remark}
%



\subsubsection{Constructing test loops}
Let $\nu$ be the periodic sequence of vertexes of
$\mathcal{Q}_{\mathcal{R}}$, used to select the free-homotopy class
in condition (b), and let $v_I^\nu$ be the linear piecewise loop
defined by $\nu$, travelling along the edges of the Archimedean
polyhedron $\mathcal{Q}_{\mathcal{R}}$ with constant speed.
\begin{definition}
  Two sides of ${\cal Q}_{\cal R}$ have
  different type if they belong to the boundary of different pairs of
  regular polygons. In particular,
  \begin{itemize}
      \item[-] for $\mathcal{R} = \mathcal{T}$ we have sides with only one type, separating a triangle and a square;
      \item[-] for $\mathcal{R} = \mathcal{O}$ we say that a side has type 1 if it separates a triangle and a square, while it has type 2 if it separates two squares; 
      \item[-] for $\mathcal{R} = \mathcal{I}$ we say that a side has type 1 if it separates a triangle and a square, while it has type 2 if it separates a square and a pentagon.
  \end{itemize}
\end{definition}
Figure~\ref{fig:3poly} shows the three Archimedean polyhedra for $\mathcal{R} = \mathcal{T},
\mathcal{O},\mathcal{I}$ and the corresponding sides of type 1 and 2.
\begin{figure}[!ht]
   \centering
   \includegraphics[width=0.3\textwidth]{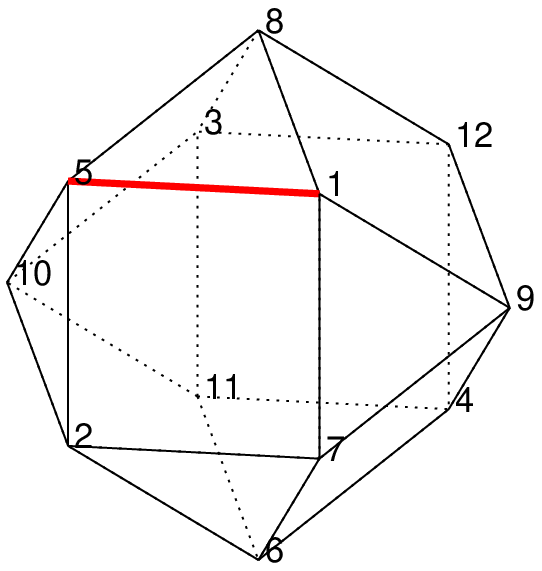}
   \includegraphics[width=0.3\textwidth]{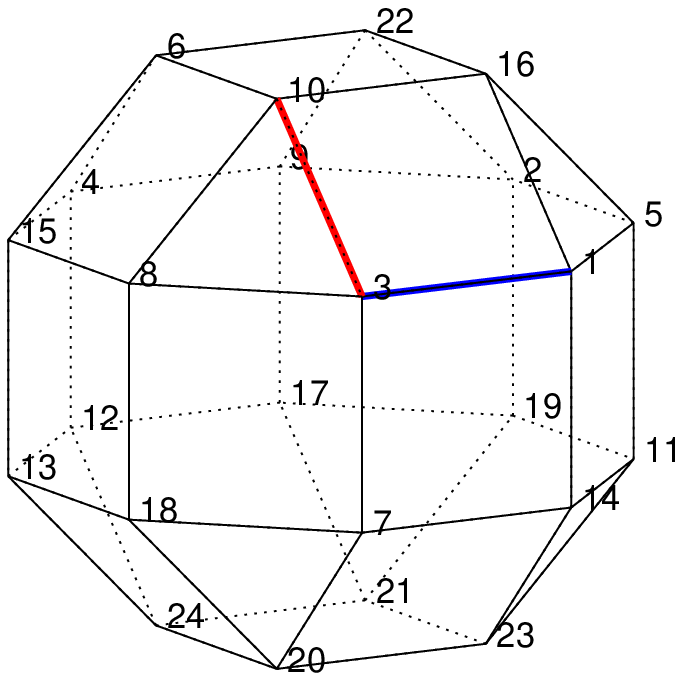}
   \includegraphics[width=0.3\textwidth]{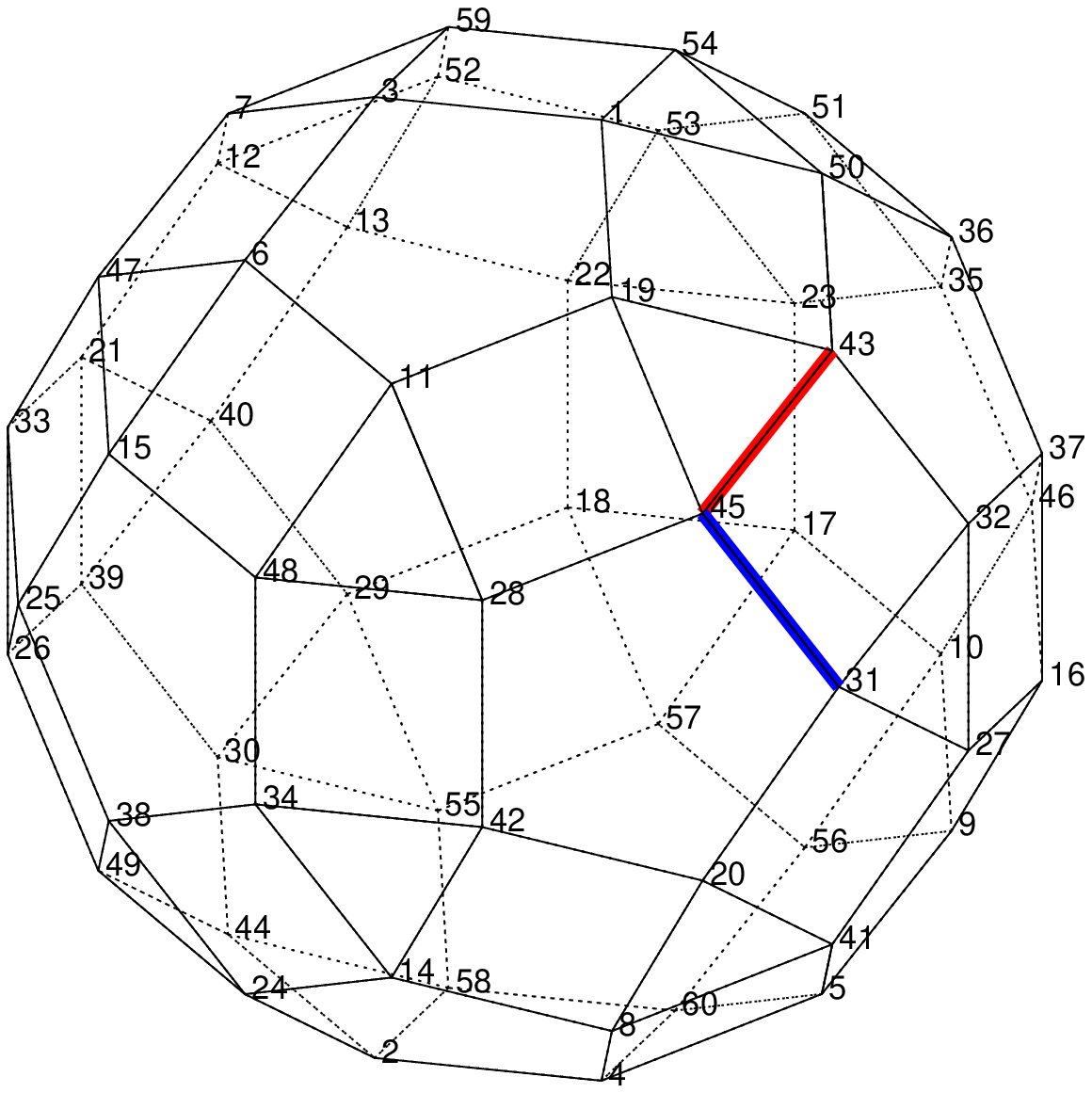}
   \caption{The three Archimedean polyhedra for $\mathcal{R} = \mathcal{T},
   \mathcal{O},\mathcal{I}$, from the left to the right. For
   $\mathcal{R}=\mathcal{T}$ there is only one type of side, while for
   $\mathcal{R}=\mathcal{O}, \mathcal{I}$ the red side is of type 1 and the blue side is
of type 2.}
   \label{fig:3poly}
\end{figure}

\begin{proposition}
The minimum value of the action $\A^\alpha$ over the
1-parameter family of rescaled loops $\{\lambda v_I^\nu\}_{\lambda>0}$ is
\small
\begin{equation}
\A^\alpha(\bar{\lambda}v_I^\nu) = \frac{2+\alpha}{2}N\left[\frac{\ell^{2\alpha}\big(k_1\zeta_{\alpha,1} +
   k_2\zeta_{\alpha,2}+m_0(k_1+k_2)\zeta_{\alpha,0}\big)^2
k_\nu^{2(\alpha-1)}}{4\alpha^\alpha}\right]^{\frac{1}{2+\alpha}}T^{\frac{2-\alpha}{2+\alpha}},
\label{eq:testLoopAction}
\end{equation}
\normalsize
for a suitable $\bar\lambda$. Here $\ell$ is the length of the edges of the Archimedean polyhedron $\mathcal{Q}_{\mathcal{R}}$ inscribed in the unit sphere $\unitsphere$, while
\[
   \zeta_{\alpha,i} = \int_0^1\sum_{R\in{\mathcal{R}}\setminus\{I\}}
   \frac{ds}{\bigl|(R-I)[(1-s)q+s q_i]\bigr|^\alpha}, \qquad i=1,2,
\]
and
\[
   \begin{split}
      \zeta_{\alpha,0}& = \int_{0}^{1}\frac{2}{|(1-s)q+sq_1|^\alpha}ds
                       = \int_{0}^{1}\frac{2}{|(1-s)q+sq_2|^\alpha}ds. 
   \end{split}
\]
The points $q, q_1, q_2 \in \unitsphere$ are the ones used for the construction of
$\mathcal{Q}_{\mathcal{R}}$ introduced in Section~\ref{ss:encoding}, $k_1, k_2$ correspond to the
number of sides $[\nu_{j-1}, \nu_j]$ of type 1 and 2 in the sequence
$\nu$, respectively,
  and $k_\nu$ is the minimal period of $\nu$. 
  \label{prop:actionTest}
\end{proposition}
\begin{proof}
Let us define
\begin{equation}
    \begin{split}
        A_{U_\alpha} & = N \int_0^T \bigg( \frac{m_0}{\lvert v_I^\nu \rvert^\alpha} + \frac{1}{2}\sum_{R \in \mathcal{R}\setminus\{I\}}\frac{1}{\lvert (R-I)v_I^\nu \rvert^\alpha} \bigg) dt, \\
        A_K & = N\int_0^T \frac{\lvert \dot{v}_I^\nu \rvert^2}{2} dt = \frac{N}{2} \frac{\ell^2 k_\nu^2}{T},  
    \end{split}    
    \label{eq:AUAK}
\end{equation}
where in the last equality we used $\lvert \dot{v}_I^\nu \rvert = \ell k_\nu/T$. From the $\alpha$-homogeneity 
of the potential, the action of a rescaled loop $\lambda v_I^\nu, \, \lambda > 0$ is
\begin{equation}
    \A^\alpha(\lambda v_I^\nu) = \lambda^2 A_K + \frac{1}{\lambda^\alpha}A_{U_\alpha},
    \label{eq:rescActFormula}
\end{equation}
and the minimum is achieved for 
\begin{equation}
    \bar{\lambda} = \bigg( \frac{\alpha}{2}\frac{A_{U_\alpha}}{A_K} \bigg)^{\frac{1}{\alpha+2}}.
    \label{eq:lambdaBar}
\end{equation}
Substituting \eqref{eq:lambdaBar} in \eqref{eq:rescActFormula}, the minimum action of the rescaled loops is
\begin{equation}
    \A^\alpha(\bar{\lambda} v_I^\nu) = (2+\alpha) 
    \bigg[ \bigg(\frac{A_{U_\alpha}}{2}\bigg)^2 \bigg(\frac{A_K}{\alpha}\bigg)^{\alpha} \bigg]^{\frac{1}{\alpha+2}}.
    \label{eq:actionLambdaBar}
\end{equation}
The term $A_{U_\alpha}$ can be expressed through the quantities $\zeta_{\alpha, i}, \, i = 0,1,2$. 
%
%
Indeed, from the construction of $\mathcal{Q}_{\mathcal{R}}$ it follows that we can associate to
each $j \in \Z$ a uniquely determined pair $(R_j,i_j)\in{\mathcal{R}}\times\{1,2\}$ 
such that $[\nu_{j-1}, \nu_j] = R_j[q, q_{i_j}]$. For each given $R'\in{\mathcal{R}}$, set
\[
   \zeta_{\alpha,i}(R') = \int_0^1\sum_{R\in{\mathcal{R}}\setminus\{I\}}
   \frac{ds}{\bigl|(R-I)R'[(1-s)q+s q_i]\bigr|^\alpha}, \qquad i=1,2.
\]
Since
\[
\bigl|(R-I)R'[(1-s)q+s q_i]\bigr| = \bigl|((R')^{-1}RR' -I)[(1-s)q+s
q_i]\bigr|,
\]
and the map $R\mapsto (R')^{-1}RR'$ is an isomorphism of $\mathcal{R}$ onto
itself, we have
\[
\zeta_{\alpha,i}(R') = \zeta_{\alpha,i}(I) = \zeta_{\alpha,i},
\quad i=1,2.
\]
%
Since $v_I^\nu$ travels along each side $[\nu_{j-1}, \nu_j]$ in a time interval of size 
$T/k_\nu$, it follows that
\begin{equation}
   \begin{split}
A_{U_\alpha} &= 
\frac{N}{2}\frac{T}{k_\nu}\sum_{j=1}^{k_\nu}
\int_0^1\bigg(\sum_{R\in\mathcal{R}\setminus\{I\}}
\frac{1}{\bigl|(R-I)R_j[(1-s)q+s q_{i_j}]\bigr|^\alpha} +
\frac{2m_0}{|(1-s)q+sq_{i_j}|^\alpha}\bigg)ds
\\
&=\frac{N}{2}\frac{T}{k_\nu}\big(k_1\zeta_{\alpha,1} + k_2\zeta_{\alpha,2}+
m_0(k_1+k_2)\zeta_{\alpha,0}\big).
\end{split}
\label{eq:AUalpha}
\end{equation}
%
From \eqref{eq:AUAK}, \eqref{eq:actionLambdaBar}, and \eqref{eq:AUalpha} we finally obtain relation \eqref{eq:testLoopAction}.
\end{proof}

\noindent For the Keplerian case, the integrals defining
$\zeta_{1,0}, \zeta_{1,1}, \zeta_{1,2}$ can be expressed by
elementary functions, and their values are provided in Table
\ref{tab:zetaAlphaEst}.  For the case $\alpha \in
(1,2)$, we can estimate the values $\zeta_{\alpha,i}, \, i=0,1,2$
using the values for the Keplerian case. 
\begin{proposition}
    In the same hypotheses of Proposition~\ref{prop:actionTest}, for every $\alpha\in(1,2)$ we have 
    \begin{equation}
    \A^\alpha(\bar{\lambda}v_I^\nu) <
    \frac{2+\alpha}{2}N\Biggl[\frac{\ell^{2\alpha}\bigg(k_1\displaystyle\frac{\zeta_{1,1}}{\delta_1} +
       k_2\frac{\zeta_{1,2}}{\delta_2}+\frac{8m_0}{4-\ell^2}(k_1+k_2)\bigg)^2
    k_\nu^{2(\alpha-1)}}{4\alpha^\alpha}\Biggr]^{\frac{1}{2+\alpha}}T^{\frac{2-\alpha}{2+\alpha}},
    \label{eq:testLoopActionEst}
    \end{equation}
    where $\delta_1, \delta_2 \in \R$ are given in Table~\ref{tab:zetaAlphaEst}.
\end{proposition}
\begin{proof}
Let
\[
   d_i(R) = \frac{1}{2}\big| (R-I)(q+q_i) \big|, \quad i=1,2
\]
be the minimal distance between the particle $u_I$ travelling along the segment $[q,q_i]$, 
and $u_R = R u_I$ travelling along $R[q,q_i]$. 
In fact the function 
\[
    f(s) = \big\lvert (R-I)[(1-s)q+sq_i]\big\rvert^2, \quad s \in [0, 1],
\]
is strictly convex and the relation $f(1-s) = f(s)$ holds, thus $\min_{[0,1]} f = f(1/2)$. 
For $i=1,2$ we introduce the minimal distance
\begin{equation}
   \delta_i = \min_{R \in \mathcal{R} \setminus \{ I \}} d_i(R), 
   \label{eq:minDistDeltai}
\end{equation}
whose values are reported in Table
\ref{tab:zetaAlphaEst}. We can estimate $\zeta_{\alpha,i}$
using the relations
\begin{equation}
   \zeta_{\alpha, i} < \frac{\zeta_{1,i}}{\delta_{i}^{\alpha-1}} <
   \frac{\zeta_{1,i}}{\delta_i}, 
   \label{eq:zetaAlphaEst}
\end{equation}
where the last inequality follows from $\delta_i < 1$.
Moreover, we have
\begin{equation}
   \zeta_{\alpha,0} < \frac{2}{\big(1-\frac{\ell^2}{4}\big)^{\frac{\alpha}{2}}} <
   \frac{8}{4-\ell^2},
   \label{eq:zetaAlpha0Est}
\end{equation}
where for the first inequality we used the relation
\[
\big|(1-s)q + sq_1\big| \geq \frac{|q+q_1|}{2} = \bigg( 1 - \frac{\ell^2}{4}\bigg)^{1/2}, \qquad s\in [0,1].
\]
Therefore, using \eqref{eq:testLoopAction}, \eqref{eq:zetaAlphaEst} and \eqref{eq:zetaAlpha0Est}, for $\alpha\in(1,2)$ we obtain the estimate \eqref{eq:testLoopActionEst}.
\end{proof}

{\renewcommand{\arraystretch}{1.2}%
\begin{table}[!ht]
   \centering
   \begin{tabular}{cccc}
      \toprule
      & $\mathcal{T}$ & $\mathcal{O}$ & $\mathcal{I}$ \\
      \midrule
      $\delta_1$                  & $0.35740$ & $\phantom{-}0.35740$ &  $\phantom{-}0.36230$ \\
      $\delta_2$                  & $0.35740$ & $\phantom{-}0.50544$ &  $\phantom{-}0.22391$ \\
      $\zeta_{1,0}$               & $2.19722$ & $\phantom{-}2.09234$ &  $\phantom{-}2.03446$ \\
      $\zeta_{1,1}$               & $9.50838$ & $20.32244$ &  $53.99031$ \\
      $\zeta_{1,2}$               & $9.50838$ & $19.73994$ &  $52.57615$ \\
      $\frac{8}{4-\ell^2}$        & $2.66666$ & $\phantom{-}2.29297$ &  $\phantom{-}2.10560$ \\
      \bottomrule
   \end{tabular}
   \caption{Rounded values of $\delta_1,\delta_2, \zeta_{1,0},
   \zeta_{1,1}, \zeta_{1,2}$ and $8/(4-\ell^2)$ for the three different rotation groups.}
   \label{tab:zetaAlphaEst}
\end{table}
}

%
\noindent Using relations \eqref{eq:totCollEstAlpha}, \eqref{eq:testLoopAction} and \eqref{eq:testLoopActionEst}, for
some free-homotopy classes $\nu$, we have
\begin{equation}
  \A^\alpha(\bar{\lambda}v_I^\nu) < \A^\alpha(u_I^*)
  \label{actionless}
\end{equation}
if $m_0$ is large enough, therefore minimizers of $\A^\alpha$ on the cones $\K$ generated
by those $\nu$ are free of total collisions. 
It is worth noting that both $\A^\alpha(\bar{\lambda}v_I^\nu)$ and $\A^\alpha(u_I^*)$ have order $O(m_0^{1/(2+\alpha)})$ as
$m_0\to+\infty$. Hence, to check that \eqref{actionless} holds, we have to compare the coefficients of $m_0$ in the right hand 
sides of \eqref{eq:totCollEstAlpha} and \eqref{eq:testLoopActionEst}.
Some practical examples, in which we verify the inequality \eqref{actionless} comparing the two asymptotic behaviours for $m_0 \to
\infty$, are given in Section~\ref{ss:platoExamples}.



\subsection{Partial collisions}
\label{ss:platoPartialColl}
Partial collisions can only take place on the rotation axes
$\Gamma\setminus\{0\}$. Let $u_I^* \in \overline{\K}$ be a collision
solution and $(t_1,t_2)$ an interval of regularity. Then it is a
solution of the Euler-Lagrange equation
\begin{equation}
   \ddot{w} = \alpha\sum_{R \in
      \mathcal{R}\setminus\{I\}}\frac{(R-I)w}{|(R-I)w|^{2+\alpha}}-\alpha m_0\frac{w}{|w|^{2+\alpha}}, \quad t \in (t_1,t_2). 
   \label{eq:eqMotionAlpha}
\end{equation}
Let $r$ be the rotation axis where the generating particle has a partial collision and let
$\mathcal{C}$ be the subgroup (of order $\mathfrak{o}_{\mathcal{C}}$) of the rotations in ${\cal R}$ 
with axis $r$. We can rewrite equation \eqref{eq:eqMotionAlpha} and the first integral of the energy in the form
\begin{gather}
   \ddot{w} = \alpha c_\alpha \frac{(R_\pi - \Id)w}{\lvert (R_\pi - \Id)w
      \rvert^{2+\alpha}}+V_1(w), \quad 
   c_\alpha = \sum_{j=1}^{\mathfrak{o}_\mathcal{C} - 1} \frac{1}{\sin^\alpha \big{(}
   \frac{j\pi}{\mathfrak{o}_{\mathcal{C}}}
   \big{)}}, \label{eq:newtEqPartAlpha} \\ 
   \lvert \dot{w} \rvert^2 - c_\alpha\frac{1}{\lvert (R_\pi-\Id)w \rvert^\alpha} - V(w) = h, \label{eq:energyPartAlpha}
\end{gather}
where $R_\pi$ is the rotation of $\pi$ around $r$, and $V_1(w)$, $V(w)$ are smooth
functions defined in an open set $\Omega \subseteq \R^3$ that contains $r \setminus \{ 0
\}$. Moreover, if $\tilde{R} \in \tilde{\mathcal{R}}$ is a reflection such that
$\tilde{R}r=r$, then $V_1,V$ satisfy the conditions
\begin{equation}
   V_1(\tilde{R}w) = \tilde{R}V_1(w), \quad V(\tilde{R}w) = V(w).
   \label{eq:symmCondPartCollAlpha}
\end{equation}
Therefore, partial collisions can be seen as binary collisions in a perturbed Kepler
problem and asymptotic collision and ejection directions $\mathtt{n}^{\pm}$, orthogonal to $r$, can be
defined, see Appendix \ref{app:collparz}. 
Partial collisions are excluded by the following theorem.
\begin{theorem}
   Assume that $\K$ is $\alpha$-simple and it is not tied to two coboundary axes. 
   Then any minimizer $u_I^* \in \overline{\K}$ of $\A^\alpha$ that has no total collisions, does not have partial collisions either. 
   \label{th:noPartialCollisions}
\end{theorem}
To prove Theorem~\ref{th:noPartialCollisions} we introduce the following definition. 
\begin{definition}
   With the notation above, a partial collision is said to be of type
   $(\rightrightarrows)$ if
   \begin{itemize}
      \item[(1)] $\mathtt{n}^+=\mathtt{n}^-$;
      \item[(2)] the plane generated by $r$ and $\mathtt{n} = \mathtt{n}^{\pm}$ is fixed by
         some reflection $\widetilde{R} \in \widetilde{\mathcal{R}}$.
   \end{itemize}
   \label{def:doppiaFreccia}
\end{definition}
\noindent We also introduce the {\em ejection-collision}
parabolic motion
\begin{equation}
\omega(\pm t) = \npm s^\alpha(t), \qquad t\geq 0,
\label{parab_eject_coll}
\end{equation}
with
\[
s^\alpha(t) =
\frac{(2+\alpha)^{2/(2+\alpha)}}{2}\calpha^{1/(2+\alpha)}
t^{2/(2+\alpha)}, \qquad t\in[0,+\infty).
\]
The proof of Theorem~\ref{th:noPartialCollisions} is divided in three steps: 1) if the cone is $\alpha$-simple, then partial collisions are necessarily of type $(\rightrightarrows)$; 2) if the collision if of type $(\rightrightarrows)$, then the trajectory of a minimizer is contained in a single reflection plane, and the generating particle bounces back and forth between two rotation axes; 3) if $\K$ is not tied to two coboundary axes, then the conclusion of point 2) is not possible. 

\begin{proposition}
    Assume that $\K$ is $\alpha$-simple, and let $u_I^* \in \overline{\K}$ be a minimizer of $\A^\alpha$
    with a partial collision at time $t_c$. Then the partial collision is of type $(\rightrightarrows)$.
    \label{prop:SoloDoppiaFreccia}
\end{proposition}
\begin{proof}

We sketch the proof of the exclusion of collisions that are not of type $(\rightrightarrows)$,
being the details already presented in \cite{fgn10}.
We can associate an angle $\theta$ to the minimizer $u_I^*$ at the
collision time $t_c$ (see Appendix~\ref{app:collang}). This angle is the same for all the loops in a
minimizing sequence $\{\hat{u}_I^{\ell}\}_{\ell\in\N}$ converging to $u_I^*$. It represents the angle
between the two asymptotic directions $\mathtt{n}^+, \mathtt{n}^-$
taking into account the (signed) number of revolutions of the
trajectories of the minimizing sequence converging to $u_I^*$ around
the collision axis $r$.
We discuss separately the cases $\alpha=1$ and $1<\alpha<2$.

\bigbreak
\noindent{\bf Keplerian case}
Assuming that $\K$ is $1$-simple, thanks to Proposition~\ref{thetabounds} we have that
\[
   -\frac{\pi}{\mathfrak{o}_{\mathcal{C}}} \leq \theta \leq 2\pi.
\]
If the collision is not of type $(\rightrightarrows)$, then we can always reduce the discussion to the
case $\theta\neq2\pi$.
In fact, if property (1) of Definition~\ref{def:doppiaFreccia}
holds, but not property (2), then there are two possible behaviors with $\theta=2\pi$, that we sketch in Figure~\ref{fig:thetanot2pi}. 
\begin{figure}[t]
    \centering
    \includegraphics[width=0.35\textwidth]{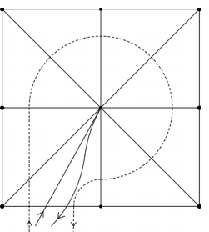}
    \hspace{1cm}
    \includegraphics[width=0.35\textwidth]{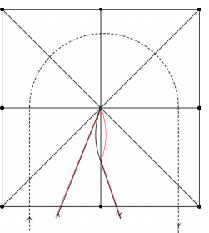}
    \caption{Two different behaviours when the collision angle $\theta$ is $2\pi$. The dashed curves indicate how the loops in the free-homotopy class wind around the colliding axis, while the black thick curves are the collision paths. On the left, the cone $\K$ is not $1$-simple. On the right, a portion of the path is reflected into a new colliding path (the red curve) whose collision angle is less than $2\pi$.}
    \label{fig:thetanot2pi}
\end{figure}
The case on the left can be excluded using the fact that $\K$ is 1-simple. For the case on the right, we observe that we can reflect part of the trajectories of $\hat{u}_I^\ell$, $u_I^*$ without changing the values $\A^1(u_I^*)$, $\A^1(u_I^\ell)$ of the action. In this way we would obtain another minimizer with collision axis $r$, but with $\theta\neq 2\pi$.  

Since $\theta\neq 2\pi$, we can exclude partial collisions by local
perturbations. Indeed, using the blow-up technique \cite{FT2004}, the collision solution
$u_I^*$ is asymptotic to the parabolic collision-ejection solution $\omega$
\eqref{parab_eject_coll}. Then, a local perturbation $\tilde{u}_I^*$ without partial
collisions can be constructed using either the direct or the indirect Keplerian arc (see Figure~
\ref{fig:surgery} for a sketch), and by Lemma \ref{lemma:marchal} we can prove that the
action of $\tilde{u}_I^*$ is lower than the action of $u_I^*$.
\begin{figure}[!ht]
   \centering
   \includegraphics[scale=1]{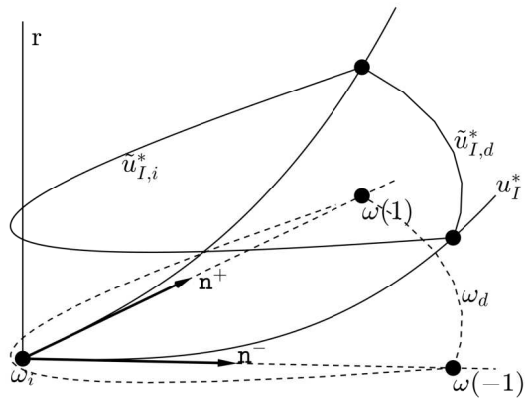}
   \caption{The sketch of the construction of collision-free perturbations. The arcs
   $\omega_d, \omega_i$ are the direct and indirect Keplerian arcs, respectively,
   connecting the two points $\omega(\pm 1)$ of the parabolic collision-ejection solution
   at $t=\pm1$. The two perturbations of $u_I^*$, where the collision is removed, are
   denoted with $\tilde{u}_{I,d}^*$ and $\tilde{u}_{I,i}^*$.}
   \label{fig:surgery}
\end{figure}


\bigbreak
\noindent{\bf Non-Keplerian case}
The blow-up technique can be also used for $\alpha$-homogeneous potentials, with $\alpha \in (1,2)$
(see \cite{FT2004}). Again, we can show that a solution with partial collisions
is asymptotic to the
parabolic collision-ejection solution $\omega$ defined in \eqref{parab_eject_coll}.
The generalization of Marchal's Lemma to the case $\alpha \in (1,2)$
can be deduced from the existent literature, see Lemma
\ref{lemma:marchal} in Appendix \ref{app:marchal}.
%
Assuming that the cone $\K$ is $\alpha$-simple, we have
\[
   -\frac{\pi}{\mathfrak{o}_{\mathcal{C}}} \leq \theta \leq 2\pi \bigg[ \frac{1}{2-\alpha}
   \bigg],
\]
where $[\, \cdot \,]$ denotes the integer part. Let $\bar{\theta} \in
[0, 2\pi)$ and $h\in \Z$ be such that 
\[
   \theta = \bar{\theta} + 2  \pi h, \qquad h = [\theta/(2\pi)].
\]
Note that $-1 \leq h \leq [1/(2-\alpha)]$.
Let us suppose that the collision is not of type ($\rightrightarrows$), hence we can
always assume (see the discussion for the Keplerian case)
\[
   \theta \neq 2\pi\bigg[\frac{1}{2-\alpha}\bigg],
\]
from which it follows that $h < [1/(2-\alpha)]$.
From Lemma \ref{lemma:connectingArcs}, the arcs which are always available, i.e.
independently from the value of $\bar{\theta} \in [0, 2\pi)$, sweep a total angle of
$\bar{\theta}+2\pi k$, where
\[
   \max_{\varphi \in (0,2\pi)}k_{\text{min}} \leq k \leq \min_{\varphi \in
   (0,2\pi)}k_{\text{max}},
\]
and $k_{\text{min}}, k_{\text{max}}$ are given by \eqref{eq:kappaMinApp},
\eqref{eq:kappaMaxApp}, respectively.
Note that
\[
   \max_{\varphi \in (0,2\pi)}k_{\text{min}} = -\bigg[\frac{1}{2-\alpha}\bigg], \quad
   \min_{\varphi \in (0,2\pi)}k_{\text{max}} = \bigg[\frac{1}{2-\alpha}\bigg]-1,
\]
hence it is always possible to find an arc sweeping an angle of $\theta = \bar{\theta} + 2
\pi h$ since $-1 \leq h < [1/(2-\alpha)]$.
The exclusion of collisions of type different from $(\rightrightarrows)$ can be done in
the same way as in the Keplerian case. Indeed, we can choose a suitable connecting arc and construct a local perturbation $\tilde{u}_I^*$,
which removes the collision and belongs to the cone $\K$.
Moreover, by Lemma \ref{lemma:marchal}, the action of $\tilde{u}_I^*$ is lower than the action of the colliding solution
$u_I^*$.

\end{proof}
\begin{proposition}
  Assume that $\K$ is $\alpha$-simple. Let $u_I^* \in \overline{\K}$ be a minimizer with a
  partial collision of type $(\rightrightarrows)$ at time $t_c$, and $\pi_{r, \mathtt{n}}$
  be the plane generated by $r$ and $\mathtt{n}=\mathtt{n}^+=\mathtt{n}^-$ of Definition~\ref{def:doppiaFreccia}. 
  Assume also that $u_I^*$ has no total collisions.
    Then $u_I^*(t) \in \pi_{r,\mathtt{n}}$ for all $t \in [0,T]$ and the generating particle bounces back and forth between two coboundary axes.
    \label{prop:OnRefPlane}
\end{proposition}
\begin{proof}
Assume that $u_I^*(t)$ solves \eqref{eq:eqMotionAlpha} on $(t_c-t^-,t_c)$ and $(t_c,t_c+t^+)$,
 and these are the maximal intervals of existence for these collision and ejection solutions relative to the axis $r$.
Let $\tilde{R}$ be the reflection with respect to the plane $\pi_{r, \mathtt{n}}$. 
Then from \eqref{eq:symmCondPartCollAlpha} it follows that $\tilde{R}u_I^*(t)$ is also a solution
of \eqref{eq:eqMotionAlpha} with the same value of $b$ defined in
Proposition~\ref{prop:unique} and with the same energy $h$. 
Since the reflection $\tilde{R}$ does not change $\mathtt{n}$, by Proposition~\ref{prop:unique} it follows that $\tilde{R}u_I^*(t) = u_I^*(t)$ for all $t \in (t_c-t^-, t_c+t^+)$. 

By the assumptions on $\K$, the minimizer cannot have total collisions, hence by Proposition~\ref{prop:SoloDoppiaFreccia} if follows that at times $t=t_c-t^-, t_c+t^+$ the only possibility is that $u_I^*$ has a partial collision of type $(\rightrightarrows)$. With the same argument as above we obtain therefore that $u_I^*(t) \in \pi_{r, \mathtt{n}}$ for all $t\in [0,T]$, and the generating particle bounces back and forth between two axes lying in $\pi_{r,\mathtt{n}}$. 

\end{proof}

Note that a collision loop such as the one in Proposition~\ref{prop:OnRefPlane} can
be reached by a sequence of non-collision loops in $\mathcal{K}$ only if the cone is tied
to two coboundary axes\footnote{Actually, it can be reached also when the
free-homotopy class of the cone is represented by a loop winding around one rotation axis
only. However, we excluded this case to get the coercivity (see Proposition~\ref{prop:PlatoCoercivity}).}.
Therefore, we conclude the proof of Theorem~\ref{th:noPartialCollisions} using Propositions~\ref{prop:SoloDoppiaFreccia}, ~\ref{prop:OnRefPlane}, and the fact that $\K$ is not tied to two coboundary
axes.



\subsection{Minimizers of the $\Gamma$-limit}
\label{ss:platoGammaLim}
Provided that the cone $\K$ is not central, we show that the minimizers of the $\Gamma$-limit belong to the boundary
$\partial\K$, which means
that the satellites pass closer and closer to partial collisions as the mass
$m_0$ increases, and they collide in the limit, as we have already seen in the example of
Section \ref{s:kleinGamma} with the Klein group symmetry. However, as opposite to the previous case,
here it can happen that the minimizers of the $\Gamma$-limit are not
$C^1$, see Section~\ref{ss:platoExamples}. 
For the case $\alpha=1$ we can show the following result.
\begin{proposition}
   Assume $\K$ is not central and $\alpha=1$. Then there exists a minimizer of the $\Gamma$-limit functional
   $\overline{\widetilde{\A}_0^1}$ that either has total collisions,
   or is composed by circular arcs centered at the origin and passing through some
   rotation axes, swept with uniform motion.
   \label{prop:gammalimitKeplerian}
\end{proposition}
\begin{proof}
   Let us suppose that $v_I^*$ is a minimizer without total
   collisions. Then, by Theorem~\ref{th:gammaConv} ii), we can conclude that $v_I^*$ is also a local minimizer of $\A_0^1$.
   By contradiction, if $v_I^*$ did not pass through
   any rotation axis, then it would be a classical smooth solution in $\K$ of
   \begin{equation}
     \ddot{v}_I = -\frac{v_I}{|v_I|^{3}}.
     \label{CFalpha}
   \end{equation}
   In particular, $v_I^*$ would lie on a plane passing through
   the origin, and this would imply that $\K$ is central.  Therefore,
   we can associate to $v_I^*$ a sequence of
   rotation semi-axes represented by unit vectors $\hat{r}_1,\dots,\hat{r}_m$, and a sequence 
   $0 \leq \tau_1 < \dots < \tau_m < T$ of collision times such that
   \[
      v_I^*(\tau_i) \in \{\lambda\hat{r}_i, \lambda>0\}, \quad i=1,\dots,m.
   \]
   Set $\hat{r}_{m+1} = \hat{r}_1$, let $\theta_i>0$ be the angle between 
   $\hat{r}_i$ and $\hat{r}_{i+1}$, for $i=1,\dots,m$, and define the total angle
   \begin{equation}
   \Delta \theta = \sum_{i=1}^m \theta_i.
      \label{eq:deltaTheta}
   \end{equation}
   Natural boundary conditions arising at the passage through the axes are
   \begin{equation}
      \begin{cases}
         \dot{v}_I^*(\tau_i^+)\cdot \hat{r}_i = \dot{v}_I^*(\tau_{i}^-) \cdot \hat{r}_i,
         \\[2ex]
         \lvert\dot{v}_I^*(\tau_i^+)\rvert = \lvert \dot{v}_I^*(\tau_{i}^-)\rvert. \\
      \end{cases}
      \label{eq:naturalBC}
   \end{equation}
   The first condition in \eqref{eq:naturalBC} is obtained by taking $T$-periodic variations
   $\eta$ such that $\eta(\tau_i) \in \{\lambda\hat{r}_i, \lambda>0\}$ and $\eta(t)=0$ for
   $t \in [0,T] \setminus (\tau_i-\delta, \tau_i+\delta)$, where $\delta>0$ is small
   enough.
   In fact, in this case we obtain 
   \begin{equation}
      0=\frac{d}{d\lambda} \A^1_0(v_I^* + \lambda\eta)\bigg|_{\lambda=0} =
      \dot{v}_I^* \cdot \eta \bigg|_{\tau_i - \delta}^{\tau_i} +  
      \dot{v}_I^* \cdot \eta \bigg|_{\tau_i}^{\tau_i+\delta} = \big( \dot{v}^*_I(\tau_i^-) -
      \dot{v}_I^*(\tau_i^+) \big) \cdot \eta(\tau_i).
      \label{eq:firstVariationZero}
   \end{equation}
   The second condition in \eqref{eq:naturalBC} is obtained by taking variations
   $v_\lambda(t) = v_I^*\big(t + \lambda \gamma(t) \big)$, where $\gamma \in C^\infty_0(\R)$
   with support in $(\tau_i -\delta, \tau_i +\delta)$ with $\delta>0$ small: 
   \begin{equation}
      0=\frac{d}{d\lambda} \A^1_0(v_\lambda(t))\bigg|_{\lambda=0} = \gamma
      |\dot{v}_I^*|^2\bigg|_{\tau_i - \delta}^{\tau_i} + 
      \gamma |\dot{v}_I^*|^2\bigg|_{\tau_i}^{\tau_i+\delta}
      = 
      \big( |\dot{v}_I^*(\tau_i^-)|^2 - |\dot{v}_I^*(\tau_i^+)|^2\big)\gamma(\tau_i).
      \label{eq:firstVariationZero2}
   \end{equation}
   These conditions imply that the values $E$ of the energy $|\dot v_I|^2/2 - 1/|v_I|$ 
   and $c$ of the size of the angular momentum $|v_I \times \dot{v}_I|$ are the same for each
   arc connecting two consecutive semi-axes. 
   Therefore, the radial component $\rho(t)=|v_I^*(t)|$ of the minimizer fulfills the
   equation of the reduced dynamics of the Kepler problem
   \begin{equation}
      \ddot{\rho} = -\frac{1}{\rho^2} + \frac{c^2}{\rho^3},
      \label{eq:keplerReduced}
   \end{equation}
   on the whole interval $[0,T]$.
   Moreover, it follows that $\rho(0) = \rho(T)$ because $v_I^*(t)$ is periodic, and $\dot{\rho}(0) = \dot{\rho}(T)$ because the values $E, c$ are the same for each arc. 
   Note that the action $\A^1_0(v_I^*)$ is the same as the action of a solution of the
   Kepler problem (lying in a plane) with the same radial component $\rho(t)$ and the same value $c$ of the angular momentum, sweeping a total angle
   $\Delta \theta$ in the time interval $[0,T]$. 
   %

   %
   %

   Therefore, if $\Delta\theta \neq 2k\pi, k\in \N$ the only possibility is that
   $v_I^*$ is composed by circular arcs.  
   %
   In the case $\Delta\theta = 2k\pi$ for some $k\in\N$, $v_I^*$ can be either composed by circular arcs or by elliptic arcs, sharing the same value of the action (see \cite{gordon77}). In this case they are all 
   minimizers of $\overline{\widetilde{\A}_0^1}$. 
\end{proof} 

\begin{remark}
   \label{rmk:deltaThetaMin}
   Let $\omega$ be the angular velocity of a minimizer $v_I^*$ composed by circular arcs,
   hence by the third Kepler's law
   \begin{equation}
      \omega = \frac{\Delta \theta}{T}, \qquad 
      \rho = \bigg( \frac{1}{\omega^2} \bigg)^{\frac{1}{3}}.
      \label{eq:angVel}
   \end{equation}
   The action of $v_I^*$ is therefore
   \begin{equation}
      \A_0^1(v_I^*) = \sum_{i=1}^m
      \int_{\tau_{i}}^{\tau_{i+1}} \bigg( \frac{\omega^2\rho^2}{2} +
      \frac{1}{\rho} \bigg) dt =
      \frac{3}{2}\frac{T}{\rho} =
      \frac{3}{2^{\frac{1}{3}} }\Delta
      \theta^{\frac{2}{3}}T^{\frac{1}{3}}.
      \label{eq:minAction}
   \end{equation}
   Note that here the action depends only on the total angle 
   $\Delta\theta$. Let $(\Delta\theta)_{\min}$ be the minimum total angle compatible with
   the cone $\K$. By the proof of Proposition~\ref{prop:gammalimitKeplerian} and from 
   the above observation, it follows that if $(\Delta\theta)_{\min} < 2\pi$ then the
   minimizer is necessarily composed by circular arcs.
   %
   In addition, if the symmetry \eqref{eq:extrasym} is fulfilled with $M>1$, the same fact
   holds true if $(\Delta\theta)_{\min}/M<2\pi$.
\end{remark}

In the case $\alpha \in (1,2)$ it is still possible to show that the minimizers of the
$\Gamma$-limit functional pass through some rotation semi-axes, provided that the cone
$\K$ is not central, and the proof is the same as in the case $\alpha=1$. 
However, in this case we cannot draw the same conclusion. We only provide the following conjecture.
%
%
\begin{conjecture}
   Assume $\K$ is not central, let $M \in \N$ be as in
   \eqref{eq:extrasym}, and $(\Delta\theta)_{\min}$ be the minimum total angle compatible
   with $\K$. Let $v_I^* \in \overline{\K}$
   be a minimizer of the $\Gamma$-limit functional $\overline{\widetilde{\A}_0^\alpha}$.
   Then if $(\Delta\theta)_{\min}/M<2\pi$, the trajectory $v_I^*([0,T])$ is composed
   by circular arcs, centered at the origin and passing through some rotation axes, which 
   are swept with uniform motion.
   \label{th:minColl}
\end{conjecture}

\subsection{Some examples}
\label{ss:platoExamples}
Here we discuss some examples, gathering the discussions of Sections \ref{ss:platoTotColl}, 
\ref{ss:platoPartialColl}, \ref{ss:platoGammaLim} to prove the existence of collision-free 
minimizers of $\A^\alpha_\varepsilon$ for $\varepsilon>0$.
In Table \ref{tab:sequences} we report the list of the selected free-homotopy classes. 
In Figure \ref{fig:exampleOrbits}, the trajectories of some of them are displayed, for different 
values of the central mass $m_0$. More images and videos can be found at the website
\cite{MFwebGamma}.
\begin{table}[ht!]
  \begin{center}
  \begin{tabular}{clcccc}
    \toprule
    $\mathcal{R}$ & $\nu$ & $M$ & $k_1$ & $k_2$ & $\alpha$ \\
    \midrule
    $\mathcal{T}$ & {\small$\nu_1 =[1, 5, 2, 6, 11, 3, 12, 9, 1]$} & $2$ & $8$ & / & $1$ \\
    & {\small $\nu_2 = [1, 5, 8, 3, 12, 4, 9, 7, 1]$ } & $2$ & $8$ & / & $1$ \\
    & {\small $\nu_3 = [1, 5, 8, 3, 10, 11, 3, 12, 4, 9, 12, 8, 1]$} & $3$ & $12$ & / & $1$ \\
    & {\small $\nu_4 = [1, 7, 6, 2, 7, 9, 12, 4, 9, 1, 5, 8, 1]$} & $3$ & $12$ & / & $1.7$ \\
    & {\small $\nu_5 = [1, 9, 7, 2, 5, 1, 7, 2, 10, 5, 1, 7, 2, 5, 1]$} & $2$ & $14$ & / &
    $1.85$ \\
    & {\small $\nu_6=[1, 9, 4, 12, 9, 4, 12, 9, 7, 2, 10, 3, 11, 10, 3, 11, 10, 5, 1]$} &
    $2$ & $18$ & / & $1.86$ \\
    $\mathcal{O}$ & {\small$\nu_{1} =  [1, 3, 7, 20, 24, 12, 4, 9, 2, 5, 1]$} & $2$
    &$4$ & $6$ & $1$ \\
    &\small{$\nu_{2} = [ 1, 3, 8, 18, 13, 12, 4, 9, 2, 19, 11, 14, 1]$} & $2$ & $4$&$8$ & $1$   \\
    &\small{$\nu_{3} = [ 1, 3, 7, 20, 18, 8, 15, 4, 6, 10, 16, 5, 1]$} & $3$ & $6$&$6$ & $1$ \\
    &\small{$\nu_{4} = [ 1, 3, 8, 15, 4, 9, 2, 5, 1]$} & $4$ & $4$&$4$ & $1$\\
    &\small{$\nu_5 = [1, 3, 10, 8, 15, 6, 4, 9, 22, 2, 5, 16, 1]$} & $4$ & $8$&$4$ &
    $1$ \\
    &\small{$\nu_{6}=[ 1, 3, 8, 10, 3, 7, 20, 18, 7, 14, 11, 23, 14, 1, 16, 5, 1]$}
    & $4$ & $12$&$2$ & $1.6$\\
    &\small{$\nu_{7} = [ 1, 14, 7, 20, 23, 14, 7, 3, 1, 16, 10, 3, 1]$} 
    & $2$ & $4$&$8$ & $1.7$\\
    &\small{$\nu_{8}=[ 1, 14, 7, 20, 23, 14, 7, 3, 1, 14, 7, 3, 1, 16, 10, 3, 1,
    14, 7, 3, 1]$} & $2$ & $4$&$ 16$ & $1.8$\\
    & \small{$\nu_{9} = [ 1, 16, 22, 6, 10, 16, 5, 1, 3, 7, 14, 1, 16, 5, 11, 19, 2, 5, 1]$}
    & $3$ & $6$&$ 12$ & $1.75$ \\
    $\mathcal{I}$ &\small{$\nu_1 = [1, 3, 6, 11, 48, 15, 25, 26, 33, 47, 7, 12, 52, 59,
    54, 50, 1]$} & $2$ & $6$&$10$ & $1$ \\
    &\small{$\nu_2=[1, 3, 59, 54, 51, 36, 35, 46, 10, 17, 57, 56, 60, 5, 4, 8, 14, 24, 38,
    34,$} & $3$ & $9$&$ 15$ & $1$ \\
 & \small{\phantom{00000}$48, 28, 11, 19, 1]$} &  && &  \\
    &\small{$\nu_3=[1, 3, 7, 12, 21, 39, 30, 44, 2, 4, 8, 20, 31, 45, 19, 1]$} & $5$ &
    $5$&$10$ & $1$ \\
    &\small{$\nu_4=[1, 3, 59, 7, 3, 6, 47, 15, 6, 11, 48, 28, 11, 19, 45, 43, 19, 1, 50,
    54, 1]$} & $5$ & $15$&$5$ & $1$ \\
    \bottomrule
  \end{tabular}
  \end{center}
  \caption{Sequences of vertexes of $\mathcal{Q}_{\mathcal{R}}$, defining the
  free-homotopy classes. The enumeration of the vertexes for
  $\mathcal{Q}_{\mathcal{O}}$ is referred to the one in Figure~\ref{fig:encoding}. 
  Images with the enumeration of the vertexes of the other two Archimedean polyhedra can
  be found at \cite{MFwebGamma}. 
  Since for $\mathcal{R} = \mathcal{T}$ the distinction
  between the two kind of sides is not relevant, only one value is reported.}
  \label{tab:sequences}
\end{table}
\begin{figure}[!ht]
   \centerline{
      \includegraphics[width=4.5cm, height=4.5cm]{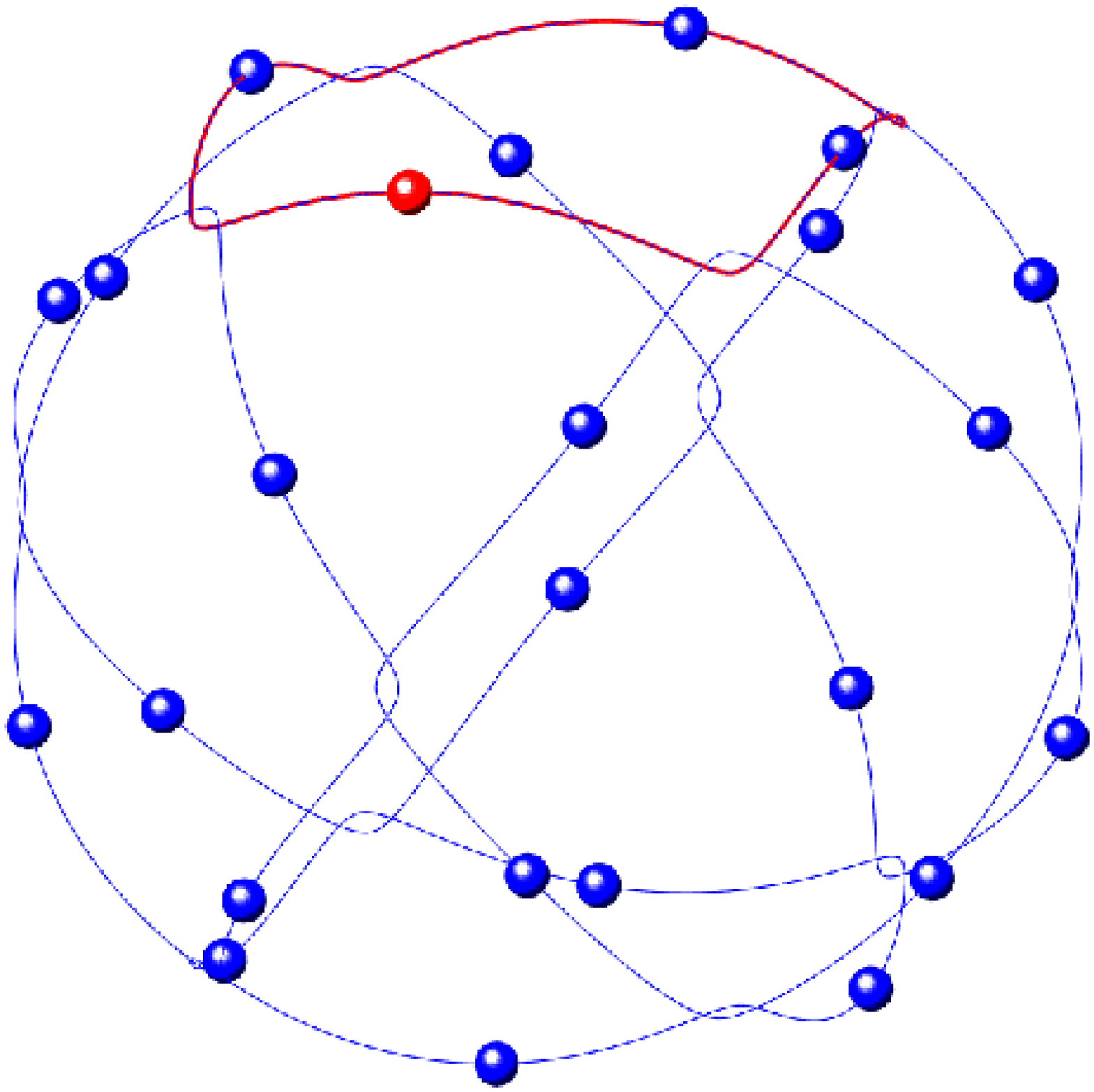}
      \includegraphics[width=4.5cm, height=4.5cm]{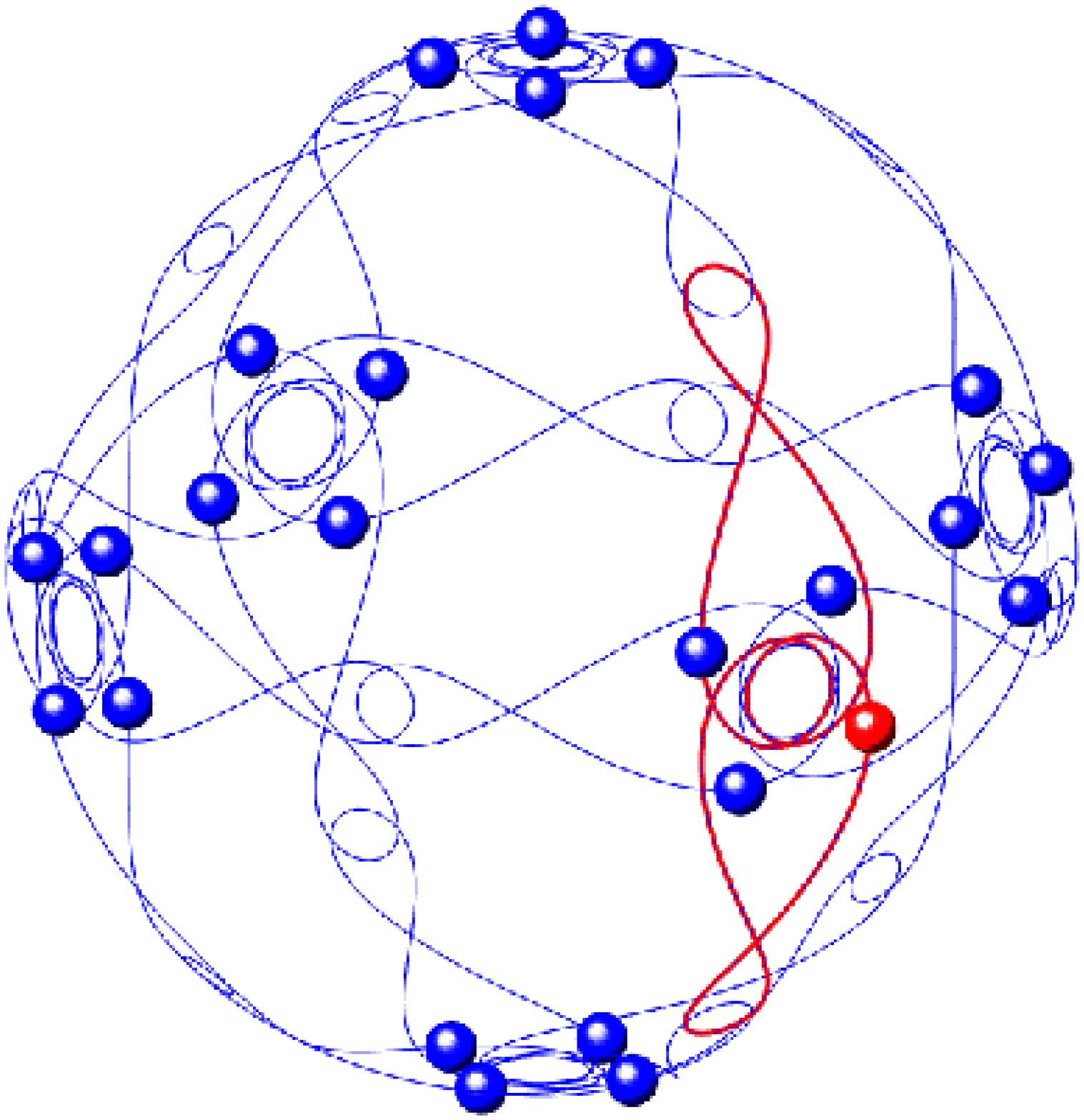}
      \includegraphics[width=4.5cm, height=4.5cm]{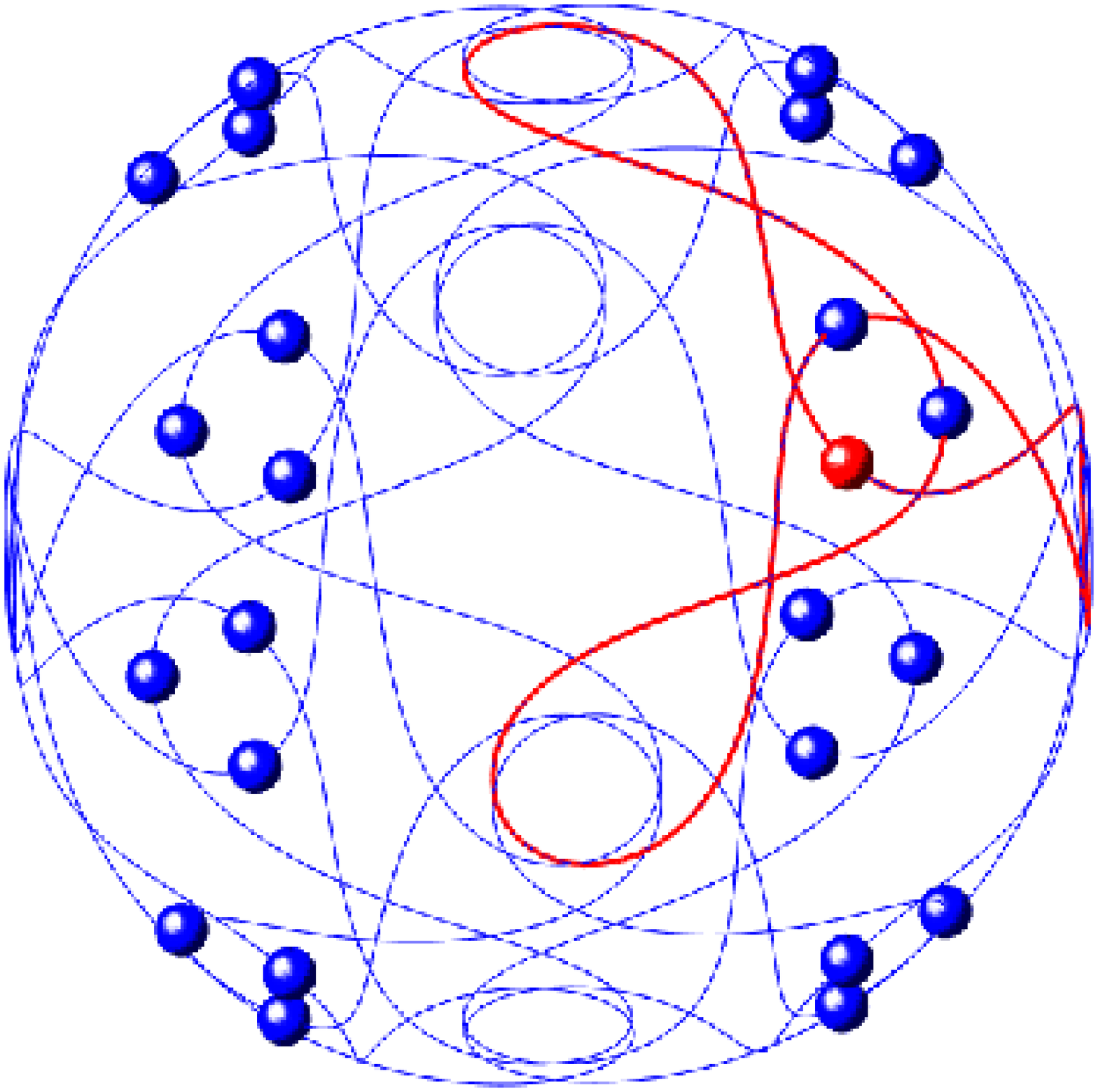}
   }
   \vskip 0.2cm 
   \centerline{
      \includegraphics[width=4.5cm, height=4.5cm]{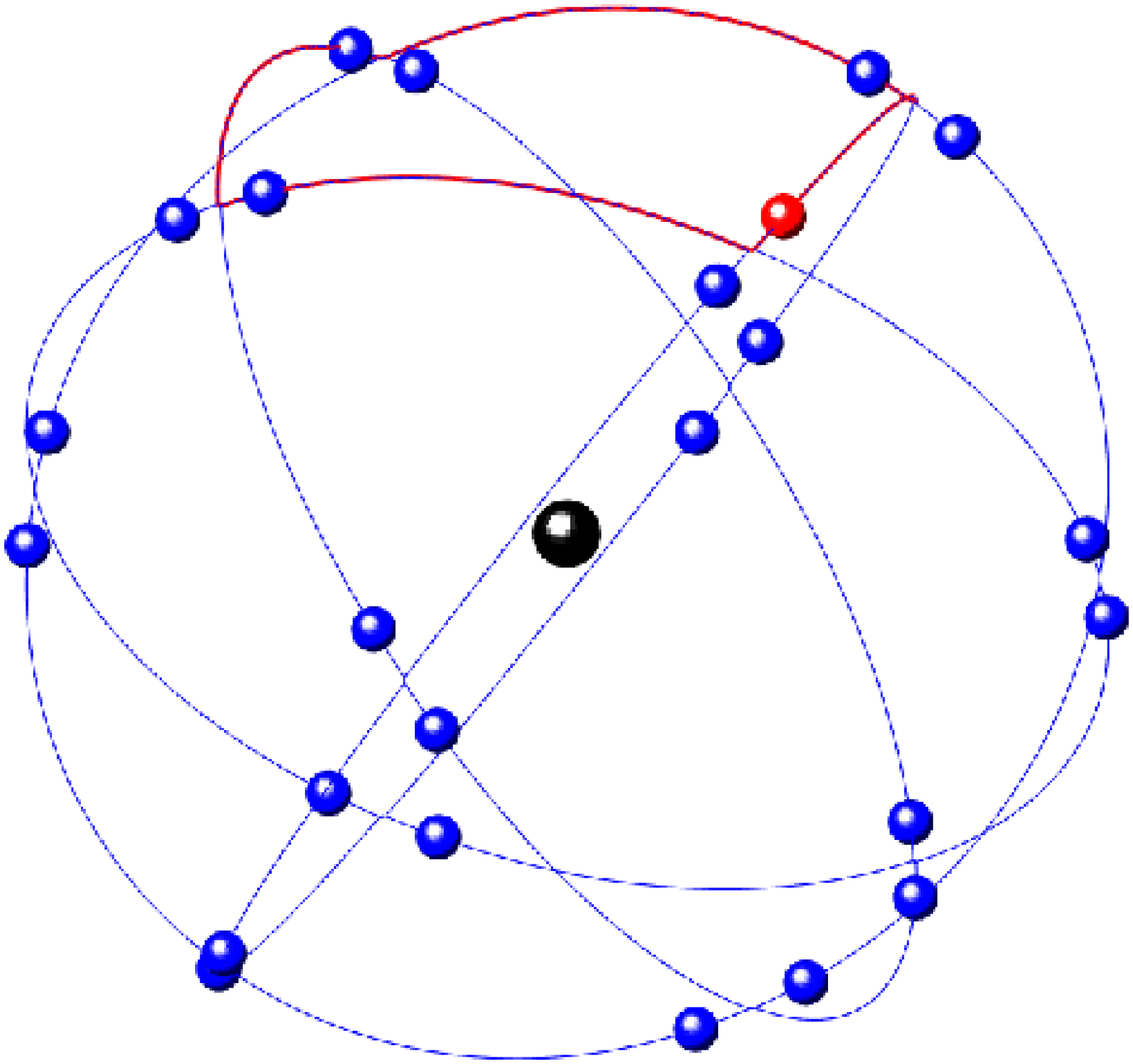}
      \includegraphics[width=4.5cm, height=4.5cm]{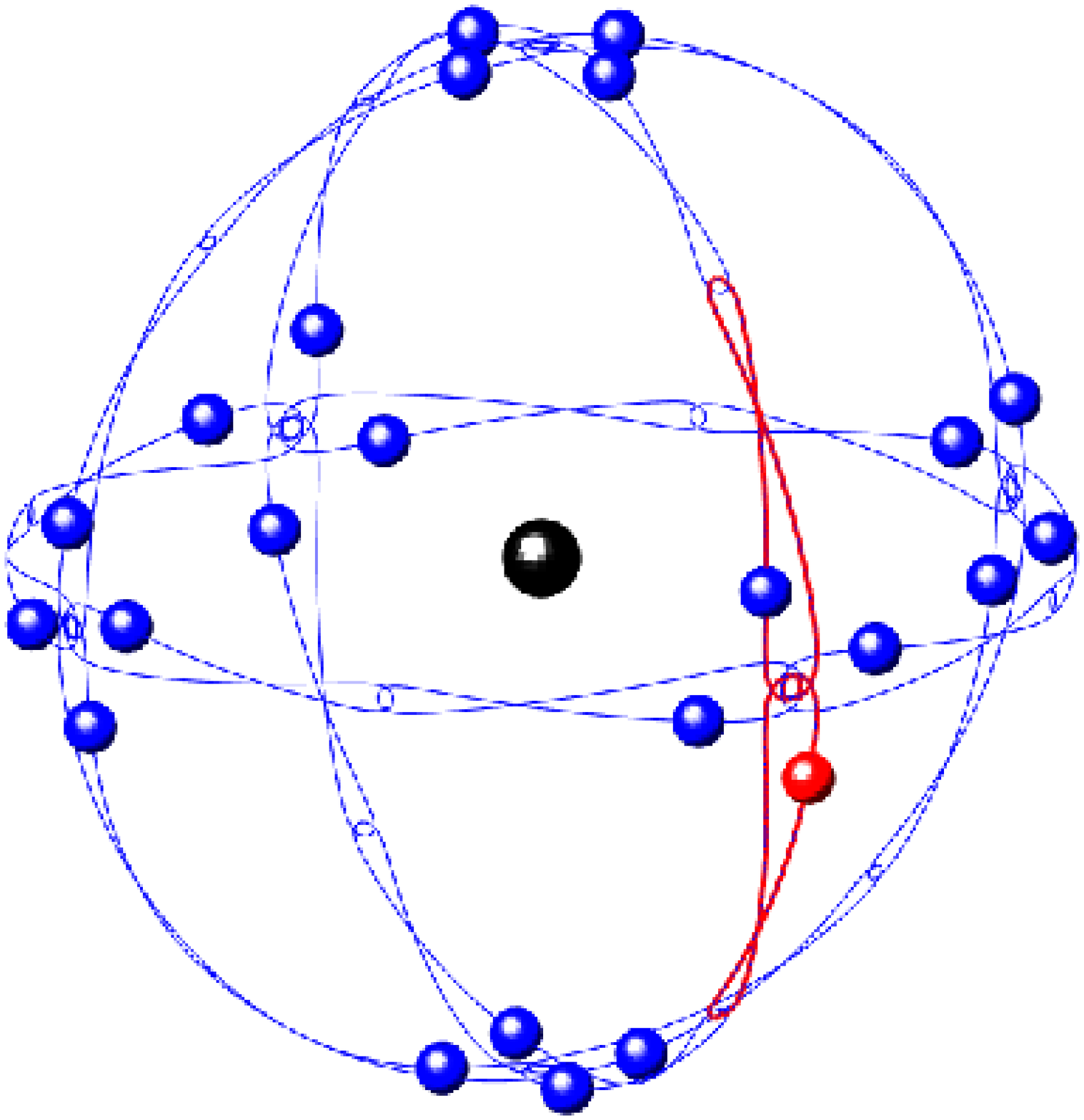} 
      \includegraphics[width=4.5cm, height=4.5cm]{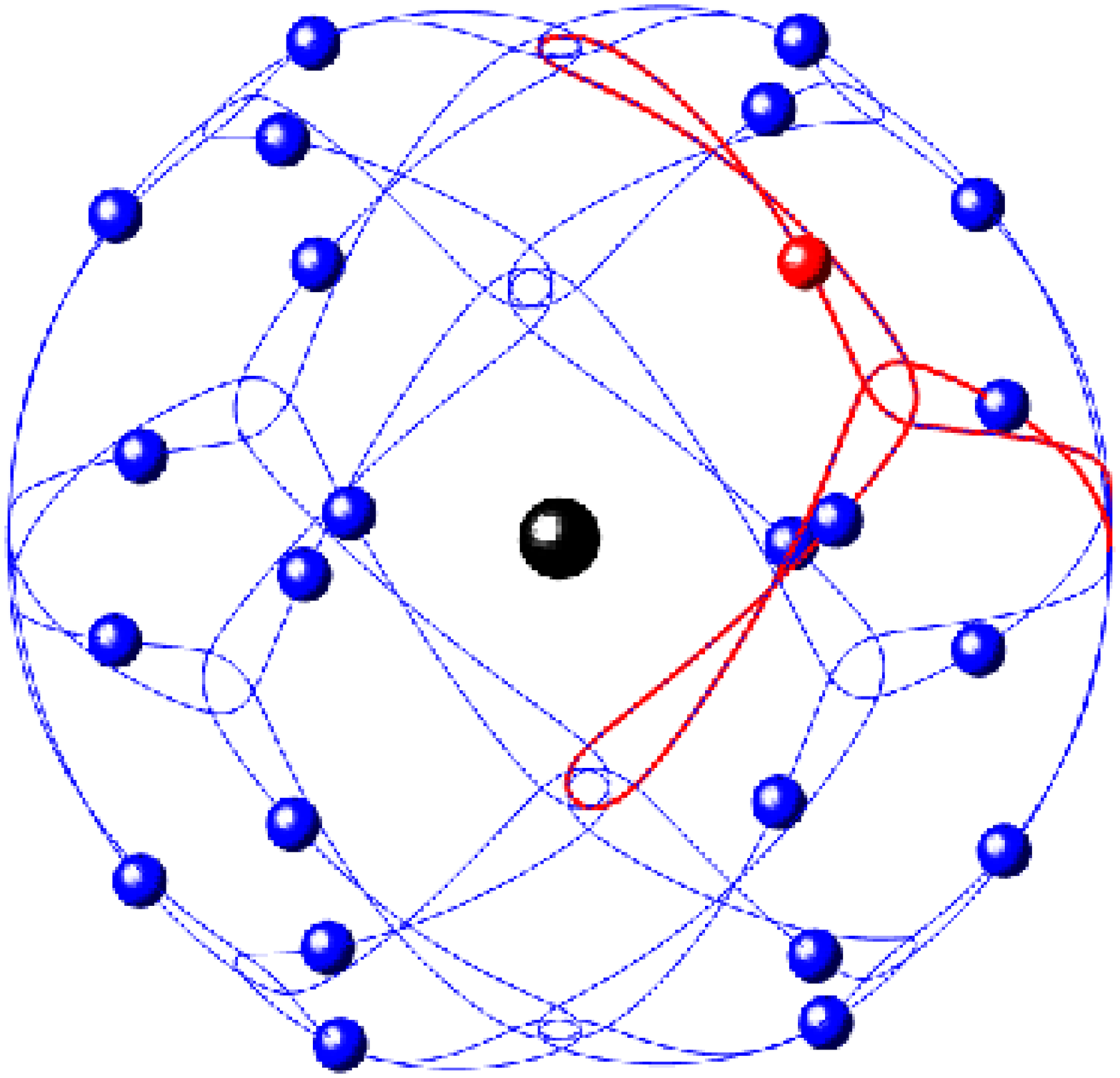} 
   }
   \caption{Some periodic orbits with the symmetry of the Cube. The
     topological constraints are given, from the left to the right, by
     the sequences $\nu_5, \nu_8, \nu_9$ of Table
     \ref{tab:sequences}. Orbits in the same column belong to the same
     free-homotopy class.  The mass of the central body for the figures on the top
     is $m_0=0$, while $m_0=100$ for the figures on the bottom. As the
     value of $m_0$ increases, the minimizer approaches an
     orbit composed by circular arcs, joined at some points on the
     rotation axes, where partial collisions occur.}
   \label{fig:exampleOrbits}
\end{figure}
\begin{theorem}
   For each sequence $\nu$ and the corresponding value of
   the exponent $\alpha$ listed in Table \ref{tab:sequences}, there
   exists a sequence $\{ v^*_{I,\varepsilon} \}_{\varepsilon>0}$ of
   collision-free minimizers of $\A^\alpha_\varepsilon$ in the cone $\K=\K(\nu)$, such that
   each $v^*_{I,\varepsilon}$ is a classical $T$-periodic solution of the $(1+N)$-body problem 
   with $\alpha$-homogeneous potential.
   \label{th:minExist}
\end{theorem}
\begin{proof}
  Sequences of Table~\ref{tab:sequences} do not wind
  around one axis only, hence the action functional $\A^\alpha_\varepsilon$ is coercive in the cone $\K(\nu)$, independently from the value of $m_0$, and
  minimizers therefore exist.

  We use the estimates of Section
  \ref{ss:platoTotColl} to exclude that minimizers have total collisions.  Let us distinguish two cases: $\alpha=1$ and
  $\alpha \in (1,2)$. If $\alpha=1$ the action
  \eqref{eq:testLoopAction} of the test loop $v_I^\nu$, associated to the
  sequence $\nu$, can be computed through elementary
  functions.  Therefore, we use the
  estimates \eqref{eq:totCollEstAlpha} and \eqref{eq:testLoopAction},
  and check that inequality \eqref{actionless} holds.
   After some computations we find that the above inequality is
   satisfied if and only if
   \begin{equation}
      k_1\zeta_{1,1} + k_2\zeta_{1,2}+ m_0(k_1+k_2)\zeta_{1,0} <
      \frac{2\pi M}{\ell} (\tilde{U}_{1,0}+2m_0).
      \label{eq:sufCondGammaEst}
   \end{equation}
   Note that inequality \eqref{eq:sufCondGammaEst} is verified for
   every value of $m_0\geq0$ if and only if
   \begin{equation}
   \begin{cases}
      \displaystyle k_1\zeta_{1,1}+k_2\zeta_{1,2} < \frac{2\pi M}{\ell}\tilde{U}_{1,0},
      \\[2ex]
      \displaystyle (k_1+k_2)\zeta_{1,0} < \frac{4\pi M}{\ell}.
   \end{cases}
   \label{eq:sufCondGammaEstTrue}
   \end{equation}
   Values of the members of the above inequalities for the sequences of Table~\ref{tab:sequences}
   are reported in Table~\ref{tab:totCollEst}.
   In the case $\alpha \in (1,2)$, the action of the test loop $v_I^\nu$ has to be estimated. For this purpose 
   we use \eqref{eq:testLoopActionEst} to estimate the left hand side of
   \eqref{actionless}. A sufficient condition to exclude total collisions is
   therefore 
   \begin{equation}
      k_1\frac{\zeta_{1,1}}{\delta_1} + k_2\frac{\zeta_{1,2}}{\delta_2} +
      \frac{8(k_1+k_2)}{4-\ell^2} m_0 <
      C
      \big(\tilde{U}_{\alpha,0} + 2m_0 \big),
      \label{eq:sufCondAlpha}
   \end{equation}
   where
   \begin{equation}
      C = 
      \frac{2k_\nu}{(2-\alpha)^{\frac{2+\alpha}{2}}}
      \bigg(\frac{\pi M}{\alpha^{1/2}\ell k_\nu}\bigg)^\alpha.
      \label{eq:C1C2forEst}
   \end{equation}
   Condition \eqref{eq:sufCondAlpha} is verified for every value of $m_0\geq0$
   if and only if
   \begin{equation}
   \begin{cases}
      \displaystyle 
      k_1\frac{\zeta_{1,1}}{\delta_1} + k_2\frac{\zeta_{1,2}}{\delta_2} 
       <
       C
       \tilde{U}_{\alpha,0},
      \\[2ex]
      \displaystyle 
      \frac{8(k_1+k_2)}{4-\ell^2}  <
      C.
   \end{cases}
   \label{eq:sufCondAlphaTrue}
   \end{equation}
   Values of the members of the above inequalities for the sequences of Table~\ref{tab:sequences}
   are reported in Table~\ref{tab:totCollEstAlpha}.
   In all the cases, inequality \eqref{actionless} holds true, hence the minimizers are free
   of total collisions for all the values of $m_0 \geq 0$.

\begin{table}
   \centering
   \begin{tabular}{cccccc}
      \toprule
      $\mathcal{R}$ & $\nu$ & $k_1\zeta_{1,1}+k_2\zeta_{1,2}$ & $\displaystyle\frac{2\pi M
         \tilde{U}_{1,0}}{\ell}$ &
         $(k_1+k_2)\zeta_{1,0}$ & $\displaystyle\frac{4\pi M}{\ell}$ \\
      \midrule
      $\mathcal{T}$& $\nu_1$ &  $76.6704$ &  $80.0636$ &  $17.5776$ &  $25.1327$ \\
                   & $\nu_2$ &  $76.6704$ &  $80.0636$ &  $17.5776$ &  $25.1327$ \\
                   & $\nu_3$ & $115.0056$ &  $120.0954$ &  $26.3664$ &  $37.6991$ \\
      $\mathcal{O}$ & $\nu_1$ &  $199.7300$ &  $253.2198$ &  $20.9230$ &  $35.1556$ \\
                    & $\nu_2$ &  $239.2100$ &  $253.2198$ &  $25.1076$ &  $35.1556$ \\
                    & $\nu_3$ &  $240.3750$ &  $379.8298$ &  $25.1076$ &  $53.7334$ \\
                    & $\nu_4$ &  $160.2500$ & $506.4397$ &  $16.7384$ &  $70.3112$ \\
                    & $\nu_5$ &  $241.5400$ & $506.4397$ &  $25.1076$ &  $70.3112$ \\
      $\mathcal{I}$ & $\nu_1$ &  $ 849.7033$ & $1151.4$ &  $32.5513$ &  $56.1123$ \\
                    & $\nu_2$ &  $1274.5550$ & $1727.1$ &  $48.8270$ &  $84.1685$ \\
                    & $\nu_3$ &  $ 795.7130$ & $2878.5$ &  $30.5169$ & $140.2809$ \\
                    & $\nu_4$ &  $1072.7354$ & $2878.5$ &  $40.6892$ & $140.2809$ \\
      \bottomrule
   \end{tabular}
   \caption{Values of the terms in the inequalities
      \eqref{eq:sufCondGammaEstTrue} corresponding to the sequences with $\alpha=1$ in
      Table~\ref{tab:sequences}.}
   \label{tab:totCollEst}
\end{table}

\begin{table}
   \centering
   \begin{tabular}{cccccc}
      \toprule
      $\mathcal{R}$ & $\nu$ & $\displaystyle k_1\frac{\zeta_{1,1}}{\delta_1} +
      k_2\frac{\zeta_{1,2}}{\delta_2}$ & $C\tilde{U}_{\alpha,0}$ &
      $\displaystyle \frac{8(k_1+k_2)}{4-\ell^2}$ & $C$ \\
      \midrule
      $\mathcal{T}$& $\nu_4$ & $321.7840$  &  $410.0057$ &    $32.0000$   &   $94.0390$ \\
                   & $\nu_5$ & $375.4147$  &  $558.0238$ &    $37.3333$  &  $138.7856$  \\
                   & $\nu_6$ & $482.6760$  &  $507.4591$ &    $48.0000$  &  $126.8925$  \\
%
      $\mathcal{O}$ & $\nu_6$ &  $760.4405$  &  $1588.5795$ &   $32.1016$  &  $143.9767$ \\
                    & $\nu_7$ &  $539.8787$  &   $882.2706$ &   $27.5157$  &   $83.5095$ \\
                    & $\nu_8$ &  $852.3135$  &  $1155.3966$ &   $45.8595$  &  $114.1789$ \\
                    & $\nu_9$ &  $809.8181$  &  $1779.5666$ &   $41.2735$  &  $172.1174$ \\
      \bottomrule
   \end{tabular}
   \caption{Values of the terms in the inequalities
      \eqref{eq:sufCondAlphaTrue} corresponding to the sequences with $\alpha \in (1,2)$  in
      Table~\ref{tab:sequences}.}
   \label{tab:totCollEstAlpha}
\end{table}

   Partial collisions are excluded by the results of Section
   \ref{ss:platoPartialColl}. Indeed, the 
   sequences in Table \ref{tab:sequences} correspond to $\alpha$-simple cones, for the listed values of
   $\alpha$. Moreover, they are not tied to two coboundary axes.
   Hence, for each sequence $\nu$ of Table \ref{tab:sequences}, there exists
   a sequence $\{ v_{I,\varepsilon}^* \}_{\varepsilon>0}$ of
   collision-free minimizers of $\A^\alpha_\varepsilon$, corresponding
   to a classical $T$-periodic solution of the $(1+N)$-body problem.
\end{proof}

\begin{theorem}
   For each sequence $\nu$ with $\alpha=1$ listed in Table \ref{tab:sequences}, 
   the sequence of minimizers $\{ v^*_{I,\varepsilon} \}_{\varepsilon>0}$ converges as
   $\varepsilon\to0$ to a minimizer $v^*_I$ of the Kepler problem \eqref{eq:gammaLim}, 
   which is composed by circular arcs centered at the origin and passes through the minimum number of rotation axis compatible with the cone $\K$.
   \label{th:minExist2}
\end{theorem}
\begin{proof}
   The convergence of the sequence of minimizers $\{ v^*_{I,\varepsilon}
   \}_{\varepsilon>0}$ to a minimizer of the $\Gamma$-limit $\A^1_0$ 
   is guaranteed by Theorem~\ref{th:convMin}.
   To understand the shape of the minimizers of $\A^1_0$, we note that none the cones $\K$ 
   identified by the sequences of Table~\ref{tab:sequences}
   is central, hence for $\alpha=1$ we are in the hypotheses of
   Proposition~\ref{prop:gammalimitKeplerian}. 
   Moreover, each cone is such that
   $(\Delta\theta)_{\min}/M < 2\pi$, hence by Remark~\ref{rmk:deltaThetaMin} the minimizer
   of the $\Gamma$-limit functional can be composed only by circular arcs, joined at the
   rotation semi-axes that realize $(\Delta\theta)_{\min}$. 
\end{proof}

\begin{remark}[Minimizers of the $\Gamma$-limit for $\alpha\in(1,2)$]
   For the remaining sequences $\nu$ of Table~\ref{tab:sequences} with $\alpha\in(1,2)$,
   we can still prove that the sequence of minimizers $\{ v^*_{I,\varepsilon}
   \}_{\varepsilon>0}$ converges to a minimizer $v_I^*$ of $\A^\alpha_0$, that passes
   through some rotation semi-axes.
   However, by the remarks reported in the previous section, we cannot make a
   conclusion on the shape of $v_I^*$. 

   On the other hand, in the numerical computations that we performed we 
   did not find
   any example in which the sequence of minimizers converges to a loop that is not composed by
   circular arcs. This provides numerical evidence for Conjecture~\ref{th:minColl}.
\end{remark}


\appendix
\section{Level estimates for total collisions}
\label{app:LE}

\begin{proposition}
Let $\alpha\in[1,2)$ and let $u:[0,\textsc{T}] \to \R^{3N}$ be a
  motion of $N$ masses $m_1,..., m_N$ 
such that $u(0) = u(T) = 0$.
  Then for
  the action
\[ 
\A^\alpha(u) = \int_0^\textsc{T}\Bigl(\frac{1}{2}\sum_{h=1}^N m_h|\dot{u}_h|^2 +
\sum_{1\leq h<k\leq N}\frac{m_h\,m_k}{|u_h-u_k|^\alpha} \Bigr)dt,
\]
we have the estimate
\begin{equation}
\A^\alpha(u) \geq 
\mathcal{M} \bar{c}_{U_{\alpha,0}}^\alpha(\textsc{T}),
\quad \mathcal{M}=\sum_{h=1}^N m_h,
\label{eq:propLL}
\end{equation}
where
\[
U_{\alpha,0} = \min_{\rho(u)=1} \mathcal{U}_\alpha(u), \quad
\mathcal{U}_\alpha(u) = \frac{1}{\mathcal{M}}\sum_{\stackrel{h,k=1}{h\neq k}}^N
\frac{m_h m_k}{|u_h-u_k|^\alpha}, \quad
\rho(u)=\biggl(\sum_{h=1}^N
\frac{m_h}{\mathcal{M}}|u_h|^2\biggr)^{1/2}, 
\]
and
\begin{equation}
  \bar{c}_{U_{\alpha,0}}^\alpha(\textsc{T}) =
\textsc{T}^{\frac{2-\alpha}{2+\alpha}}\frac{2+\alpha}{2-\alpha}
{(2\alpha^2)}^{-\frac{\alpha}{2+\alpha}}\left(\frac{U_{\alpha,0}}{2}\right)^{\frac{2}{2+\alpha}}\left(\int_0^{2\pi}
|\sin t|^{\frac{2}{\alpha}}\right)^\frac{2\alpha}{2+\alpha} .
\label{eq:cbar}
\end{equation}
\label{prop:lower_bound_gen}
\end{proposition}

\begin{proof}
  The proof follows the same steps of \cite{fgn10}, where only the
  case $\alpha=1$ was considered.
Using the $\alpha$-homogeneity of the potential we obtain
\[
\A^\alpha(u) \geq
\frac{\mathcal{M}}{2}\int_0^\textsc{T}\Bigl(\dot\rho^2(u) +
\frac{1}{\rho^\alpha(u)} \mathcal{U}_\alpha
\Bigl(\frac{u}{\rho(u)}\Bigr)\Bigr)\,dt \geq
\mathcal{M}\int_0^{\textsc{T}}\Bigl(\frac{1}{2}\dot\rho^2 +
\frac{U_{\alpha,0}}{2\rho^\alpha}\Bigr)dt.
\]
Then the result follows from the relation
\[
\bar{c}^\alpha_{a}(\textsc{T}) = \inf_{\mathfrak{S}}
\int_0^{\textsc{T}}\left(\frac{1}{2}|\dot u|^2 +
\frac{a}{|u|^\alpha}\right)dt,\quad
\mathrm{ for }\ a>0,
\]
where
\[
\mathfrak{S} = \{u\in H^1_T(\R,\R^3) : u(t)=0 \mathrm{\ for\ some\ } t\},
\]
see \cite{RT95}.
\end{proof}


\begin{corollary}
   In the same hypotheses of Proposition \ref{prop:lower_bound_gen} we
   have
\begin{equation}
  \A^\alpha(u) >
  \frac{2+\alpha}{2-\alpha}\frac{\mathcal{M}}{2}
  \left[U_{\alpha,0}\Bigl(\frac{\pi}{\alpha}\Bigr)^\alpha\right]^{\frac{2}{2+\alpha}}
  \textsc{T}^{\frac{2-\alpha}{2+\alpha}}.
\label{eq:sndbnd}
\end{equation}
\label{cor:LE}
\end{corollary}

\begin{proof}
  The inequality \eqref{eq:sndbnd} follows immediately from
  Proposition~\ref{prop:lower_bound_gen} and the estimate
\[
   \int_0^{2\pi} |\sin t|^{\frac{2}{\alpha}}dt > \pi,
\quad \forall\,\alpha\in(1,2).
\]
\end{proof}

\section{Marchal's Lemma for $\alpha$-homogeneous potentials}
\label{app:marchal}
In the literature, Marchal's Lemma is often referred to the following
result (see for example \cite{marchal01, ch02, FT2004, montgomery2018}
for a proof).
\begin{lemma}[Marchal's Lemma]
   Let $\tau > 0$, $\alpha \in [1,2)$ and $x_A, x_B \in \R^2\setminus
     \{ 0 \}$ be two points in the plane.  Then any minimizer of the
     action
   \begin{equation}
      \A(x)=
      \int_{0}^{\tau}\bigg(\frac{|\dot{x}|^2}{2}+\frac{1}{|x|^\alpha}\bigg)dt,
      \label{eq:actionMarchalFixedPoints}
   \end{equation}
   on the set of curves $x:[0,\tau] \to \R^2$ such
   that $x(0) = x_A, \, x(\tau) = x_B$, is free of interior
   collisions.
   \label{lemma:usualMarchalLemma}
\end{lemma}
In the Keplerian case different proofs of this lemma are given in \cite{ch02}.
However, this statement does not provide any information on the number
of minimizers, and it does not relate the action of other stationary
points to the action of the collision-ejection solution.

Another version of Marchal's Lemma, stated for the Keplerian case $\alpha=1$, can be found
in \cite{fgn10}. Here, only the parabolic collision-ejection solution is taken into account
and it states that there are actually two Keplerian arcs, 
connecting any two distinct points at the same distance from the origin and in the same
time as the parabolic collision-ejection solution, whose action is less than the action 
of the parabolic collision-ejection solution itself. In general, the action of the two
Keplerian arcs is different.
Moreover, if the end points coincide, there is only one non-collision connecting
solution.
%
Hence we can conclude that the action of any non-collision solution of
\eqref{eq:actionMarchalFixedPoints} with $\alpha=1$ is less than the action of the parabolic
collision-ejection solution.
Furthermore, the number of solutions is related to the angle between
$x_A$ and $x_B$: this is also important in the proof of the exclusion of
partial collisions given in Section \ref{s:platoGamma}.
However, the proof in \cite{fgn10} of this version of the lemma relies on 
the explicit form of the solutions of the Kepler problem, hence it cannot 
be adapted to the case $\alpha \in (1,2)$.

A more general statement of Marchal's Lemma is contained in \cite{barutello2013}, and 
the technique used for the proof was already present in \cite{terrvent07}. The result can be summarized
as follows.
\begin{lemma}
Let $x_A, x_B \in \R^2 \setminus \{ 0 \}$ be two points in the plane 
   and $\tau>0$. Let 
   \[
      x_A=r_A(\cos\varphi_A, \sin\varphi_A), \quad x_B=r_B(\cos\varphi_B, \sin \varphi_B),
   \]
   be the polar coordinates of the two points. Given an integer $k \in \Z$ such that
   \[
      |\varphi_A-(\varphi_B+2k\pi)|<\frac{2\pi}{2-\alpha},
   \]
   define
   \[
      \begin{split}
      \mathcal{G} = \big\{
         x \in H^1([-\tau,\tau], \R^2\setminus \{ 0 \} )&: x(-\tau) = x_A, \, x(\tau) = x_B \\ & \text{ and the total
         angle swept by } x\text{ is } \varphi_B+2k\pi-\varphi_A
      \big\}.
   \end{split}
   \]
   Then any minimizer of the action
   \[
      \A(x) = \int_{-\tau}^{\tau} \bigg(
      \frac{|\dot{x}|^2}{2}+\frac{1}{|x|^\alpha}\bigg)dt,
   \]
   in the $H^1$-closure of $\mathcal{G}$ is free of collisions.
   \label{lemma:barTerrVerz_Marchal}
\end{lemma}
This more general version of Marchal's Lemma gives information not
only on the minimizer, but also on other solutions of the
fixed-ends problem. The idea for the proof is the
following. A disk of radius $\varepsilon>0$ centered at the origin is
removed from the plane and an obstacle problem is introduced. If a
minimizer in the $H^1$-closure of $\mathcal{G}$ has a collision, then the minimizers
$x^*_\varepsilon$ of the obstacle problem touch the border of the
disk, for every $\varepsilon>0$.  Using a blow-up technique we can
prove that $x_\varepsilon^*$ is composed by two parabolic arcs,
connected by a circular arc on the border of the disk. Then, the
total variation of the angle can be estimated and it results to be
greater than or equal to $2\pi/(2-\alpha)$ for every
$\epsilon$. For the original problem without obstacles, passing to the limit as $\epsilon\to 0$, we obtain the 
same estimate for the total variation of the angle, 
and this is in contradiction with the admissible arcs being in the
set $\mathcal{G}$. Hence minimizers are collision free.

A version of Marchal's Lemma which is suitable for our purposes, can
be stated as follows.
\begin{lemma}
   Let $\alpha \in (1,2)$ and $\bar{\rho}>0$. Let $\tau(0)$ be the time needed
   to arrive at the collision for the parabolic collision-ejection solution of
   \begin{equation}
     \ddot{x}=-\alpha\frac{x}{|x|^{\alpha+2}},\qquad x\in\R^2\setminus\{0\}
     \label{ahomogkep}
   \end{equation}
   starting at distance $\bar{\rho}$ from the origin. Let $\varphi \in [0,2\pi)$ and set
   \[
   \begin{cases}
      x_A = (\bar{\rho}, 0), \\
      x_B = \bar{\rho} (\cos\varphi, \sin\varphi). 
   \end{cases}
   \]
   Let $\bar{x}:[-\tau(0), \tau(0)]\to \R^2$ be any non-collision solution of the
   Euler-Lagrange equation \eqref{ahomogkep}
   such that
   \[
      \bar{x}\big(-\tau(0)\big) = x_A, \quad \bar{x}\big( \tau(0) \big) = x_B,
   \]
   and denote with $x_{0}:[-\tau(0), \tau(0)]\to \R^2$ the parabolic collision-ejection 
   solution with the same boundary conditions.
   Then we have that
   \begin{equation}
      \A(\bar{x}) < \A(x_{0}),
      \label{eq:marchalTheorem}
   \end{equation}
   where
   \[
      \A(x) = \int_{-\tau(0)}^{\tau(0)} \bigg(
      \frac{|\dot{x}|^2}{2}+\frac{1}{|x|^\alpha}\bigg)dt.
   \]
\label{lemma:marchal}
\end{lemma}
Moreover, the total number of connecting arcs that are solutions of
\eqref{ahomogkep} is a function of the angle $\varphi$ between the
points and of the exponent $\alpha \in (1,2)$.
\begin{lemma}
   In the hypotheses of Lemma \ref{lemma:marchal}, connecting arcs all have a
   different winding number with respect to the origin. In particular, the total angle swept by each arc
   is
   \[
      \varphi + 2\pi k, \quad k=k_{\text{min}},\dots,k_{\text{max}},
   \]
   where
\begin{equation}
   k_{\text{min}} = 
   \begin{cases}
      \displaystyle -\bigg[ \frac{1}{2-\alpha} + \frac{\varphi}{2\pi} \bigg]
      &\displaystyle \text{if  }
      \frac{1}{2-\alpha} + \frac{\varphi}{2\pi} \notin \Z, \\[3ex]
      \displaystyle -\bigg( \frac{1}{2-\alpha} + \frac{\varphi}{2\pi} \bigg)+1
      &\displaystyle \text{if  }
      \frac{1}{2-\alpha} + \frac{\varphi}{2\pi} \in \Z, \\
   \end{cases}
   \label{eq:kappaMinApp}
\end{equation}
\begin{equation}
   k_{\text{max}} = 
   \begin{cases}
      \displaystyle \bigg[ \frac{1}{2-\alpha} - \frac{\varphi}{2\pi} \bigg] &
      \displaystyle \text{if  }
      \frac{1}{2-\alpha} - \frac{\varphi}{2\pi} \notin \Z, \\[3ex]
      \displaystyle \bigg( \frac{1}{2-\alpha} - \frac{\varphi}{2\pi} \bigg)-1
      &\displaystyle \text{if  }
      \frac{1}{2-\alpha} - \frac{\varphi}{2\pi} \in \Z, \\
   \end{cases}
   \quad
   \label{eq:kappaMaxApp}
\end{equation}
   and $[\, \cdot \,]$ denotes the integer part. Moreover, the total number of arcs is
   given by
   \[
      k_{\text{tot}}(\alpha,\varphi) = k_{\text{max}}-k_{\text{min}}+1.
   \]
   \label{lemma:connectingArcs}
\end{lemma}
\noindent This lemma follows immediately from
Lemma \ref{lemma:barTerrVerz_Marchal}.
\begin{figure}[!ht]
   \centering
   \includegraphics[scale=0.4]{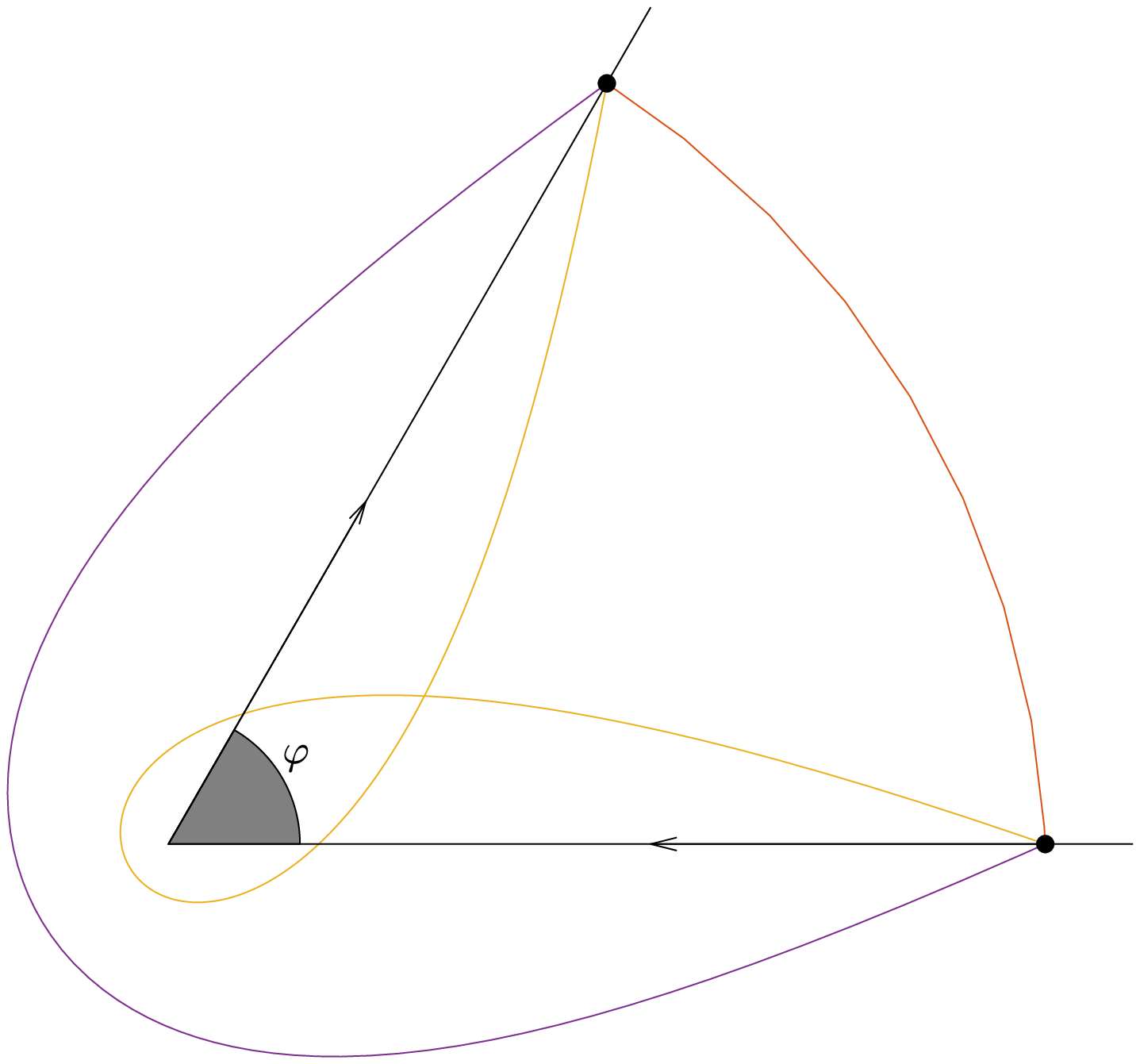}
   \hskip 0.8cm
   \includegraphics[scale=0.4]{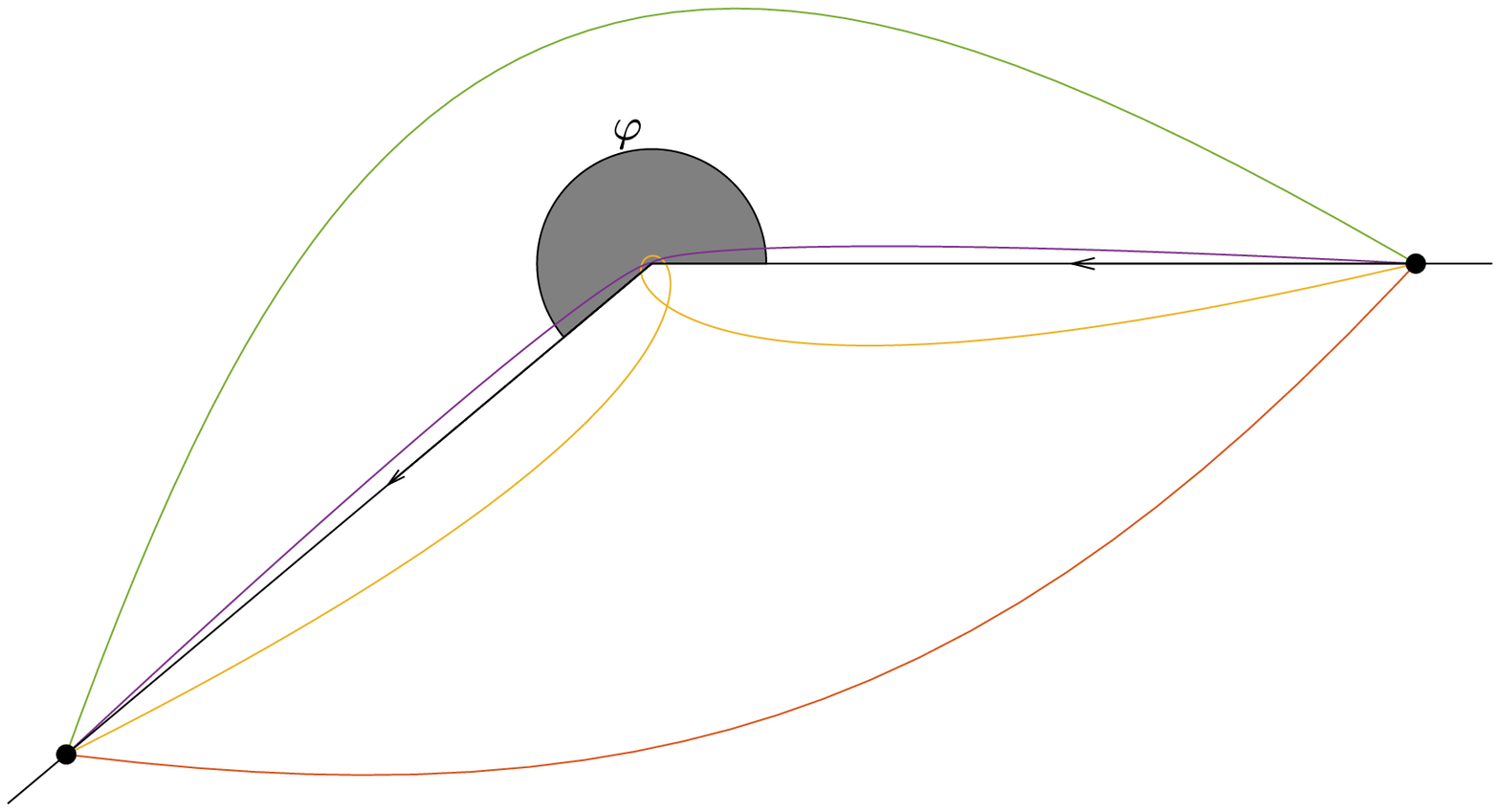}
   \caption{Connecting arcs computed for $\alpha=7/5$ and $\varphi = \pi/3$ (on the left),
   $\varphi=11\pi/9$ (on the right). Note that the total number of connecting arcs changes
with the angle $\varphi$ between the ending points.}
   \label{fig:connArcs}
\end{figure}

\section{The collision angle}
\label{app:collang}
Let us denote with $\Pi$ the union of the symmetry planes of the selected Platonic polyhedron, then 
\[
\R^3 \setminus \Pi = \bigcup_{\widetilde{R} \in \widetilde{\mathcal{R}}} \widetilde{R}D,
\]
where $D \subset \R^3$ is an open fundamental region for the action of
$\widetilde{\mathcal{R}}$ on $\R^3$, i.e. the cone over a triangle of the tessellation of
the sphere. 
We denote with $\mathcal{D} = \{ \widetilde{R}D \}_{\widetilde{R} \in
\widetilde{\mathcal{R}}}$ the set of all connected components of $\R^3 \setminus \Pi$, and we refer to one of its elements as a \emph{chamber}.
We also denote with $\mathcal{S}$ the set of the connected components of $\Pi \setminus \Gamma$, and we refer to one of its elements as a \emph{wall}. 
\begin{remark}
Note that there is a $1$-$1$ correspondence between the chambers in $\mathcal{D}$ and the triangles of the tessellation $\mathscr{T}_{\cal R}$, hence a free-homotopy class can be described also with a periodic sequence $\{ D_k \}_{k \in \Z}$.
\end{remark}
From Proposition 4.4 and Remark 4.4 of \cite{fgn10}, given a set of loops $\K$ represented 
by a periodic sequence of chambers $\{ D_k \}_{k \in \Z}$ and a minimizer $u_I^* \in \overline{\K}$ of $\A^\alpha$, 
there exists a minimizing sequence $\{ \hat{u}^\ell_I \}_{\ell \in \N} \subset \K$ uniformly converging 
to $u_I^*$. Moreover, up to reflections with respect to the symmetry planes in $\Pi$ (that do not change the value of the action), we can assume that for each $\hat{u}_I^\ell$ there exist sequences $\{ t^{\ell}_k \}_{k \in 
\Z} \subset \R$, $\{ S_k \}_{k \in \Z} \subset \mathcal{S}$, such that $t^{\ell}_k < t^{\ell}_{k+1}$, $S_k \subset \overline{D_k}$, $k\in\Z$, and
\begin{itemize}
    \item[(i)] $\hat{u}_I^\ell([t_k^{\ell}, t_{k+1}^{\ell}]) \subset \overline{D_k}\setminus \Gamma$,
    \item[(ii)] $\hat{u}_I^\ell(t_k^{\ell}) \in S_k$, 
    \item[(iii)] $S_{k+1}\neq S_k$.
\end{itemize}

Let us denote with $S^k \notin \{S_k, S_{k+1} \}$ the third wall of $D_k$, i.e. the wall that is not crossed 
by the trajectory of $\hat{u}_I^\ell$ in the interval $[t^\ell_{k-1}, t^\ell_{k+2}]$.
We introduce a {\em collision angle} following the steps of Proposition~5.8 in \cite{fgn10}. If $\minloop$ has a collision at time $t_c$, then we
necessarily have $\minloopgen(t_c)\in r_h$ for some $r_h$ in the
sequence $k\to r_k=(\overline{S_{k+1}}\cap
\overline{S^k})\setminus \{ 0 \}$. 
Let $\tilde{h}>h$ be defined by 
\begin{eqnarray*}
  \left\{\begin{array}{cl} r_h\subset
  \overline{S^h}\cap\overline{S^{\tilde h}}\\ h<k<\tilde h&\Rightarrow
  \overline{S^h}\cap\overline{S^k}=\{0\}.\end{array}\right.
\end{eqnarray*}
Note that the walls $S^h$ and $S^{\tilde h}$ 
both include the collision axis $r_h$, see Figure~\ref{fig:Dk_seq}.
\begin{figure}[ht]
    \centering
    \includegraphics[width=0.5\textwidth]{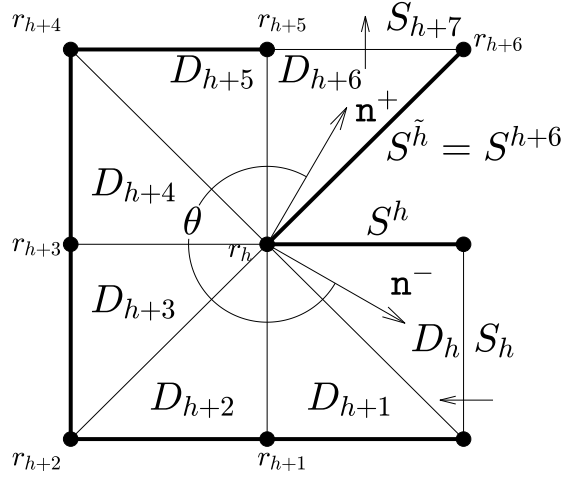}
    \caption{Sketch for the definition of the collision angle.}
    \label{fig:Dk_seq}
\end{figure}

Hereafter we denote with
$\mathfrak{p}^\pm$ the open half--planes parallel to $\npm$, with origin
the line of $r_h$.  In a similar way, we denote with
$\mathfrak{p}^\ell(t)$ the open half--plane passing through
$\hat{u}^\ell_I(t)$, with origin the line of $r_h$.

Up to passing to a
subsequence of $\{t_k^\ell\}_{\ell\in\Z}$, we can assume that $t_k^*=\lim_{\ell\rightarrow
  +\infty}t_k^\ell$ exists for all $k\in\Z$.
Let $i$, $j$ be respectively the largest and the smallest integers such that $t_i^*<t_c$ and $t_c<t_j^*$. We observe that
\begin{equation}
i+1\leq j\,,\qquad h\leq i\leq \tilde{h}\,,\qquad h<j\leq\tilde{h}+1\ .
\label{stime_per_ij}
\end{equation}
Let $\theta^\ell$ be the signed angle swept by $\mathfrak{p}^\ell(t)$ while $t$ increases from ${t}_i^\ell$ to ${t}_j^\ell$.  
We take as
positive the versus of rotation around $r_h$ defined by $D_k$ when $k$
increases from $h$ to $\tilde{h}$. 
We also introduce the signed angles $\theta^{\ell,+}$, $\theta^{\ell,-}$ swept 
to rotate respectively from $\mathfrak{p}^\ell({t}_j^\ell)$  to
$\mathfrak{p}^+$ and from $\mathfrak{p}^\ell({t}_i^\ell)$ to $\mathfrak{p}^-$.
%

\noindent We associate to the collision at time $t_c$ the collision angle $\theta$ (see Figure~\ref{fig:Dk_seq}) defined by
\begin{equation}
\theta:= \theta^\ell + \theta^{\ell,+} - \theta^{\ell,-} .
\label{thetaang}
\end{equation}
From 
the definitions of $\theta^\ell$, $\theta^{\ell,\pm}$ 
it follows that the collision angle $\theta$
is independent from $\ell$. 

\begin{proposition}
Assume that $\K$ is $\alpha$-simple. Then the collision angle $\theta$ satisfies the inequalities
\begin{equation}
-\frac{\pi}{\ord_h} \leq\theta\leq 2\bigg[\frac{1}{2-\alpha}\bigg]\pi .
\label{stimapertheta}
\end{equation}
\label{thetabounds}
\end{proposition}

\begin{proof}
By the definition \ref{def:simpCone} of $\alpha$-simple cones we have
\begin{eqnarray}
\tilde h-h + 1&\leq& 2\bigg[\frac{1}{2-\alpha}\bigg]\ord_h,
\label{simple}
\end{eqnarray}
where $\ord_h$ is the order of the pole $r_h\cap\unitsphere.$ 
%
We note that the definition of $\theta$
implies
\[
\theta\leq\frac{\pi}{\ord_h}(j-i).
\]
On the other hand we have, using (\ref{simple}), (\ref{stime_per_ij}), 
\[
\frac{\pi}{\ord_h}(j-i) \leq \frac{\pi}{\ord_h}(\tilde{h}-h+1) \leq \bigg[\frac{1}{2-\alpha}\bigg]2\pi.
\]
%
%
The inequality 
\[
\theta \geq -\frac{\pi}{\ord_h}
\]
can be proved by a similar argument.

\end{proof}

\section{Properties of solutions with partial collisions}
\label{app:collparz}

Here we list some properties of {\em ejection} solutions to
\eqref{eq:newtEqPartAlpha}.
We recall that an ejection solution $w(t)$ is such that
\[
\lim_{t\to t_c^+} w(t) = w(t_c)\in r\setminus\{0\},
\]
for some collision time $t_c\in\R$, and up to translations of time and
space\footnote{Note that when translating, also the term $V_1(w)$ in \eqref{eq:newtEqPartAlpha} needs to be
modified accordingly to obtain solutions of the differential equations, i.e. also the
center of attraction should be moved.} we can assume
\[
w(t_c)=0,\qquad t_c=0.
\]
Analogous statements apply to {\em collision} solutions of
\eqref{eq:newtEqPartAlpha}, that is solutions satisfying $\lim_{t\to
  0^-}w(t)=0$.
Let us denote by $\er$ a unit vector parallel to $r$.  The following
result generalizes Proposition~5.6 in \cite{fgn10} to the case of
$\alpha$-homogeneous potentials.  We omit the proof, that can be
derived from the results in \cite{FT2004}, \cite{barutello2008}.
\begin{proposition}
Let $w:(0,\bar t)\rightarrow \R^3$ be a maximal solution of
\eqref{eq:newtEqPartAlpha}.
Assume that
\begin{equation}
\lim_{t\rightarrow 0^+}w(t) = 0.
\label{eq:P2.14}
\end{equation}
Then \\
\noindent
{\rm ($i$)} there exist $\olda \in \R$ and a unit vector ${\mathsf n}$,
orthogonal to $r$, such that
\begin{equation}
\lim_{t\rightarrow 0^+}\frac{\dot w(t) + R_\pi\dot w(t)}{2} = \olda \er, 
\label{lim_pdot}
\end{equation}
\begin{equation}
\lim_{t\rightarrow 0^+}\frac{w(t) - R_\pi w(t)}{|w(t) - R_\pi w(t)|} =
\lim_{t\rightarrow 0^+}\frac{w(t)}{|w(t) |} = {\mathsf n}.
\label{n_vector}
\end{equation}

\noindent{\rm ($ii$)} The rescaled function $w^\lambda:[0,1]\to\R^3$
defined by $w^\lambda(0)=0, w^\lambda(\tau) = \lambda^{2/(2+\alpha)}
w(\tau/\lambda), \lambda>1/\bar t$, satisfies
\begin{equation}
\begin{array}{l}
\displaystyle\lim_{\lambda\to+\infty}w^\lambda(\tau) =
 s^\alpha(\tau){\mathsf n} \mbox{ uniformly in } [0,1],\cr
\displaystyle\lim_{\lambda\to+\infty} \dot{w}^\lambda(\tau) = 
\dot{s}^\alpha(\tau){\mathsf n} \mbox{ uniformly in } [\delta,1],
0<\delta<1,\cr
\end{array}
\label{wlambda_eqs}
\end{equation}
where
\[
s^\alpha(\tau) = \frac{(2+\alpha)^{2/(2+\alpha)}}{2}\calpha^{1/(2+\alpha)} \tau^{2/(2+\alpha)},
\tau\in[0,+\infty)
\]
is the parabolic ejection motion, that is the solution of
\[
\dot{s} = (\calpha/2^\alpha)^{1/2} s^{-\alpha/2}
\]
that satisfies $\lim_{\tau\to 0^+}s(\tau) = 0$.

\noindent{\rm ($iii$)} The following estimates hold for positive
constants $t_0$, $\rho_0$, $x_0$, $C_j$, $j=1,..,7$ that depend only
on $\olda$ and on the energy constant $h$.
\begin{equation}
\left\{
\begin{array}{ll}
C_1 t^{2/(2+\alpha)} \le \rho(t) \le C_2 t^{2/(2+\alpha)},&\cr 
C_3 t^{-\alpha/(2+\alpha)} \le \dot \rho(t) \le 
C_4 t^{-\alpha/(2+\alpha)},  &t\in (0, t_0], \cr
|\frac{dx}{d\rho} - 1| \le C_5 \rho^{1+\alpha},  
&\rho \in (0, \rho_0],\cr
|y'| \le C_6 x^{\alpha/2}, 
\qquad
|z'| \le  C_7 x^{1+\alpha},
&x\in(0,x_0],\cr
\end{array}
\right.
\label{est_collect}
\end{equation}
where $\rho = \frac{1}{2}|(R_\pi-I)w|$ and $x,y,z$ are the components
of $w$ on ${\mathsf n}$, $\er$, $\eort = \er\times{\mathsf n}$ and $'$
denotes differentiation with respect to $x$.
\label{prop:limcoll}
\end{proposition}


\begin{remark}
The estimates in $(iii)$ are valid for all the solutions of
\eqref{eq:newtEqPartAlpha},\eqref{eq:energyPartAlpha} that satisfy
(\ref{lim_pdot}), (\ref{n_vector}) for fixed $\olda$ and $h$.
This is essential for the analysis of partial collisions.
\end{remark}
%


The same result stated in Proposition 5.9 of \cite{fgn10} holds also in the
case of $\alpha$-homogeneous potentials. We recall the statement
below.
\begin{proposition}
Let $w_i:(0,\bar t_i)\rightarrow \R^3, \bar{t}_i>0, i=1,2$ be two maximal
solutions of \eqref{eq:newtEqPartAlpha}
such that
\[
\lim_{t\rightarrow 0^+}w_i(t) = 0.
\]
If $h_i$, $\olda_i$, ${\mathsf n}_i$  are the corresponding values
of the energy and the values of $\olda$ and ${\mathsf n}$ given
by Proposition~\ref{prop:limcoll}, then
\[
\left\{
\begin{array}{l}
h_1 = h_2\cr
\olda_1 = \olda_2\cr
{\mathsf n}_1 = {\mathsf n}_2\cr
\end{array}
\right.
\hskip 1cm\Longrightarrow
\hskip 1cm
\left\{
\begin{array}{l}
\bar{t}_1 = \bar{t}_2\cr
w_1 = w_2\cr
\end{array}
\right.
.
\]
\label{prop:unique}
\end{proposition}
\begin{proof}
  The proof follows exactly the same steps as in the case
  $\alpha=1$. Here we recall the main points.  Projecting the equation
  of motion \eqref{eq:newtEqPartAlpha} onto the basis ${\mathsf n},
  \er, \eort$ and setting
  \[
  w = x\mathsf{n} + y\er + z\eort\] we get
\begin{equation}
\left\{
\begin{array}{l}
\ddot x = -\displaystyle\frac{\alpha\calpha}{(2x)^{1+\alpha}
(1+\frac{z^2}{x^2})^{\frac{2+\alpha}{2}}} +
V_1 \cdot {\mathsf n}\cr
\ddot y = V_1 \cdot \er\cr
\ddot z = -\displaystyle\frac{\alpha\calpha z}{2^{1+\alpha} x^{2+\alpha}
(1+ \frac{z^2}{x^2})^{\frac{2+\alpha}{2}}} +
V_1 \cdot\eort\cr
\end{array}
\right.
.
\label{xyzddot}
\end{equation}
We take $x$ as independent variable and write the energy equation
\eqref{eq:energyPartAlpha} as
\begin{equation}
\label{eq:P2.53}
\dot x^2(1 + |y'|^2 + |z'|^2) = \frac{\calpha}{{(2x)}^\alpha
(1+\frac{z^2}{x^2})^{\frac{\alpha}{2}}} + V + h,
\end{equation}
where $'$ denotes differentiation with respect to $x$.
Setting
\begin{equation}
x = e^s, \hskip 1cm s \in(-\infty, s_0],
  \label{xes}
  \end{equation}
where $s_0 < 0$ is chosen later, and introducing the variables
\[
\eta = \frac{d y}{ds}, \qquad \zeta = \frac{dz}{ds}
\]
we can write the first order system
%
\begin{equation}
\left\{
\begin{array}{ll}
\displaystyle\frac{dy}{ds} = \eta,
\hskip 0.5cm
&\displaystyle\frac{d\eta}{ds}= \Bigl(\frac{2+\alpha}{2} + 
\frac{\alpha}{2}\mathcal U\Bigr)
\eta + e^{(2+\alpha)s}\mathcal A \cr &\cr
\displaystyle\frac{dz}{ds} = \zeta,
\hskip 0.5cm
&\displaystyle\frac{d\zeta}{ds}= \Bigl(\frac{2+\alpha}{2} + 
\frac{\alpha}{2}\mathcal U\Bigr) \zeta - 
\Bigl(\frac{\alpha}{2} + \frac{\alpha}{2}\mathcal V\Bigr)z +
e^{(2+\alpha)s}\mathcal B\cr
\end{array}
\right.,
\label{eq:P2.58}
\end{equation}
where $\mathcal{U}, \mathcal{A}, \mathcal{V}, \mathcal{B}$ are defined
by relations
\[\begin{split}
&1+\mathcal{U} = \frac{(1 + |y'|^2 + |z'|^2)}
{(1+\frac{z^2}{x^2})^{\frac{2+\alpha}{2}}} 
\frac{\left(1 - (2x)^{1+\alpha} (1+ \frac{z^2}{x^2})^{\frac{2+\alpha}{2}} 
\frac{V_1\cdot \mathsf{n}}{\alpha\calpha}\right)}
{\left(1 + {(2x)}^\alpha (1+\frac{z^2}{x^2})^{\frac{\alpha}{2}}
  \frac{V + h}{\calpha}\right)},\cr
&1+\mathcal{W} = \frac{(1 + |y'|^2 + |z'|^2)
( 1+ \frac{z^2}{x^2})^{\frac{\alpha}{2}}}
{1 + {(2x)}^\alpha(1+\frac{z^2}{x^2})^{\frac{\alpha}{2}}\frac{V+h}{\calpha}},\cr
&\mathcal{A} = \frac{2^\alpha}{\calpha} V_1\cdot\er(1+\mathcal{W}),\cr
&1 + \mathcal{V} = \frac{1}{( 1+
    \frac{z^2}{x^2})^{\frac{2+\alpha}{2}}}(1 + \mathcal{W}),\cr
&\mathcal{B} = \frac{2^\alpha}{\calpha} V_1\cdot\eort(1+\mathcal{W}).\cr
\end{split}
\]
System~\ref{eq:P2.58} can be written in compact form as
\begin{equation}
\frac{d \gamma}{ds} = M \gamma + \mathcal{N}(\gamma,s)
\label{dgammads}
\end{equation}
where $\gamma = (y, z, \eta, \zeta)^T$,
\[
\mathcal{N}(\gamma,s) =\Bigl(0,0,
\frac{\alpha}{2}\mathcal U\eta + e^{(2+\alpha)s}\mathcal A, 
\frac{\alpha}{2}\mathcal U \zeta -
\frac{\alpha}{2}\mathcal V z + e^{(2+\alpha)s}\mathcal B\Bigr)^T
\]
and $M$ is the constant matrix
\[
M=
\left[
\begin{array}{cccc}
0&0&1&0\\
0&0&0&1\\
0&0& \frac{2+\alpha}{2}&0\\
0&-\frac{\alpha}{2}&0& \frac{2+\alpha}{2}
\end{array}
\right].
\]
To each solution $w$ of
\eqref{eq:newtEqPartAlpha}
satisfying (\ref{eq:P2.14}) there corresponds a solution $\gamma_w$ of
(\ref{dgammads}) and this correspondence is 1--1.
%
Moreover, from the estimates (\ref{est_collect}), we can find a
constant $C_0>0$, depending only on $h, \olda$, such that
\begin{equation}
|\gamma_w(s)| \leq \Cquad e^{\frac{2+\alpha}{2}s},\qquad s\in(-\infty,s_0].
\label{stima_gammaw}
\end{equation}
%
Computing explicitly the eigenvalues
\[
\lambda_1 = 0,\qquad\lambda_2 = \frac{\alpha}{2},\qquad\lambda_3 = 1,\qquad
\lambda_4 = \frac{2+\alpha}{2},
\]
and the eigenvectors
\[
\begin{array}{llll}
  \rho_1 =(1,0,0,0)^T,\qquad
  &\rho_2 =(0,1,0,\frac{\alpha}{2})^T,\qquad
  &\rho_3 =(0,1,0,1)^T,\qquad
  &\rho_4 =(1,0,\frac{2+\alpha}{2},0)^T, \cr
\end{array}
\]
of the matrix $M$, we find that
there exists a constant $C_1>0$ such that
\begin{equation}
|e^{Ms}| \le C_1 e^{\frac{2+\alpha}{2}s},\quad s\in[0,+\infty).   
\label{stima_eMs}
\end{equation}
Let us consider a solution of the homogeneous equation
$\frac{d\gamma}{ds} = M \gamma$ of the form
\begin{equation}
\label{eq:P2.63}
\gamma_\delta(s) = 
e^{\frac{2+\alpha}{2}s} \delta \rho_4,\hskip 1cm \delta\in\R.
\end{equation}
Given $K > 0$ and $c\in (0, \frac{\alpha}{2}]$,
consider the complete metric space of the continuous maps 
\begin{equation}
X =\{ \gamma :(-\infty, s_0]\rightarrow \R^4 : |(\gamma -\gamma_\delta)(s)| \le
K e^{(1 + c)s}\},
\label{eq:P2.64}
\end{equation}
endowed with the distance
\begin{equation}
\label{eq:P2.64.1}
d(\gamma, \tilde \gamma)=\max_{s\in(-\infty, s_0]}|\gamma(s) -
\tilde\gamma(s)| e^{-s}.
\end{equation}
For each fixed $\delta$ and for $K$ large enough we have
\[
\gamma_w\in X
\]
for all solutions $w$ of \eqref{eq:newtEqPartAlpha},
(\ref{eq:P2.14}) corresponding to given
values of $h, \olda, \mathsf{n}$. Moreover, solutions of
\eqref{eq:newtEqPartAlpha}
correspond
to continuous solutions $\gamma :(-\infty,
s_0]\rightarrow \R^4$ of the nonlinear integral equation
\begin{equation}
\gamma(s)=\gamma_\delta(s) + \int_{-\infty}^s e^{M(s - r)}{\mathcal
N}(\gamma(r),r)dr.
\label{integral_eq}
\end{equation}
We can show that the map
\begin{equation}
\label{eq:P2.66}
(T\gamma)(s)=\gamma_\delta(s) + \int_{-\infty}^s e^{M(s - r)}{\mathcal
N}(\gamma(r),r)dr
\end{equation}
defines a contraction on $X$ for $-s_0> 0$ sufficiently large,
implying that equation (\ref{integral_eq}) has a unique solution for
each $\delta \in \R$.  Moreover, the choice of $\delta$ is uniquely
determined by the value of $\olda$ in (\ref{lim_pdot}).

\noindent From the estimates 
\begin{equation}
\mathcal U, \mathcal V, \mathcal W = O(e^{\alpha s}),\qquad 
\mathcal A, \mathcal B = O(1)
\hskip 0.5cm s\in(-\infty,s_0]
\label{UVAB_est}
\end{equation}
and, for the gradients,
\begin{equation}
\mathcal {U}_\gamma, \mathcal {V}_\gamma, 
\mathcal {A}_\gamma, \mathcal {B}_\gamma = O(e^{-\frac{(2-\alpha)}{2}s}),
\qquad s\in(-\infty,s_0],
\label{grad_est}
\end{equation}
%
%
we get
\begin{equation}
  |\mathcal N(\gamma(s),s)| \le Ce^{\frac{2+3\alpha}{2}s}
  \qquad
|\mathcal N(\gamma(s),s) -{\mathcal N}(\tilde{\gamma}(s),s)|  
\le Ce^{(1+\alpha)s}d(\gamma, \tilde \gamma)
\label{N_DeltaN}
\end{equation}
and, by \eqref{stima_eMs},
\begin{equation}
|(T\gamma)(s) - \gamma_\delta(s)|  \le Ce^{\frac{2+3\alpha}{2}s},
\quad s\in(-\infty,s_0],\label{eq:P2.72}
\end{equation}
  %
and
\begin{equation}
d(T\gamma, T\tilde\gamma)\le Ce^{\alpha s_0} d(\gamma, \tilde \gamma), \qquad \forall \gamma, \tilde{\gamma}\in X .
\label{eq:P2.73}
\end{equation}
Relation \eqref{eq:P2.73} shows that the map $T: X \rightarrow X$ is a
contraction, provided $-s_0>0$ is sufficiently large.

If $\gamma=(y,z,\eta,\zeta)^T$ is the fixed point of $T$,
(\ref{eq:P2.72}) implies that
\begin{equation}
\label{eq:P2.75}
\lim_{s\rightarrow - \infty}|\gamma(s) - \gamma_\delta(s)|
e^{-\frac{2+\alpha}{2}s} = 0.
\end{equation}
Proposition~\ref{prop:limcoll} and the variable change \eqref{xes}
imply the asymptotic estimates

\begin{equation}
  t \propto
  \frac{2}{2+\alpha} \sqrt{\frac{2^\alpha}{\calpha}}e^{\frac{2+\alpha}{2}s},
  \qquad
\frac{dt}{dx} \propto
\sqrt{\frac{2^\alpha}{\calpha}} e^{\frac{\alpha}{2}s}.
\label{t_dtdx}
\end{equation}
From these asymptotic formulas and \eqref{eq:P2.75}
%
%
it follows that
\[
\delta = \frac{2}{2+\alpha} \sqrt{\frac{2^\alpha}{\calpha}}\olda.
\]

\end{proof}

\section*{Acknowledgments}
We wish to thank S. Terracini for her useful suggestions and comments. We also wish to thank the anonymous referee 
for his/her careful reading of the manuscript, and for the several comments which really helped us to 
improve the final version of the manuscript.
Both authors acknowledge the project MIUR-PRIN 20178CJA2B ``New
frontiers of Celestial Mechanics: theory and applications'' and the GNFM-INdAM (Gruppo Nazionale
per la Fisica Matematica).
This work was partially supported through the H2020 MSCA ETN Stardust-Reloaded, grant agreement number 813644.
\bibliography{mybib}{} 
\bibliographystyle{plain}
\end{document}